\documentclass[sigconf,cleveref, autoref, thm-restate]{acmart}

\acmConference[PPDP 2023]{International Symposium on Principles and Practice of Declarative Programming}{October 22--23, 2023}{Lisboa, Portugal}

\setcopyright{none}

\usepackage{booktabs}   
\usepackage{subcaption} 

\usepackage{ifthen}     
\usepackage{scalerel}   
\usepackage{thm-restate}
\usepackage{multicol}
\usepackage{adjustbox}

\usepackage{cleveref}
\crefname{lstlisting}{listing}{listings}
\Crefname{lstlisting}{Listing}{Listings}

\usepackage{mathpartir}
\usepackage{listings}
\usepackage{fancybox}
\usepackage{xspace}
\makeatletter
\DeclareFontFamily{OMX}{MnSymbolE}{}
\DeclareSymbolFont{MnLargeSymbols}{OMX}{MnSymbolE}{m}{n}
\SetSymbolFont{MnLargeSymbols}{bold}{OMX}{MnSymbolE}{b}{n}
\DeclareFontShape{OMX}{MnSymbolE}{m}{n}{
    <-6>  MnSymbolE5
   <6-7>  MnSymbolE6
   <7-8>  MnSymbolE7
   <8-9>  MnSymbolE8
   <9-10> MnSymbolE9
  <10-12> MnSymbolE10
  <12->   MnSymbolE12
}{}
\DeclareFontShape{OMX}{MnSymbolE}{b}{n}{
    <-6>  MnSymbolE-Bold5
   <6-7>  MnSymbolE-Bold6
   <7-8>  MnSymbolE-Bold7
   <8-9>  MnSymbolE-Bold8
   <9-10> MnSymbolE-Bold9
  <10-12> MnSymbolE-Bold10
  <12->   MnSymbolE-Bold12
}{}

\let\llangle\@undefined
\let\rrangle\@undefined
\DeclareMathDelimiter{\llangle}{\mathopen}%
                     {MnLargeSymbols}{'164}{MnLargeSymbols}{'164}
\DeclareMathDelimiter{\rrangle}{\mathclose}%
                     {MnLargeSymbols}{'171}{MnLargeSymbols}{'171}
\makeatother

\usepackage{stmaryrd} 
\newcommand*\NewstuffPadding{\columnwidth - 2mm}%
\newenvironment{Newstuff}%
{\vspace{\baselineskip - 2mm}\leavevmode\newline\begin{Sbox}%
\setlength{\abovedisplayskip}{0pt}%
\setlength{\belowdisplayskip}{0pt}%
    \begin{minipage}%
    }%
  {\end{minipage}\end{Sbox}\fbox{\TheSbox}\vspace{\baselineskip - 2mm}\leavevmode\newline}
\definecolor{newcolor}{gray}{0.95}
\newcommand\Keyword[1]{\texttt{\textbf{#1}}}
\newcommand\Syntax\texttt
\newcommand\Meta\textit

\newcommand\PolyVGR{PolyVGR\xspace}

\newcommand\semicolon{;\,}

\newcommand\Nchan{\alpha}            
\newcommand\Nchanb\beta

\newcommand\Int{\Terminal{Int}}
\newcommand\Inp[1]{{\Syntax{?}#1}.}
\newcommand\Outp[1]{{\Syntax{!}#1}.}
\newcommand\Chan{\Keyword{Chan}~}
\newcommand\Unit{\Keyword{Unit}}
\newcommand\Ap[1]{{[#1]}} 

\newcommand\KWReceive{\Syntax{receive}}
\newcommand\KWAccept{\Syntax{accept}}
\newcommand\KWSend{\Syntax{send}}
\newcommand\KWNew{\Syntax{new}}

\newcommand\Send[1]{\KWSend\,#1\,\Syntax{on}\,}
\newcommand\Receive{\KWReceive\,}
\newcommand\Accept{\KWAccept\,}
\newcommand\Request{\Syntax{request}\,}

\newcommand\New{\KWNew\,}

\newcommand\Hole\Box

\newcommand\ThreadOpen\langle
\newcommand\ThreadClose\rangle

\newcommand\Fresh{\textup{\textrm{fresh}}\,}
\newcommand\Dual\overline

\newcommand\Dom[1]{\textup{\textrm{dom}} (#1)}
\newcommand\ReduceTo\rightarrow

\newcommand\RLabel\ell          
\newcommand\ELabel\tau          

\newcommand\Pcong\equiv          

\newcommand\FSTequiv\equiv

\newcommand*\Terminal[1]{\operatorname{\mathsf{#1}}}

\newcommand*\Label{\ell}
\newcommand*\LabelA{1}
\newcommand*\LabelB{2}

\newcommand*\Kind{K}
\newcommand*\KType{\Terminal{Type}}
\newcommand*\KSt{\Terminal{State}}
\newcommand*\KSession{\Terminal{Session}}
\newcommand*\KShape{\Terminal{Shape}}
\newcommand*\KDomain[1]{\Terminal{Dom}({#1})}
\newcommand*\KTypeFun[1]{\KDomain{#1} \to \KType}
\newcommand*\KStFun[1]{\KDomain{#1} \to \KSt}

\newcommand*\Cfg{C}
\newcommand*\CExp[1]{#1}
\newcommand*\CPar[2]{{#1}\operatorname{\|}{#2}}
\newcommand*\CBindAP[2]{\nu{#1}\colon\TAccessPoint{#2}.\ }
\newcommand*\CBindChan[3]{\nu\StBind{{#1},{#2}}{#3}.\ }

\newcommand*\ELet[2]{\Terminal{let}{#1}={#2}\Terminal{in}}

\newcommand*\EChan{\Terminal{chan}}
\newcommand*\EFork{\Terminal{fork}}
\newcommand*\Exp{e}
\newcommand*\EApp[2]{{#1}\ {#2}}
\newcommand*\ENew{\Terminal{new}}
\newcommand*\ERequest{\Terminal{request}}
\newcommand*\EAccept{\Terminal{accept}}
\newcommand*\ESend[2]{\Terminal{send}{#1}\Terminal{on}{#2}}
\newcommand*\ERecv{\Terminal{receive}}
\newcommand*\ESelect[1]{\Terminal{select}{#1}\Terminal{on}}
\newcommand*\ECase[3]{\Terminal{case}{#1}\Terminal{of}\{{#2}\semicolon{#3}\}}
\newcommand*\EMatch[3]{\Terminal{match}{#1}\Terminal{of}\{{#2}\semicolon{#3}\}}

\newcommand*\EClose{\Terminal{close}}
\newcommand*\EProj[2]{\mathop{\pi_{#1}}{#2}}
\newcommand*\EInj[2]{\mathop{\Terminal{inj}_{#1}}{#2}}
\newcommand*\VInj[2]{\mathop{\Terminal{vinj}_{#1}}{#2}}
\newcommand*\ETApp[2]{{#1}[{#2}]}

\newcommand*\THL[1]{\textcolor{blue}{#1}}

\newcommand*\Val{v}
\newcommand*\EAbs[3]{\lambda(\THL{#1}\semicolon{#2}\colon{#3}).}
\newcommand*\ETAbs[3]{
  \ifthenelse { \equal {#3} {\CstrTrue} }
    { \Lambda({#1}\colon {#2}). }
    {\Lambda({#1}\colon {#2}).\ {#3}\Rightarrow}
}

\newcommand*\EUnit{\Terminal{unit}}
\newcommand*\EVar{\EVarX}
\newcommand*\EVarX{x}
\newcommand*\EVarY{y}

\newcommand*\EPair[2]{({#1},{#2})}

\newcommand*\Ses{S}
\newcommand*\SRecv[4]{\operatorname?(\exists{#1}\colon{#2}.{#3}\semicolon{#4}).}
\newcommand*\SSend[4]{\operatorname!(\exists{#1}\colon{#2}.{#3}\semicolon{#4}).}
\newcommand*\SSendBegin[2]{\operatorname!(\exists{#1}\colon{#2}.}
\newcommand*\SSendMid[1]{{#1}\semicolon}
\newcommand*\SSendEnd[1]{{#1}).}
\newcommand*\SEnd{\Terminal{End}}
\newcommand*\SChoice[2]{{#1}\oplus{#2}}
\newcommand*\SBranch[2]{{#1}\mathbin{\&}{#2}}
\newcommand*\SVar{\sigma}

\newcommand*\Typ{T}
\newcommand*\TVar{\alpha}
\newcommand*\TVara{\alpha}
\newcommand*\TVarb{\beta}
\newcommand*\TVarc{\gamma}
\newcommand*\TAll[3]{%
  \ifthenelse { \equal {#3} {\CstrTrue} }
  { \forall({#1}\colon {#2}). }
  { \forall({#1}\colon {#2}).\ {#3}\Rightarrow }
}
\newcommand*\TExists[1]{%
  \ifthenelse{ \equal{#1}{\Empty}}
  {}
  {\exists{#1}.}
}
\newcommand*\TArr[5]{(\THL{#1}\semicolon {#2} \rightarrow \TExists {#3} \THL{#4}\semicolon {#5})}

\newcommand*\TArrBeginX[4]{(\THL{#1}\semicolon {#2} \rightarrow \TExists {#3} \THL{#4}\semicolon}
\newcommand*\TArrEndX{)}

\newcommand*\TArrBegin[2]{(\THL{#1}\semicolon {#2} \rightarrow}
\newcommand*\TArrEnd[3]{ \TExists {#1} \THL{#2}\semicolon {#3})}
\newcommand*\TChan[1]{\Terminal{Chan}{#1}}
\newcommand*\TPair[2]{{#1}\times{#2}}
\newcommand*\TSum[2]{{#1}+{#2}}
\newcommand*\TUnit{\Terminal{Unit}}
\newcommand*\TAccessPoint[1]{[{#1}]}

\newcommand*\Cstr{\mathbb{C}}

\newcommand*\CstrTrue{\top}
\DeclareRobustCommand*{\Disjoint}{
  \hspace{0.6mm}\raisebox{0.8pt}{\scaleobj{0.7}\#}\hspace{0.6mm}
}
\DeclareRobustCommand*{\DisjointAppend}{
  \operatorname{,_{\scaleobj{0.7}\#}}
}

\newcommand*\POLYShape{N}
\newcommand*\POLYDom{D}
\newcommand*\TDomFrom[2]{\operatorname{\varphi_{#1}} {#2}}
\newcommand*\TDomProj[2]{\operatorname{\pi_{#1}} {#2}}
\newcommand*\TDomEither[2]{{#1}+{#2}}
\newcommand*\TDomZero{\ensuremath{*}}
\newcommand*\TShapeSum[2]{{#1}+{#2}}
\newcommand*\TShapePair[2]{{#1}\fatsemi{#2}}
\newcommand*\TShapeOne{\ensuremath{\mathbb{X}}}
\newcommand*\TShapeZero {\ensuremath{\mathbb{I}}}

\newcommand*\TFun{\hat\Typ}
\newcommand*\TLam[2]{\lambda ({#1}\colon \KDomain{#2}).}
\newcommand*\StFun{\hat\St}

\newcommand*\St{\Sigma}
\newcommand*\Ctx{\Gamma}

\newcommand*\StEmpty{{}\Empty{}}
\newcommand*\Empty{\cdot}
\newcommand*\StBind[2]{{#1}\mapsto{#2}}
\newcommand*\CtxRestrictNonDom[1]{\lfloor{#1}\rfloor}
\newcommand*\CtxRestrictOnlyDom[1]{\lceil{#1}\rceil}

\newcommand*\TVarShape{n}

\newcommand*\TVarStFun{\hat\Sigma}
\newcommand*\TVarTFun{\hat T}
\newcommand*\TVarChan{c}

\newcommand*\ECtxSym{\mathcal E}
\newcommand*\ECtx[1]{\ECtxSym[{#1}]}
\newcommand*\ECtxA[1]{\ECtxSym_1[{#1}]}
\newcommand*\ECtxB[1]{\ECtxSym_2[{#1}]}
\newcommand*\CCtxSym{\mathcal C}
\newcommand*\CCtx[1]{\CCtxSym[{#1}]}
\newcommand*\CCtxA[1]{\CCtxSym_1[{#1}]}
\newcommand*\CCtxB[1]{\CCtxSym_2[{#1}]}
\newcommand*\CCtxBegin{\CCtxSym[}
\newcommand*\CCtxEnd{]}
\newcommand*\Cong{\equiv}
\newcommand*\ReducesToE{\hookrightarrow_\Exp}
\newcommand*\ReducesToC{\hookrightarrow_\Cfg}
\newcommand*\Subst[2]{\{{#1}/{#2}\}}

\newcommand*\IndCase[1]{\emph{Case} \textsc{#1}.\ }

\newcommand*\OPE[2]{{#1}\overset{\mathrm{OPE}}\Rightarrow{#2}}

\newcommand*\IsOuter[1]{\mathop{\mathrm{Outer}} {#1}}
\newcommand*\IsValue[1]{\mathop{\mathrm{Value}} {#1}}
\newcommand*\IsComm[1]{\mathop{\mathrm{Comm}} {#1}}
\newcommand*\IsFinal[1]{\mathop{\mathrm{Final}} {#1}}
\newcommand*\IsDeadlock[1]{\mathop{\mathrm{Deadlock}} {#1}}

\newcommand\JVGRValue[3]{
  #1; #2 \mapsto #3
}
\newcommand\JVGRExpr[6]{
  #1; #2; #3 \mapsto #4; #5; #6
}

\newcommand\ruleVGRReceiveD{
  \inferrule[C-ReceiveD]{
    \JVGRValue\Gamma v {\Chan\alpha}
  }{
    \JVGRExpr\Gamma{\Sigma, \alpha: \Inp{D}S}
    {\Receive v}
    \Sigma D {\alpha : S}
  }
}

\newcommand\ruleVGRReceiveS{
  \inferrule[C-ReceiveS]{
    \JVGRValue\Gamma v {\Chan\alpha} \\
    \Fresh d
  }{
    \JVGRExpr\Gamma{\Sigma, \alpha: \Inp{S'}S}
    {\Receive v}
    \Sigma {\Chan d} {d : S', \alpha : S}
  }
}

\newcommand\ruleVGRSendD{
  \inferrule[C-SendD]{
    \JVGRValue \Gamma v {D} \\
    \JVGRValue \Gamma {v'}{\Chan\alpha}
  }{
    \JVGRExpr \Gamma{\Sigma, \alpha: \Outp{D}{S}}
    {\Send v{v'}}
    \Sigma {\Unit}{\alpha:S}
  }
}

\newcommand\ruleVGRSendS{
  \inferrule[C-SendS]{
    \JVGRValue \Gamma v {\Chan\beta} \\
    \JVGRValue \Gamma {v'}{\Chan\alpha}
  }{
    \JVGRExpr \Gamma{\Sigma, \alpha: \Outp{S'}{S}, \beta: S'}
    {\Send v{v'}}
    \Sigma {\Unit}{\alpha:S}
  }
}

\newcommand\ruleVGRAccept{
  \inferrule[C-Accept]{
    \JVGRValue\Gamma v{[S]} \\
    \Fresh c
  }{
    \JVGRExpr\Gamma\Sigma
    {\Accept v}
    \Sigma{\Chan c}{\{c:S\}}
  }
}

\newcommand\ruleVGRRequest{
  \inferrule[C-Request]{
    \JVGRValue\Gamma v{[S]} \\
    \Fresh c
  }{
    \JVGRExpr\Gamma\Sigma
    {\Request v}
    \Sigma{\Chan c}{\{c:\Dual S\}}
  }
}

\newcommand*\PVGRIsKind[2]{{#1}\vdash{#2}}
\newcommand*\PVGRIsCtx[1]{\vdash{#1}}
\newcommand*\PVGRCstrEntail[2]{{#1}\vdash{#2}}
\newcommand*\PVGRHasKind[3]{{#1}\vdash{#2}:{#3}}
\newcommand*\PVGRHasExpType[6]{{#1};{#2}\vdash{#3}:\exists{#4}.{#5};{#6}}
\newcommand*\PVGRHasValType[3]{{#1}\vdash{#2}:{#3}}
\newcommand*\PVGRHasConfType[3]{{#1};{#2}\vdash{#3}}
\newcommand*\PVGREvalType[2]{{#1}\Downarrow{#2}}
\newcommand*\PVGRTypeConv[2]{{#1}\Cong{#2}}

\newcommand*\PVGRHasExpTypeBegin[3]{{#1};{#2}\vdash{#3}:}
\newcommand*\PVGRHasExpTypeEnd[3]{\exists{#1}.{#2};{#3}}

\newcommand\rulePVGRIsKindType{
  \inferrule[KF-Type]{}{
    \PVGRIsKind\Ctx\KType
  }
}
\newcommand\rulePVGRIsKindSession{
  \inferrule[KF-Session]{}{
    \PVGRIsKind\Ctx\KSession
  }
}
\newcommand\rulePVGRIsKindState{
  \inferrule[KF-State]{}{
    \PVGRIsKind\Ctx\KSt
  }
}
\newcommand\rulePVGRIsKindShape{
  \inferrule[KF-Shape]{}{
    \PVGRIsKind\Ctx\KShape
  }
}
\newcommand\rulePVGRIsKindDom{
  \inferrule[KF-Dom]{
    \PVGRHasKind\Ctx \POLYShape \KShape
  }{
    \PVGRIsKind\Ctx{\KDomain \POLYShape}
  }
}
\newcommand\rulePVGRIsKindArr{
  \inferrule[KF-Arr]{
    \PVGRIsKind{\Ctx}{\Kind_1} \\
    \PVGRIsKind{\Ctx}{\Kind_2}
  }{
    \PVGRIsKind{\Ctx}{\Kind_1 \to \Kind_2}
  }
}

\newcommand\rulePVGRIsCtxEmpty{
  \inferrule[CF-Empty]{}{
    \PVGRIsCtx\Empty
  }
}
\newcommand\rulePVGRIsCtxConsKind{
  \inferrule[CF-ConsKind]{
    \PVGRIsCtx\Ctx \\
    \PVGRIsKind\Ctx\Kind \\
    \TVar \not\in \Dom\Ctx
  }{
    \PVGRIsCtx{\Ctx,\TVar:\Kind}
  }
}
\newcommand\rulePVGRIsCtxConsType{
  \inferrule[CF-ConsType]{
    \PVGRIsCtx\Ctx \\
    \PVGRHasKind\Ctx\Typ\KType \\
    \EVar \not\in \Dom\Ctx
  }{
    \PVGRIsCtx{\Ctx,\EVar:\Typ}
  }
}
\newcommand\rulePVGRIsCtxConsCstr{
  \inferrule[CF-ConsCstr]{
    \PVGRIsCtx\Ctx \\
    \PVGRHasKind\Ctx{\POLYDom_1}{\KDomain{\POLYShape_1}} \\
    \PVGRHasKind\Ctx{\POLYDom_2}{\KDomain{\POLYShape_2}}
  }{
    \PVGRIsCtx{\Ctx, \POLYDom_1 \Disjoint \POLYDom_2}
  }
}

\newcommand\rulePVGRCstrEntailAxiom{
  \inferrule[CE-Axiom]{}{
    \PVGRCstrEntail{\Ctx,\POLYDom_1 \Disjoint \POLYDom_2}{\POLYDom_1 \Disjoint \POLYDom_2}
  }
}
\newcommand\rulePVGRCstrEntailZero{
  \inferrule[CE-Empty]{}{
    \PVGRCstrEntail{\Ctx}{\POLYDom \Disjoint \TDomZero}
  }
}
\newcommand\rulePVGRCstrEntailSym{
  \inferrule[CE-Sym]{
    \PVGRCstrEntail\Ctx{\POLYDom_2 \Disjoint \POLYDom_1}
  }{
    \PVGRCstrEntail\Ctx{\POLYDom_1 \Disjoint \POLYDom_2}
  }
}
\newcommand\rulePVGRCstrEntailSplit{
  \inferrule[CE-Split]{
    \PVGRCstrEntail\Ctx{\POLYDom \Disjoint (\POLYDom_1,\POLYDom_2)}
  }{
    \PVGRCstrEntail\Ctx{\POLYDom \Disjoint \POLYDom_1} \\
    \PVGRCstrEntail\Ctx{\POLYDom \Disjoint \POLYDom_2}
  }
}
\newcommand\rulePVGRCstrEntailMerge{
  \inferrule[CE-Merge]{
    \PVGRCstrEntail\Ctx{\POLYDom \Disjoint \POLYDom_1} \\
    \PVGRCstrEntail\Ctx{\POLYDom \Disjoint \POLYDom_2}
  }{
    \PVGRCstrEntail\Ctx{\POLYDom \Disjoint (\POLYDom_1,\POLYDom_2)}
  }
}
\newcommand\rulePVGRCstrEntailEmpty{
  \inferrule[CE-Empty]{}{
    \PVGRCstrEntail\Ctx\Empty
  }
}
\newcommand\rulePVGRCstrEntailCons{
  \inferrule[CE-Cons]{
    \PVGRCstrEntail\Ctx{\Cstr} \\
    \PVGRCstrEntail\Ctx{\POLYDom_1 \Disjoint \POLYDom_2}
  }{
    \PVGRCstrEntail\Ctx{\Cstr, \POLYDom_1 \Disjoint \POLYDom_2}
  }
}
\newcommand\rulePVGRCstrEntailProjMerge{
  \inferrule[CE-ProjMerge]{
    \PVGRCstrEntail\Ctx{\POLYDom_1 \Disjoint \TDomProj 1 {\POLYDom_2}} \\
    \PVGRCstrEntail\Ctx{\POLYDom_1 \Disjoint \TDomProj 2 {\POLYDom_2}}
  }{
    \PVGRCstrEntail\Ctx{\POLYDom_1 \Disjoint \POLYDom_2}
  }
}
\newcommand\rulePVGRCstrEntailProjSplit{
  \inferrule[CE-ProjSplit]{
    \PVGRCstrEntail\Ctx{\POLYDom_1 \Disjoint \POLYDom_2} \\
  }{
    \PVGRCstrEntail\Ctx{\POLYDom_1\Disjoint \TDomProj\Label{\POLYDom_2}}
  }
}

\newcommand\rulePVGRKindingVar{
  \inferrule[K-Var]{}{
    \PVGRHasKind{\Ctx,\TVar:\Kind}\TVar\Kind
  }
}
\newcommand\rulePVGRKindingApp{
  \inferrule[K-App]{
    \PVGRHasKind{\Ctx}{\Typ_1}{\Kind_1 \to \Kind_2} \\
    \PVGRHasKind{\Ctx}{\Typ_2}{\Kind_1}
  }{
    \PVGRHasKind{\Ctx}{\Typ_1~\Typ_2}{\Kind_2}
  }
}
\newcommand\rulePVGRKindingLam{
  \inferrule[K-Lam]{
    \PVGRHasKind{\Ctx}{\POLYShape}{\KShape} \\
    \PVGRHasKind{\CtxRestrictNonDom{\Ctx},\TVar : \KDomain \POLYShape}{\Typ}{\Kind} \\
    \Kind \in \{ \KType, \KSt \}
  }{
    \PVGRHasKind{\Ctx}{\TLam\TVar\POLYShape\Typ}{\KDomain \POLYShape \to \Kind}
  }
}
\newcommand\rulePVGRKindingAll{
  \inferrule[K-All]{
    \PVGRIsCtx{\Ctx, \TVar : \Kind, \Cstr} \\
    \PVGRHasKind{\Ctx, \TVar : \Kind, \Cstr}{\Typ}{\KType}
  }{
    \PVGRHasKind{\Ctx}{\TAll\TVar\Kind{\Cstr}\Typ}{\KType}
  }
}
\newcommand\rulePVGRKindingArr{
  \inferrule[K-Arr]{
    \PVGRHasKind{\Ctx_1}{\St_1}{\KSt} \and
    \PVGRHasKind{\Ctx_1}{\Typ_1}{\KType} \\\\
    \PVGRHasKind{\Ctx_1\DisjointAppend\Ctx_2}{\St_2}{\KSt} \and
    \PVGRHasKind{\Ctx_1\DisjointAppend\Ctx_2}{\Typ_2}{\KType} \\\\
    \PVGRIsCtx{\Ctx_1\DisjointAppend\Ctx_2} \and
    \Ctx_2 = \CtxRestrictOnlyDom{\Ctx_2}
  }{
    \PVGRHasKind
      {\Ctx_1}
      {\TArr{\St_1}{\Typ_1}{\Ctx_2}{\St_2}{\Typ_2}}
      {\KType}
  }
}
\newcommand\rulePVGRKindingChan{
  \inferrule[K-Chan]{
    \PVGRHasKind{\Ctx}{\POLYDom}{\KDomain \TShapeOne}
  }{
    \PVGRHasKind{\Ctx}{\TChan\POLYDom}{\KType}
  }
}
\newcommand\rulePVGRKindingAccessPoint{
  \inferrule[K-AccessPoint]{
    \PVGRHasKind{\Ctx}{\Ses}{\KSession}
  }{
    \PVGRHasKind{\Ctx}{\TAccessPoint\Ses}{\KType}
  }
}
\newcommand\rulePVGRKindingUnit{
  \inferrule[K-Unit]{
  }{
    \PVGRHasKind{\Ctx}{\TUnit}{\KType}
  }
}
\newcommand\rulePVGRKindingPair{
  \inferrule[K-Pair]{
    \PVGRHasKind{\Ctx}{\Typ_1}{\KType} \\
    \PVGRHasKind{\Ctx}{\Typ_2}{\KType}
  }{
    \PVGRHasKind{\Ctx}{\TPair{\Typ_1}{\Typ_2}}{\KType}
  }
}
\newcommand\rulePVGRKindingSum{
  \inferrule[K-Sum]{
    \PVGRHasKind{\Ctx}{\Typ_1}{\KType} \\
    \PVGRHasKind{\Ctx}{\Typ_2}{\KType}
  }{
    \PVGRHasKind{\Ctx}{\TSum{\Typ_1}{\Typ_2}}{\KType}
  }
}
\newcommand\rulePVGRKindingSend{
  \inferrule[K-Send]{
    \PVGRHasKind{\Ctx}{\POLYShape}{\KShape} \\
    \PVGRHasKind{\CtxRestrictNonDom{\Ctx},\TVar : \KDomain\POLYShape}{\St}{\KSt} \\
    \PVGRHasKind{\CtxRestrictNonDom{\Ctx},\TVar : \KDomain\POLYShape}{\Typ}{\KType} \\
    \PVGRHasKind{\Ctx}{\Ses}{\KSession}
  }{
    \PVGRHasKind{\Ctx}{\SSend\TVar{\KDomain\POLYShape}\St\Typ\Ses}{\KSession}
  }
}
\newcommand\rulePVGRKindingRecv{
  \inferrule[K-Recv]{
    \PVGRHasKind{\Ctx}{\POLYShape}{\KShape} \\
    \PVGRHasKind{\CtxRestrictNonDom{\Ctx},\TVar : \KDomain\POLYShape}{\St}{\KSt} \\
    \PVGRHasKind{\CtxRestrictNonDom{\Ctx},\TVar : \KDomain\POLYShape}{\Typ}{\KType} \\
    \PVGRHasKind{\Ctx}{\Ses}{\KSession}
  }{
    \PVGRHasKind{\Ctx}{\SRecv\TVar{\KDomain\POLYShape}\St\Typ\Ses}{\KSession}
  }
}
\newcommand\rulePVGRKindingBranch{
  \inferrule[K-Branch]{
    \PVGRHasKind{\Ctx}{\Ses_1}{\KSession} \\
    \PVGRHasKind{\Ctx}{\Ses_2}{\KSession}
  }{
    \PVGRHasKind{\Ctx}{\SBranch{\Ses_1}{\Ses_2}}{\KSession}
  }
}
\newcommand\rulePVGRKindingChoice{
  \inferrule[K-Choice]{
    \PVGRHasKind{\Ctx}{\Ses_1}{\KSession} \\
    \PVGRHasKind{\Ctx}{\Ses_2}{\KSession}
  }{
    \PVGRHasKind{\Ctx}{\SChoice{\Ses_1}{\Ses_2}}{\KSession}
  }
}
\newcommand\rulePVGRKindingEnd{
  \inferrule[K-End]{}{
    \PVGRHasKind{\Ctx}{\SEnd}{\KSession}
  }
}
\newcommand\rulePVGRKindingDual {
  \inferrule[K-Dual]{
    \PVGRHasKind{\Ctx}{\Ses}{\KSession}
  }{
    \PVGRHasKind{\Ctx}{\Dual\Ses}{\KSession}
  }
}
\newcommand\rulePVGRKindingDomMerge{
  \inferrule[K-DomMerge]{
    \PVGRHasKind{\Ctx}{\POLYDom_1}{\KDomain {\POLYShape_1}} \\
    \PVGRHasKind{\Ctx}{\POLYDom_2}{\KDomain {\POLYShape_2}} \\
    \PVGRCstrEntail\Ctx{\POLYDom_1 \Disjoint \POLYDom_2}
  }{
    \PVGRHasKind{\Ctx}{\POLYDom_1,\POLYDom_2}{\KDomain {\TShapePair{\POLYShape_1}{\POLYShape_2}}}
  }
}
\newcommand\rulePVGRKindingDomProj{
  \inferrule[K-DomProj]{
    \PVGRHasKind{\Ctx}{\POLYDom}{\KDomain {\TShapePair{\POLYShape_1}{\POLYShape_2}}}
  }{
    \PVGRHasKind{\Ctx}{\TDomProj \Label \POLYDom}{\KDomain {\POLYShape_\Label}}
  }
}
\newcommand\rulePVGRKindingDomFrom{
  \inferrule[K-DomFrom]{
    \PVGRHasKind{\Ctx}{\POLYDom}{\KDomain {\TShapeSum{\POLYShape_1}{\POLYShape_2}}}
  }{
    \PVGRHasKind{\Ctx}{\TDomFrom \Label \POLYDom}{\KDomain {\POLYShape_\Label}}
  }
}
\newcommand\rulePVGRKindingDomSum{
  \inferrule[K-DomSum]{
    (\forall\Label)\ \PVGRHasKind{\Ctx}{\POLYDom_\Label}{\KDomain {\POLYShape_\Label}}
  }{
    \PVGRHasKind{\Ctx}{\TDomEither{ \POLYDom_\LabelA}{ \POLYDom_\LabelB}}{\KDomain {\TShapeSum{\POLYShape_\LabelA}{\POLYShape_\LabelB}}}
  }
}
\newcommand\rulePVGRKindingDomZero{
  \inferrule[K-DomEmpty]{}{
    \PVGRHasKind{\Ctx}{\TDomZero}{\KDomain {\TShapeZero}}
  }
}
\newcommand\rulePVGRKindingShapeZero{
  \inferrule[K-ShapeEmpty]{}{
    \PVGRHasKind{\Ctx}{\TShapeZero}{\KShape}
  }
}
\newcommand\rulePVGRKindingShapeOne{
  \inferrule[K-ShapeChan]{}{
    \PVGRHasKind{\Ctx}{\TShapeOne}{\KShape}
  }
}
\newcommand\rulePVGRKindingShapePair{
  \inferrule[K-ShapePair]{
    \PVGRHasKind{\Ctx}{\POLYShape_1}{\KShape} \\
    \PVGRHasKind{\Ctx}{\POLYShape_2}{\KShape}
  }{
    \PVGRHasKind{\Ctx}{\TShapePair{\POLYShape_1}{\POLYShape_2}}{\KShape}
  }
}
\newcommand\rulePVGRKindingShapeSum{
  \inferrule[K-ShapeSum]{
    \PVGRHasKind{\Ctx}{\POLYShape_1}{\KShape} \\
    \PVGRHasKind{\Ctx}{\POLYShape_2}{\KShape}
  }{
    \PVGRHasKind{\Ctx}{\TShapeSum{\POLYShape_1}{\POLYShape_2}}{\KShape}
  }
}
\newcommand\rulePVGRKindingStEmpty{
  \inferrule[K-StEmpty]{}{
    \PVGRHasKind{\Ctx}{\Empty}{\KSt}
  }
}
\newcommand\rulePVGRKindingStChan{
  \inferrule[K-StChan]{
    \PVGRHasKind{\Ctx}{\POLYDom}{\KDomain \TShapeOne} \\
    \PVGRHasKind{\Ctx}{\Ses}{\KSession}
  }{
    \PVGRHasKind{\Ctx}{\StBind\POLYDom\Ses}{\KSt}
  }
}

\newcommand\rulePVGRKindingStMerge{
  \inferrule[K-StMerge]{
    \PVGRHasKind{\Ctx}{\St_1}{\KSt} \\
    \PVGRHasKind{\Ctx}{\St_2}{\KSt} \\
    \PVGRCstrEntail{\Ctx}{\Dom{\St_1} \Disjoint \Dom{\St_2}} \\
  }{
    \PVGRHasKind{\Ctx}{\St_1, \St_2}{\KSt}
  }
}

\newcommand\rulePVGRTypingVar{
  \inferrule[T-Var]{}{
    \PVGRHasValType{\Ctx,\EVar:\Typ}\EVar\Typ
  }
}

\newcommand\rulePVGRTypingChan{
  \inferrule[T-Chan]{
    \PVGRHasKind{\Ctx}{\POLYDom}{\KDomain{\TShapeOne}}
  }{
    \PVGRHasValType{\Ctx}{\EChan\POLYDom}{\TChan\POLYDom}
  }
}

\newcommand\rulePVGRTypingAbs{
  \inferrule[T-Abs]{
    \PVGRHasKind
      {\Ctx_1}
      {\TArr {\St_1} {\Typ_1} {\Ctx_2} {\St_2} {\Typ_2}}
      {\KType}
    \\
    \PVGRHasExpType
      {\Ctx_1, \EVar : \Typ_1}
      {\St_1}
      {\Exp}
      {\Ctx_2}
      {\St_2}
      {\Typ_2}
  }{
    \PVGRHasValType
      {\Ctx_1}
      {\EAbs{\St_1}\EVar{\Typ_1}\Exp}
      {\TArr {\St_1} {\Typ_1} {\Ctx_2} {\St_2} {\Typ_2}}
  }
}
\newcommand\rulePVGRTypingTAbs{
  \inferrule[T-TAbs]{
    \PVGRHasKind
      {\Ctx}
      {\TAll {\TVar} {\Kind} {\Cstr} {\Typ}}
      {\KType}
    \\
    \PVGRHasValType
      {\Ctx, \TVar : \Kind, \Cstr}
      {\Val}
      {\Typ}
  }{
    \PVGRHasValType
      {\Ctx}
      {\ETAbs {\TVar} {\Kind} {\Cstr} {\Val}}
      {\TAll {\TVar} {\Kind} {\Cstr} {\Typ}}
  }
}
\newcommand\rulePVGRTypingUnit{
  \inferrule[T-Unit]{}{
    \PVGRHasValType
      {\Ctx}
      {\EUnit}
      {\TUnit}
  }
}
\newcommand\rulePVGRTypingPair{
  \inferrule[T-Pair]{
    \PVGRHasValType
      {\Ctx}
      {\Val_1}
      {\Typ_1}
    \\
    \PVGRHasValType
      {\Ctx}
      {\Val_2}
      {\Typ_2}
  }{
    \PVGRHasValType
      {\Ctx}
      {\EPair {\Val_1} {\Val_2}}
      {\TPair {\Typ_1} {\Typ_2}}
  }
}

\newcommand\rulePVGRTypingVal{
  \inferrule[T-Val]{
    \PVGRHasValType{\Ctx}{\Val}{\Typ}
  }{
    \PVGRHasExpType
      {\Ctx}
      {\St}
      {\Val}
      {\Empty}
      {\St}
      {\Typ}
  }
}

\newcommand\rulePVGRTypingNew{
  \inferrule[T-New]{
    \PVGRHasKind{\Ctx}{\Ses}{\KSession}
  }{
    \PVGRHasExpType
      {\Ctx}
      {\Empty}
      {\ENew \Ses}
      {\Empty}
      {\Empty}
      {\TAccessPoint\Ses}
  }
}

\newcommand\rulePVGRTypingRequest{
  \inferrule[T-Request]{
    \PVGRHasValType{\Ctx}{\Val}{\TAccessPoint\Ses}
  }{
    \PVGRHasExpType
      {\Ctx}
      {\Empty}
      {\ERequest\Val}
      {\TVar:\KDomain\TShapeOne}
      {\StBind\TVar\Ses}
      {\TChan \TVar}
  }
}

\newcommand\rulePVGRTypingAccept{
  \inferrule[T-Accept]{
    \PVGRHasValType{\Ctx}{\Val}{\TAccessPoint\Ses}
  }{
    \PVGRHasExpType
      {\Ctx}
      {\Empty}
      {\EAccept\Val}
      {\TVar:\KDomain\TShapeOne}
      {\StBind\TVar{\Dual\Ses}}
      {\TChan \TVar}
  }
}

\newcommand\rulePVGRTypingLet{
  \inferrule[T-Let]{
    \PVGRHasExpType
      {\Ctx_1}
      {\St_1}
      {\Exp_1}
      {\Ctx_2}
      {\St_2'}
      {\Typ_1}
    \and
    \PVGRHasExpType
      {\Ctx_1 \DisjointAppend \Ctx_2, \EVar:\Typ_1}
      {\St_2,\St_2'}
      {\Exp_2}
      {\Ctx_3}
      {\St_3}
      {\Typ_2}
    \\
    \PVGRHasKind
      {\Ctx_1 \DisjointAppend \Ctx_2, \EVar:\Typ_1}
      {\St_2,\St_2'}
      {\KSt}
  }{
    \PVGRHasExpType
      {\Ctx_1}
      {\St_1,\St_2}
      {\ELet\EVar{\Exp_1}{\Exp_2}}
      {\Ctx_2,\Ctx_3}
      {\St_3}
      {\Typ_2}
  }
}
\newcommand\rulePVGRTypingFork{
  \inferrule[T-Fork]{
    \PVGRHasValType
      {\Ctx}
      {\Val}
      {\TArr{\St}{\TUnit}{\cdot}{\cdot}{\TUnit}}
  }{
    \PVGRHasExpType
      {\Ctx}
      {\St}
      {\EFork\Val}
      {\Empty}
      {\Empty}
      {\TUnit}
  }
}
\newcommand\rulePVGRTypingApp{
  \inferrule[T-App]{
    \PVGRHasValType
      {\Ctx_1}
      {\Val_1}
      {\TArr{\St_1}{\Typ_1}{\Ctx_2}{\St_2}{\Typ_2}}
    \\
    \PVGRHasValType
      {\Ctx_1}
      {\Val_2}
      {\Typ_1}
  }{
    \PVGRHasExpType
      {\Ctx_1}
      {\St_1}
      {\EApp{\Val_1}{\Val_2}}
      {\Ctx_2}
      {\St_2}
      {\Typ_2}
  }
}
\newcommand\rulePVGRTypingSend{
  \inferrule[T-Send]{
    \PVGRHasKind \Ctx {\POLYDom'} {\KDomain\POLYShape} \and
    \PVGRTypeConv {\Subst{\POLYDom'}{\TVar'}{\St'}} \St \and
    \PVGRTypeConv {\Subst{\POLYDom'}{\TVar'}{\Typ'}} \Typ \\
    \PVGRHasKind \Ctx {\POLYDom} {\KDomain\TShapeOne} \and
    \PVGRHasValType{\Ctx}{\Val_1}{\Typ} \and
    \PVGRHasValType{\Ctx}{\Val_2}{\TChan \POLYDom} \\
  }{
    \PVGRHasExpType
      {\Ctx}
      {\St, \StBind\POLYDom{\SSend{\TVar'}{\KDomain\POLYShape}{\St'}{\Typ'}\Ses}}
      {\ESend{\Val_1}{\Val_2}}
      {\Empty}
      {\StBind\POLYDom\Ses}
      {\TUnit}
  }
}
\newcommand\rulePVGRTypingRecv{
  \inferrule[T-Recv]{
    \PVGRHasKind \Ctx {\POLYDom} {\KDomain\TShapeOne} \and
    \PVGRHasValType{\Ctx}{\Val}{\TChan \POLYDom}
  }{
    \PVGRHasExpTypeBegin
      {\Ctx}
      {\StBind\POLYDom{\SRecv{\TVar'}{\KDomain\POLYShape}{\St'}{\Typ'}\Ses}}
      {\ERecv{\Val}}
    \\
    \PVGRHasExpTypeEnd
      {(\TVar' : \KDomain\POLYShape)}
      {\St', \StBind\POLYDom\Ses}
      {\Typ'}
  }
}
\newcommand\rulePVGRTypingSelect{
  \inferrule[T-Select]{
    \PVGRHasValType{\Ctx}{\Val}{\TChan\POLYDom}
  }{
    \PVGRHasExpType
      {\Ctx}
      {\StBind\POLYDom{\SChoice{\Ses_1}{\Ses_2}}}
      {\ESelect \Label \Val}
      {\Empty}
      {\StBind\POLYDom{\Ses_\Label}}
      {\TUnit}
  }
}
\newcommand\rulePVGRTypingCase{
  \inferrule[T-Case]{
    \PVGRHasValType{\Ctx_1}{\Val}{\TChan\POLYDom} \\
    (\forall \Label)\
    \PVGRHasExpType
      {\Ctx_1}
      {\St_1, \StBind\POLYDom{\Ses_\Label}}
      {\Exp_\Label}
      {\Ctx_2}
      {\St_2}
      {\Typ}
  }{
    \PVGRHasExpType
      {\Ctx_1}
      {\St_1, \StBind\POLYDom{\SBranch{\Ses_1}{\Ses_2}}}
      {\ECase \Val {\Exp_1} {\Exp_2}}
      {\Ctx_2}
      {\St_2}
      {\Typ}
  }
}

\newcommand\rulePVGRTypingMatch{
  \inferrule[T-Match]{
    \PVGRHasKind\Ctx{\POLYDom}{\KDomain{\TShapeSum{\POLYShape_\LabelA}{\POLYShape_\LabelB}}}\and
    (\forall\Label)\ \PVGRHasKind\Ctx{\StFun_\ell}{\KStFun{\POLYShape_\Label}} \\
    (\forall\Label)\ \PVGRHasKind\Ctx{\TFun_\ell}{\KTypeFun{\POLYShape_\Label}} \\
    (\forall\Label)\ \PVGRTypeConv {\TFun_\Label(\TDomFrom\Label\POLYDom)} {\Typ_\Label}
    \\
    \PVGRHasValType{\Ctx}{\Val}{\TSum{\Typ_\LabelA}{\Typ_\LabelB}} \and
    (\forall \Label)\
    \PVGRHasExpType
      {\Ctx, x_\Label : \TFun_\Label(\TDomFrom\Label\POLYDom)}
      {\St, \StFun_\Label(\TDomFrom\Label\POLYDom)}
      {\Exp_\Label}
      {\Ctx'}
      {\St'}
      {\Typ}
  }{
    \PVGRHasExpType
      {\Ctx}
      {\St, \StFun_1 (\TDomFrom\LabelA \POLYDom), \StFun_2 (\TDomFrom\LabelB \POLYDom)}
      {\EMatch \Val {x_1:\Exp_1} {x_2:\Exp_2}}
      {\Ctx'}
      {\St'}
      {\Typ}
  }
}
\newcommand\rulePVGRTypingClose{
  \inferrule[T-Close]{
    \PVGRHasValType{\Ctx}{\Val}{\TChan\POLYDom}
  }{
    \PVGRHasExpType
      {\Ctx}
      {\StBind\POLYDom\SEnd}
      {\EClose\Val}
      {\Empty}
      {\Empty}
      {\TUnit}
  }
}
\newcommand\rulePVGRTypingProj{
  \inferrule[T-Proj]{
    \PVGRHasValType{\Ctx}{\Val}{\TPair{\Typ_1}{\Typ_2}}
  }{
    \PVGRHasExpType
      {\Ctx}
      {\Empty}
      {\EProj \Label \Val}
      {\Empty}
      {\Empty}
      {\Typ_\Label}
  }
}

\newcommand\rulePVGRTypingInj{
  \inferrule[T-Inj1]{
    (\forall\Label)\ \PVGRHasKind\Ctx{\StFun_\ell}{\KStFun{\POLYShape_\Label}} \and
    (\forall\Label)\ \PVGRHasKind\Ctx{\TFun_\ell}{\KTypeFun{\POLYShape_\Label}} \\
    \PVGRHasValType{\Ctx}{\Val}{\Typ} \\
    \PVGRHasKind\Ctx\POLYDom{\KDomain{\POLYShape_\LabelA}} \\
    \PVGRTypeConv {\StFun_\LabelA~\POLYDom} \St \\
    \PVGRTypeConv {\TFun_\LabelA~\POLYDom} \Typ
  }{
    \PVGRHasExpType
      {\Ctx}
      {\St}
      {\EInj \LabelA \Val}
      {\Nchanb : \KDomain{{\POLYShape_2}}}
      {\StFun_\LabelA~\POLYDom, \StFun_\LabelB~\Nchanb}
      {\TSum{\TFun_\LabelA~\POLYDom}{\TFun_\LabelB~\Nchanb}}
  }
}
\newcommand\rulePVGRTypingTApp{
  \inferrule[T-TApp]{
    \PVGRHasValType
      {\Ctx}
      {\Val}
      {\TAll\TVar\Kind{\Cstr}\Typ}
    \\
    \PVGRHasKind
      {\Ctx}
      {\Typ'}
      {\Kind}
    \\
    \PVGRCstrEntail
      {\Ctx}
      {\Subst{\Typ'}\TVar{\Cstr}}
    \\
    \PVGRTypeConv
      {\Subst{\Typ'}\TVar\Typ}
      {\Typ''}
  }{
    \PVGRHasExpType
      {\Ctx}
      {\Empty}
      {\ETApp\Val{\Typ'}}
      {\Empty}
      {\Empty}
      {\Typ''}
  }
}

\newcommand\rulePVGRTypingExp{
  \inferrule[T-Exp]{
    \PVGRHasExpType
      {\Ctx}
      {\St}
      {\Exp}
      {\Ctx'}
      {\Empty}
      {\Typ}
  }{
    \PVGRHasConfType{\Ctx}{\St}{\Exp}
  }
}

\newcommand\rulePVGRTypingPar{
  \inferrule[T-Par]{
    \PVGRHasConfType{\Ctx}{\St_1}{\Cfg_1} \\
    \PVGRHasConfType{\Ctx}{\St_2}{\Cfg_2}
  }{
    \PVGRHasConfType{\Ctx}{\St_1,\St_2}{\CPar{\Cfg_1}{\Cfg_2}}
  }
}

\newcommand\rulePVGRTypingBindChan{
  \inferrule[T-NuChan]{
    \TVar,\TVar' \text{ not free in } \Ctx \\
    \PVGRHasKind
      {\Ctx}
      {\Ses}
      {\KSession} \\
    \PVGRHasConfType
      {\Ctx \DisjointAppend \TVar:\KDomain\TShapeOne \DisjointAppend \TVar':\KDomain\TShapeOne}
      {\St,\StBind\TVar\Ses,\StBind{\TVar'}{\Dual\Ses}}
      {\Cfg}
  }{
    \PVGRHasConfType
      {\Ctx}
      {\St}
      {\CBindChan\TVar{\TVar'}\Ses\Cfg}
  }
}

\newcommand\rulePVGRTypingBindChanClosed{
  \inferrule[T-NuChanClosed]{
    \TVar,\TVar' \text{ not free in } \Ctx \\
    \PVGRHasConfType
      {\Ctx \DisjointAppend \TVar:\KDomain\TShapeOne \DisjointAppend \TVar':\KDomain\TShapeOne}
      {\St}
      {\Cfg}
  }{
    \PVGRHasConfType
      {\Ctx}
      {\St}
      {\CBindChan\TVar{\TVar'}\SEnd\Cfg}
  }
}

\newcommand\rulePVGRTypingBindAP{
  \inferrule[T-NuAccess]{
    \EVar \text{ not free in } \Ctx \\
    \PVGRHasKind
      {\Ctx}
      {\Ses}
      {\KSession} \\
    \PVGRHasConfType
      {\Ctx, \EVar:\TAccessPoint\Ses}
      {\St}
      {\Cfg}
  }{
    \PVGRHasConfType
      {\Ctx}
      {\St}
      {\CBindAP\EVar\Ses\Cfg}
  }
}

\newcommand\rulePVGRTypeConvTApp{
  \inferrule[TC-TApp]{}{
    \PVGRTypeConv
      {(\TLam\TVar\POLYShape{\Typ_1})~\Typ_2}
      {\Subst{\Typ_2}{\TVar}{\Typ_1}}
  }
}

\newcommand\rulePVGRTypeConvDualSend{
  \inferrule[TC-DualSend]{}{
    \PVGRTypeConv
      {\Dual{\SSend\TVar{\KDomain\POLYShape}\St\Typ\Ses}}
      {\SRecv\TVar{\KDomain\POLYShape}\St\Typ{\Dual\Ses}}
  }
}

\newcommand\rulePVGRTypeConvDualRecv{
  \inferrule[TC-DualRecv]{}{
    \PVGRTypeConv
      {\Dual{\SRecv\TVar{\KDomain\POLYShape}\St\Typ\Ses}}
      {\SSend\TVar{\KDomain\POLYShape}\St\Typ{\Dual\Ses}}
  }
}

\newcommand\rulePVGRTypeConvDualChoice{
  \inferrule[TC-DualChoice]{}{
    \PVGRTypeConv
      {\Dual{\SChoice{\Ses_1}{\Ses_2}}}
      {\SBranch{\Dual{\Ses_1}}{\Dual{\Ses_2}}}
  }
}

\newcommand\rulePVGRTypeConvDualBranch{
  \inferrule[TC-DualBranch]{}{
    \PVGRTypeConv
      {\Dual{\SBranch{\Ses_1}{\Ses_2}}}
      {\SChoice{\Dual{\Ses_1}}{\Dual{\Ses_2}}}
  }
}

\newcommand\rulePVGRTypeConvDualEnd{
  \inferrule[TC-DualEnd]{}{
    \PVGRTypeConv
      {\Dual\SEnd}
      {\SEnd}
  }
}

\newcommand\rulePVGRTypeConvDualVar{
  \inferrule[TC-DualVar]{}{
    \PVGRTypeConv
      {\Dual{\Dual\TVar}}
      {\TVar}
  }
}

\newcommand\rulePVGRTypeConvProj{
  \inferrule[TC-Proj]{}{
    \PVGRTypeConv
      {\TDomProj \Label (\POLYDom_1,\POLYDom_2)}
      {\POLYDom_\Label}
  }
}

\newcommand\rulePVGRExprRedBetaFun{
  \inferrule[ER-BetaFun]{}{
    (\EAbs{\St}{\EVar}{\Typ}{\Exp_1})~\Val_2 \ReducesToE \Subst{\Val_2}{\EVar} \Exp_1
  }
}

\newcommand\rulePVGRExprRedBetaAll{
  \inferrule[ER-BetaAll]{}{
    \ETApp{(\ETAbs{\TVar}{\Kind}{\Cstr}{\Val})}{\Typ} \ReducesToE \Subst{\Typ}{\TVar} \Val
  }
}

\newcommand\rulePVGRExprRedBetaLet{
  \inferrule[ER-BetaLet]{}{
    \ELet{\EVar}{\Val_1}{\Exp_2} \ReducesToE \Subst{\Val_1}{\EVar} \Exp_2
  }
}

\newcommand\rulePVGRExprRedBetaPair{
  \inferrule[ER-BetaPair]{}{
    \EProj \Label {\EPair{\Val_1}{\Val_2}} \ReducesToE \Val_\Label
  }
}

\newcommand\rulePVGRExprRedLift{
  \inferrule[ER-Lift]{
    \Exp_1 \ReducesToE \Exp_2
  }{
    \ELet\EVar{\Exp_1}\Exp \ReducesToE \ELet\EVar{\Exp_2}\Exp
  }
}

\newcommand\rulePVGRCongNull{
  \inferrule[CC-Null]{}{
    \CPar{\Cfg}{\EUnit} \Cong {\Cfg}
  }
}

\newcommand\rulePVGRCongComm{
  \inferrule[CC-Comm]{}{
    \CPar{\Cfg_1}{\Cfg_2} \Cong \CPar{\Cfg_2}{\Cfg_1}
  }
}

\newcommand\rulePVGRCongAssoc{
  \inferrule[CC-Assoc]{}{
    \CPar{\Cfg_1}{(\CPar{\Cfg_2}{\Cfg_3})} \Cong \CPar{(\CPar{\Cfg_1}{\Cfg_2})}{\Cfg_3}
  }
}

\newcommand\rulePVGRCongScopeChan{
  \inferrule[CC-Scope-Chan]{
    \TVar,\TVar' \text{ not free in } \Cfg_1
  }{
    \CPar{\Cfg_1}{(\CBindChan\TVar{\TVar'}\Ses\Cfg_2)} \Cong \CBindChan\TVar{\TVar'}\Ses{(\CPar{\Cfg_1}{\Cfg_2})}
  }
}

\newcommand\rulePVGRCongSwap{
  \inferrule[CC-Swap]{}{
    {\CBindChan\TVar{\TVar'}\Ses\Cfg} \Cong \CBindChan {\TVar'}\TVar{\Dual\Ses}\Cfg
  }
}

\newcommand\rulePVGRCongScopeAP{
  \inferrule[CC-Scope-Access]{
    \EVar \text{ not free in } \Cfg_1
  }{
    \CPar{\Cfg_1}{(\CBindAP\EVar\Ses\Cfg_2)} \Cong \CBindAP\EVar\Ses{(\CPar{\Cfg_1}{\Cfg_2})}
  }
}

\newcommand\rulePVGRCongLift{
  \inferrule[CC-Lift]{
    \Cfg_1 \Cong \Cfg_2
  }{
    \CCtx{\Cfg_1} \Cong \CCtx{\Cfg_2}
  }
}

\newcommand\rulePVGRCfgRedFork{
  \inferrule[CR-Fork]{}{
    \CCtx{\ECtx{\EFork\Val}} \ReducesToC
    \CCtx{
    \CPar
      {(\EApp{\Val}{\EUnit})}
      {\ECtx{\EUnit}}}
  }
}

\newcommand\rulePVGRCfgRedSendRecv{
  \inferrule[CR-SendRecv]{
    \Cfg \Cong \CCtxBegin
    \CBindChan
      {\TVar} {\TVar'}
      {\SSend\TVar{\KDomain\POLYShape}\St\Typ\Ses}
      {\begin{array}[t]{@{}l@{}l}
         (&\ECtxA{\ESend\Val{\EChan\TVar}}\\
         \CPar{}{}&\CPar{\ECtxB{\ERecv{\EChan\TVar'}}}{\Cfg'}
                    )\CCtxEnd
       \end{array}}
  }{
    {\Cfg}
    \ReducesToC
    \CCtx{
    \CBindChan {\TVar} {\TVar'} {\Ses} {(
      \CPar{
      \CPar
        {\ECtxA{\EUnit}}
        {\ECtxB{\Val}}}{\Cfg'}
    )}}
  }
}

\newcommand\rulePVGRCfgRedSelectCase{
  \inferrule[CR-SelectCase]{
    \Cfg \Cong
    \CCtxBegin
    \CBindChan
      {\TVar} {\TVar'}
      {\SChoice{\Ses_1}{\Ses_2}}
      {\begin{array}[t]{@{}l@{}l}
         (&\ECtxA{\ESelect \Label {\EChan\TVar}}\\
         \CPar{}{}&\CPar{\ECtxB{\ECase{\EChan\TVar'}{\Exp_1}{\Exp_2}}}{\Cfg'}
                    )\CCtxEnd
       \end{array}}
  }{
    {\Cfg}
    \ReducesToC
    \CCtx{
    \CBindChan {\TVar}{\TVar'} {\Ses_\Label} {(
      \CPar{
      \CPar
        {\ECtxA{\EUnit}}
        {\ECtxB{\Exp_\Label}}}{\Cfg'}
    )}}
  }
}

\newcommand\rulePVGRCfgRedClose{
  \inferrule[CR-Close]{
    \Cfg \Cong
    \CCtx{\CBindChan {\TVar} {\TVar'} {\SEnd} {(
      \CPar{
      \CPar
        {\ECtxA{\EClose{\EChan \TVar}}}
        {\ECtxB{\EClose{\EChan \TVar'}}}}{\Cfg'}
    )}}
  }{
    {\Cfg}
    \ReducesToC
    \CCtx{
  \CBindChan {\TVar} {\TVar'} {\SEnd}
    {(\CPar {\CPar {\ECtxA{\EUnit}} {\ECtxB{\EUnit}}} {\Cfg'})}}
  }
}

\newcommand\rulePVGRCfgRedExpr{
  \inferrule[CR-Expr]{
    \Exp_1 \ReducesToE \Exp_2
  }{
    \CCtx{\Exp_1} \ReducesToC \CCtx{\Exp_2}
  }
}

\newcommand\rulePVGRCfgRedNew{
  \inferrule[CR-New]{
    \EVar \text{ fresh}
  }{
    \CCtx{\ECtx{\ENew\Ses}} \ReducesToC \CCtx{\CBindAP\EVar\Ses{\ECtx\EVar}}
  }
}

\newcommand\rulePVGRCfgRedRequestAccept{
  \inferrule[CR-RequestAccept]{
    \TVar, \TVar' \text{ fresh} \\
    \Cfg \Cong \CCtx{\CBindAP\EVar\Ses {(
      \CPar{
      \CPar{
        \ECtxA{\ERequest\EVar}
      }{
        \ECtxB{\EAccept\EVar}
      }}{\Cfg'}
    )}}
  }{
    {\Cfg}
    \ReducesToC
    \CCtx{
    \CBindAP\EVar\Ses{
      \CBindChan\TVar{\TVar'}\Ses{(
        \CPar{
        \CPar{
          \ECtxA{\EChan\TVar}
        }{
          \ECtxB{\EChan\TVar'}
        }}{\Cfg'}
      )}
    }}
  }
}

\begin{document}

\title{Polymorphic Typestate for Session Types}


\author{Hannes Saffrich}
\orcid{0009-0004-7014-754X}             
\affiliation{
  \institution{University of Freiburg}            
  \country{Germany}                    
}
\email{saffrich@informatik.uni-freiburg.de}          

\author{Peter Thiemann}
\orcid{0000-0002-9000-1239}             
\affiliation{
  \institution{University of Freiburg}           
  \country{Germany}                   
}
\email{thiemann@informatik.uni-freiburg.de}         

\begin{abstract}
  Session types provide a principled approach to typed communication
  protocols that guarantee type safety and protocol fidelity.
  Formalizations of session-typed communication are typically
  based on process calculi, concurrent lambda calculi, or linear logic.
  An alternative model based on context-sensitive
  typing and typestate has not received much attention due to its apparent
  restrictions. However, this model is attractive
  because it does not force programmers into particular patterns
  like continuation-passing style or channel-passing style, but rather
  enables them to treat communication channels like mutable  variables.

  Polymorphic typestate is the key that enables a full treatment of
  session-typed communication. Previous work in this direction
  was hampered by its setting in a simply-typed lambda calculus.
  We show that higher-order polymorphism and existential types enable
  us to lift the restrictions imposed by the previous work, thus
  bringing the expressivity of the typestate-based approach on par
  with the competition. On this
  basis, we define \PolyVGR, the system of polymorphic typestate for
  session types, establish its basic metatheory, type preservation and 
  progress, and present a prototype implementation.
\end{abstract}



\maketitle

\section{Introduction}
\label{sec:introduction}

When Honda and others
\cite{DBLP:conf/concur/Honda93,DBLP:conf/parle/TakeuchiHK94} proposed
session types, little did they know that their system would become a
cornerstone for type disciplines for communication protocols. Their original
system describes bidirectional, heterogeneously typed communication
channels between two processes in pi-calculus. It also contains facilities for offering and
accepting choices in the protocol.

Subsequent work added a plethora of features to the original
system. One strand of ongoing work considers session-typed embeddings of
communication primitives in functional and object-oriented languages,
both theoretically and practically oriented
\cite{DBLP:journals/jfp/GayV10,DBLP:conf/ecoop/HuKPYH10,lindley17:_light_funct_session_types,DBLP:conf/ecoop/ScalasY16,DBLP:journals/jfp/Padovani17}.
These embeddings impose particular programming styles, following the structure of session types.
For example, embeddings in linear functional languages
\cite{DBLP:journals/jfp/GayV10,lindley17:_light_funct_session_types}
impose writing code in what we call \emph{channel-passing style} as
demonstrated in Listing~\ref{lst:channel-passing-style}.
\begin{lstlisting}[numbers=none,caption={Channel-passing style},label={lst:channel-passing-style}]
let (x, c2) = receive c1 in
let (y, c3) = receive c2 in
let c4 = send (x + y, c3) in ...
\end{lstlisting}
We enter this code with the typing \lstinline/c1 : ?Int.?Int.!Int.s0/, which means
that \lstinline/c1/ is a channel ready to receive two integers,
then send one, and continue the protocol according to session type \lstinline/s0/. In these
systems, channels are linear resources, so \lstinline/c1/ must be used
exactly once: it is consumed in line~1 and cannot be used thereafter. The
operation \lstinline/receive/ has type $\Inp T S \to (T \times
S)$. When it consumes \lstinline/c1/,
it returns \lstinline/c2/ of type \lstinline/?Int.!Int.s0/, which is further transformed to
\lstinline/c3/ of type \lstinline/!Int.s0/ by the next \lstinline/receive/, and finally
to \lstinline/c4 : s0/ by the \lstinline/send/ operation of type $(T
\times \Outp T S) \to S$.

Writing a program in this style is cumbersome as programmers have to
thread the channel explicitly through the program.
This style is not safe for embedding session types in general
programming languages because most
languages do not enforce the linearity needed to avoid aliasing 
of channel ends at compile time (some implementations
check linear use at run time
\cite{DBLP:conf/ecoop/ScalasY16,DBLP:journals/jfp/Padovani17}). Wrapping
the channel passing in a parameterized 
monad \cite{DBLP:journals/jfp/Atkey09} would accommodate the typing
requirements and ensure linearity by encapsulation, but it is again
cumbersome to scale the monadic style to programs handling more than
one channel at the same time. Nevertheless,
\citet{DBLP:conf/haskell/PucellaT08} developed a Haskell
implementation of session types in this style. In object-oriented languages, fluent
interfaces enable the correct chaining of method calls according to a
session type~\cite{DBLP:conf/fase/HuY16}, but have similar issues as
channel-passing style when scaling to multiple channels and new issues
with receiving values which seems to require mutable references as shown in
Listing~\ref{lst:fluent-interface-with-references}.
\begin{lstlisting}[caption={Fluent interface with references},label={lst:fluent-interface-with-references}]
var x = new Ref<Int>();
var y = new Ref<Int>();
var c4 = c1.receive(x).receive(y)
           .send(x.val + y.val);
\end{lstlisting}

An alternative approach is inspired by systems with typestate
\cite{DBLP:journals/tse/StromY86}. \citet{DBLP:journals/tcs/VasconcelosGR06} proposed
a multithreaded functional language on this basis.
Their language, which we call VGR, enables programming in direct
style; it does not
require linear handling of variables; and it scales to multiple
channels.
Listing~\ref{lst:example-server} contains a program fragment in VGR
equivalent to the code in Listing~\ref{lst:channel-passing-style}.
\begin{figure}[tp]
  \begin{minipage}[t]{0.49\linewidth}
\begin{lstlisting}[label={lst:example-server},caption={Example
server},captionpos=b]
fun server u =
  let x = receive u in
  let y = receive u in
  send x + y on u
\end{lstlisting}
  \end{minipage}
  \begin{minipage}[t]{0.49\linewidth}
\begin{lstlisting}[label={lst:example-server-unit},caption={Example server with capture},captionpos=b]
fun server' () =
  let x = receive u in
  let y = receive u in
  send x + y on u
\end{lstlisting}
  \end{minipage}
\end{figure}
The parameter \lstinline/u/ of the \lstinline/server/ function is a
\emph{channel reference} of type $\Chan\Nchan$, where $\Nchan$ is a
variable representing a \emph{channel identity}.
The operation \lstinline/receive/ takes a channel reference associated with
session type \lstinline{!Int.}$S$
and returns an integer. The association is maintained at compile time
in a typestate $\Sigma = \{\Nchan \mapsto \texttt{!Int.}S\}$ that maps channel identities to session types. 
As a compile-time side effect,
\lstinline/receive/ changes the typestate to $\Sigma' = \{\Nchan
\mapsto S\}$. Thus, we can describe the action of
\lstinline/receive/ by the typing:
\begin{align*}
  \mathtt{receive}:{} & \{\Nchan \mapsto \texttt{!Int.}S\}; \Chan\Nchan
                     \to (\Chan\Nchan \times \mathtt{Int}); \{\Nchan \mapsto S\}
\end{align*}
The general shape of a function type in VGR is thus:
$
  \Sigma_1; T_1 \to T_2; \Sigma_2
$. Here, $T_1$ and $T_2$ are argument and return type of the
function. The typestate environments $\Sigma_1$ and $\Sigma_2$ map channel
identities to session types. They reflect
the state (session type) of the channels before ($\Sigma_1$) and after ($\Sigma_2$)
calling the function.  Channels in $T_1$
refer to entries in $\Sigma_1$ and channels in $T_2$ refer to entries
in $\Sigma_2$.

Similarly, the function
\lstinline/send_on_/ takes an integer to transmit and a channel
reference associated with session type \lstinline{!Int.S}.
It returns a unit value and updates the channel's type to
\lstinline{S}. 
Putting these typings together, we obtain the VGR type of the \lstinline/server/ function in
Listing~\ref{lst:example-server}:
\begin{gather}\label{eq:1}
  \{ \Nchan : \Inp\Int\Inp\Int\Outp\Int S \} ; \Chan\Nchan \to \Unit;
  \{ \Nchan : S \} \text,
\end{gather}
for some channel name $\Nchan$ and session type $S$.
Listing~\ref{lst:example-server} also demonstrates that
VGR does \textbf{not} impose linear handling of channel references, as
there are multiple uses of variable \lstinline/u/. Instead, it keeps
track of the current state of every channel using the typestate $\Sigma$, which
is threaded linearly through the typing rules, at compile time.

As VGR is based
on simple types, the typing~(\ref{eq:1}) is severely restricted.
\begin{enumerate}
\item The function \lstinline/server/ is tied to a single continuation session
  type $S$, a restriction shared with many functional systems \cite{DBLP:journals/jfp/GayV10,DBLP:journals/pacmpl/FowlerLMD19}.
\item The function \lstinline/server/ can only be called on the single channel identified
  by $\Nchan$. 
\end{enumerate}

The language \PolyVGR that we propose here fixes all those drawbacks, and
more. The \PolyVGR type for \lstinline/server/ abstracts over
continuation sessions and channel identities:\footnote{Boxes with a frame
    highlight types and expressions of \PolyVGR.}

\begin{Newstuff}{\NewstuffPadding}
  \begin{gather}\label{eq:5}
    \forall (\SVar : \KSession ).~
    \forall (\Nchan : \KDomain{\TShapeOne}).\notag\\
    \{ \Nchan : \Inp\Int\Inp\Int\Outp\Int
    \SVar \} ; \Chan\Nchan \to \exists \cdot.\{ \Nchan : \SVar \};
    \Unit
  \end{gather}
\end{Newstuff}
Quantification over session types, as in $\forall (\SVar:\KSession)$,
has been considered and analyzed in other recent 
work
\cite{lindley17:_light_funct_session_types,ALMEIDA2022104948}. 

The quantification of  $\Nchan$ is novel to \PolyVGR. Its kind,
$\KDomain{\TShapeOne}$, indicates that $\Nchan$ ranges over all
channel identities. We call $\TShapeOne$ a \emph{shape}. Shapes 
allow us to talk about and quantify over the (channel) resources embedded in a
value in \PolyVGR. For example, $\TShapePair\TShapeOne\TShapeOne$ is
the shape to describe a value with two embedded channels.
This facility enables \PolyVGR to provide a single typing rule for the
operations \lstinline/receive/ and \lstinline/send_on_/: In VGR, there
are two separate typing rules, one to transmit data values
and another to transmit one single channel.
Shapes also facilitate an extension of \PolyVGR with algebraic
datatypes like lists, which was not considered in previous work.

A final ingredient of the function type in \PolyVGR is the innocuous
existential right of the function arrow in~(\ref{eq:5}). The
existential addresses another shortcoming exhibited by this VGR type:
\begin{gather}\label{eq:8}
  \{\}; \Unit \to \Chan\Nchan; \{ \Nchan : S \}
\end{gather}
A function of type~(\ref{eq:8}) must create a new channel of session type
$S$. But lacking polymorphism, each invocation of this function has to
create a channel with the same identity $\Nchan$. To avoid this
limitation, \PolyVGR handles newly created channels (and other
resources) using existential quantification:
\begin{Newstuff}{\NewstuffPadding}
  \begin{gather}\label{eq:13}
    \{\}; \Unit \to \exists (\Nchan: \KDomain\TShapeOne).~ \{ \Nchan :
    S \}; \Chan\Nchan
  \end{gather}
\end{Newstuff}
Every use of a function of this type gives rise to a new channel
identity. Thanks to the existential, this identity is renamed as
needed to avoid clashes with any existing channel identity in the context. 

Coming back to the examples, let us have a look at function
\lstinline/server'/ in Listing~\ref{lst:example-server-unit}. This
function contains a free variable \lstinline/u/ with a channel
reference of type $\Chan\Nchan$. It can only be used in a context that
provides the same channel $\Nchan$, which is somewhat hidden in the
VGR type of \lstinline/server'/:
\begin{gather}\label{eq:2}
  \{ \Nchan : \Inp\Int\Inp\Int\Outp\Int S \} ; \Unit \to \Unit;
  \{ \Nchan : S \} \text,
\end{gather}
but which becomes very clear in its \PolyVGR type:
\begin{Newstuff}{\NewstuffPadding}
  \begin{gather}\label{eq:14}
    \forall (\SVar : \KSession ).\notag\\
    \{ \Nchan : \Inp\Int\Inp\Int\Outp\Int \SVar \} ; \Unit \to
    \exists\cdot.~\{ \Nchan : \SVar \}; \Unit \text.
  \end{gather}
\end{Newstuff}
The lack of quantification over $\Nchan$ indicates that it is not safe to use this function with any other
channel, because it is not possible to replace a channel reference
captured in the closure for \lstinline/server'/.  In any case, we can
invoke a function of type~(\ref{eq:2}) or  of type~(\ref{eq:14}) any
time the channel $\Nchan$ is in a state matching the ``before'' session type of the
function.

\subsection*{Contributions}
\label{sec:contributions}

\begin{itemize}
\item We define \PolyVGR, a novel session type system based on
  polymorphic typestate that lifts all restrictions imposed by earlier related
  systems, but still operates on the basis of the same semantics. Our type
  system exhibits a novel use of higher-kinded polymorphism
  to enable quantification over types that contain an a-priori unknown
  number of channel references.
\item We establish type preservation and progress for {\PolyVGR} on
  the basis of a standard synchronous semantics for session types (see~\Cref{sec:metatheory}).
\item Type checking for {\PolyVGR} is decidable and implemented (see~\Cref{sec:implementation}). We plan to submit
  the implementation for artifact evaluation.
\item We informally sketch an extension of {\PolyVGR} for sum types that may contain channels (see~\Cref{sec:extensions}).
\end{itemize}
Proofs and some extra examples may be found in the supplement.

\section{Motivation}
\label{sec:motivation}

We demonstrate how polymorphism in the form of universal and existential quantification lifts
various restrictions of the VGR calculus. In particular, VGR is
monomorphic with respect to channel names and states and it
requires different operations (with different types) to transmit data
and (single) channels. All these restrictions disappear in
{\PolyVGR}. Moreover, universal and existential quantification gives us fine grained
control over channel identity management, channel passing between
processes as well as channel creation.

The next few subsections systematically explain the innovations of
\PolyVGR compared to VGR. Types 
and code fragments for the new calculus  {\PolyVGR} appear in
boxes with a frame.
\PolyVGR{} offers the following key benefits over previous work.
\begin{itemize}
\item A function can be applied to different channel arguments if
  its type is polymorphic over channel names (see~(\ref{eq:5})
  and~(\ref{eq:2}); Section~\ref{sec:channel-abstraction}).
\item A function can abstract over the creation of an arbitrary number
  of channels because the names of newly created channels are
  existentially quantified (see (\ref{eq:7}) and (\ref{eq:10}); Section~\ref{sec:channel-creation}).
\item Arbitrary data structures can be transmitted. Ownership of all channels
  contained in the data structure is transferred to the receiver (see
  (\ref{eq:6}); Section~\ref{sec:data-transmission-vs}).
\item Abstraction over transmission operations is possible. In
  particular, a type can be given to a fully flexible send or receive
  operation (see (\ref{eq:6})).
\end{itemize}


\subsection{Channel Creation}
\label{sec:channel-creation}

Channel creation in VGR works in two steps. First, we create
an \emph{access point} of type $[S]$, where $S$ is a session type. This
access point needs to be known to
all threads that wish to communicate and it can be shared
freely. Second, the client thread requests a connection on the access point and the server must
accept it on the same access point. This rendezvous creates a
communication channel with one end of type $S$ on the server and the
other end of type $\Dual S$ (the dual type of $S$) on the client.
\begin{mathpar}
  \ruleVGRAccept

  \ruleVGRRequest
\end{mathpar}
In the VGR typing rules for these operations, new channels just show up with a fresh name in the outgoing
state of the expression typing. Similarly, if a function of type
$\Sigma_1; T_1 \to T_2; \Sigma_2$ creates a new channel, then its
name and session type just appear in $\Sigma_2$. Incoming
channels described in $\Sigma_1$ are either passed through to
$\Sigma_2$ or they are closed in the function. All channels mentioned
in $\Sigma_2$, but not in $\Sigma_1$ are considered new.

As the channel names in states must all be different, the number
of simultaneously open channels in a VGR
program is bounded by the number of
occurrences of the \TirName{C-Accept} and \TirName{C-Request}
rules. VGR has recursive functions, but they are monomorphic with
respect to incoming and outgoing states.  In consequence, abstraction over channel creation is not possible.

In contrast, {\PolyVGR}'s function type indicates channel creation explicitly using existential
quantification. As an example, consider abstracting over the accept
operation:
\begin{Newstuff}{\NewstuffPadding}
  \begin{gather}
    \label{eq:7}
    \begin{array}[t]{r@{}l}
      \mathit{acc}={}&
      \ETAbs\SVar\KSession{\CstrTrue}
      \EAbs{\StEmpty}\EVar{\TAccessPoint\SVar}\Accept\EVar
      \\
      \colon&
      \TAll\SVar\KSession{\CstrTrue}
      \TArr \StEmpty
      {\TAccessPoint\SVar} {\TVarc\colon \KDomain \TShapeOne} {
        \StBind\TVarc\SVar} {\TChan\TVarc}
    \end{array}
  \end{gather}
\end{Newstuff}
The core of this type still has a shape like the VGR type $\Sigma_1; T_1 \to T_2;
\Sigma_2$, but with some additions and changes. The most prominent
change is that the outgoing type and
state are swapped in a function type resulting in a structure
like this:
\begin{Newstuff}{\NewstuffPadding}
  \begin{gather}
    \TArr{\Sigma_1}{T_1}{\TVar\colon \KDomain
    \TVarShape}{\Sigma_2}{T_2}.
  \end{gather}
\end{Newstuff}
The incoming state $\Sigma_1$
specifies the part of the state that is needed
by the function; it can be applied in any state $\Sigma$ that
provides the required channels or more. 
On return, a function can provide new entries in the state, which are disjointly
added to the calling state, by way of the existential $\exists \TVar : \KDomain \TVarShape$.

The type of \textit{acc} in (\ref{eq:7}) is universally quantified
over a session type, $\sigma:\KSession$, to work with arbitrary access points.
Left of the arrow, the required incoming state is empty $\StEmpty$
and argument of type $\TAccessPoint\SVar$ is an access point for $\SVar$.
Right of the arrow,
the existential quantification $\exists \TVarc : \KDomain \TShapeOne$
abstracts over the created channel. The 
kind $\KDomain \TShapeOne$ indicates abstraction over exactly one channel
name.\footnote{%
  We defer further discussion of other shapes $\TVarShape$ and the meaning of
  $\KDomain \TVarShape$ to Sections~\ref{sec:two-channels} and~\ref{sec:general-case}.}
Hence, the variable $\TVarc$ can be used like a channel
name in constructing a state. The returned value is a channel reference for
$\TVarc$.
The existential serves as a modular alternate of the $\Fresh
c$ constraint. So we can invoke \textit{acc} multiple times and obtain
different channels from every invocation.

\subsection{Channel Abstraction}
\label{sec:channel-abstraction}

The discussion of VGR's function type $\Sigma_1;T_1 \to
T_2;\Sigma_2$ in the introduction shows that a function that takes a
channel as a parameter can only be applied to a single channel. A
function like \lstinline/server/ (\Cref{lst:example-server}) must be applied to the channel of
type $\Chan\Nchan$, for some fixed name $\Nchan$.

To lift this restriction, we apply the standard recipe of universal
quantification, i.e., polymorphism over channel identities as outlined
in the introduction. Thus, the type of \lstinline/server/ generalizes
as shown in~(\ref{eq:5}) so that it can be applied to any channel of type $\Chan\Nchan$ regardless of the name
$\Nchan$ and the type of \lstinline/server'/, which captures a
channel, is shown in~(\ref{eq:2}).

\subsection{Data Transmission vs. Channel Transmission}
\label{sec:data-transmission-vs}

VGR can pass channels
from one thread to another. The session type $\Outp{S'}{S}$ classifies
a channel on which we can send a channel of type $S'$. Here is the VGR
typing rule for the underlying operation:
\begin{mathpar}
  \ruleVGRSendS
\end{mathpar}
The premises are \emph{value typings} that indicate that $v$ and $v'$ are
references to fixed channels $\beta$ and $\alpha$ under variable environment
$\Gamma$. The conclusion is an
\emph{expression typing} of the form $\JVGRExpr\Gamma\Sigma
e{\Sigma_1}{T}{\Sigma_2}$ where $\Sigma$ is the incoming
state, $\Sigma_1$ is the part of $\Sigma$ that is passed
through without change, and $\Sigma_2$ is the outgoing  state after
executing expression $e$ which returns a result of type $T$. 
The rule states that channels $\beta$ and $\alpha$ have session type
$S'$ and $\Outp{S'}{S}$, respectively. The channel $\beta$ is consumed
(because it is sent to the other end of channel $\alpha$) and $\alpha$
gets updated to session type $S$.

Compared to the function type, sending a channel is more flexible.
Any channel of type $S'$ can be passed because $\beta$ is not part
of channel $\alpha$'s session type. If the sender still holds
references to channel $\beta$, then these references can no longer be
exercised as $\beta$ has been removed from $\Sigma$.
So one can say that rule \TirName{C-SendS} passes ownership of channel
$\beta$ to the receiver.

In addition, VGR implicitly transmits a channel reference which is captured
in a closure. To study this phenomenon, we look
at VGR's rules for sending and receiving data of type $D$. 
\begin{mathpar}
  \ruleVGRSendD

  \ruleVGRReceiveD
\end{mathpar}
One possibility for type $D$ is a function type like
$$D_1 = \{\beta: S'\}; \Unit \to \Unit; \{\beta:S''\}.$$ A function of this type
captures a channel named $\beta$ which may or may not occur in
$\Sigma$. It is instructive to see what happens at the receiving end
in rule \TirName{C-ReceiveD}.
If we receive a function of type $D_1$ and $\Sigma$ already contains
channel $\beta$ of appropriate session type, then we will be able to
invoke the function.

If channel $\beta$ is not yet present at the receiver, we may want to send it along
later. However, we find that this is not possible as the received channel
gets assigned a fresh name $d$:
\begin{mathpar}
  \ruleVGRReceiveS
\end{mathpar}
For the same reason, it is impossible to send channel $\beta$ first
and then the closure that refers to it: $\beta$ gets renamed to some
fresh $d$ while the closure still refers to $\beta$.  Sending the channel
effectively cuts all previous connections.

To address this issue, \PolyVGR abstracts over states in session
types and lifts all restrictions on the type of transmitted
values (aka the \emph{payload type}), so that a channel and a function that refers to it can be
transmitted at the same time. Here is the revised grammar of session
types: 
\begin{Newstuff}{\NewstuffPadding}
  \begin{align*}
          \Ses & ::= \SSend\TVar{\KDomain\POLYShape}\St\Typ\Ses
                 \mid \SRecv\TVar{\KDomain\POLYShape}\St\Typ\Ses
                 \mid \dots
  \end{align*}
\end{Newstuff}
A channel package can be instantiated by a state $\Sigma$ and a
payload type $T$. All channels referenced in $T$ must be bound in
$\Sigma$ so that  the sending and the receiving end of the channel
agree about the channels sent along with the value of type $T$. That is,
sending any value that contains channel references also transfers the
underlying referenced channels to the receiver. Thus, sending a
reference transfers ownership of the underlying channel. Moreover, a value may contain several
channel references.

The ``size'' of $\St$ is gauged with $\TVar$ which determines its
domain as indicated in its kind $\KDomain\POLYShape$ where
$\POLYShape$ is the shape of the domain. Shapes range over
$\TShapeZero$ (the empty shape), $\TShapeOne$ (the shape with one
binding), and $\TShapePair{\POLYShape_1}{\POLYShape_2}$ which
forms the disjoint combination of shapes $\POLYShape_1$ and $\POLYShape_2$.

\subsubsection{No channels}
\label{sec:no-channels}

To gain some intuition with this type construction, we start with a
type for sending a primitive value of type $\Int$.  In the general
pattern $\SSend\TVar{\KDomain\POLYShape}\St\Typ\Ses$ we find that
\begin{itemize}
\item $\Typ = \Int$;
\item $\St = \StEmpty$, the empty state, as an $\Int$ value
  contains no channels;
\item the type variable $\TVar$ specifies the domain of $\StEmpty$,
  which is also empty, indicated with $\POLYShape = \TShapeZero$.
\end{itemize}

Here is the resulting term and type, where we quantify over a
continuation session type $\SVar$ and a channel name $\TVarChan$ (its
kind $\KDomain\TShapeOne$ indicates that it is a single channel):
\begin{Newstuff}{\NewstuffPadding}
  \begin{gather}\label{eq:9}
    \begin{array}{r@{\hspace{2pt}}l}
      \hspace{-1mm}\mathit{send0} = &
      \ETAbs\TVarChan{\KDomain\TShapeOne}{\CstrTrue}
      \ETAbs\SVar\KSession{\CstrTrue}
      \\&
      \EAbs{\StEmpty}\EVar{\Int}
      \EAbs{
      \StBind\TVarChan{\SSend\TVar{\KDomain\TShapeZero}{\StEmpty}{\Int}\SVar}}\EVarY{\TChan\TVarChan}
      \\&
      \ESend \EVar \EVarY
      \\
      :
      &
      \TAll\TVarChan{\KDomain\TShapeOne}{\CstrTrue}
      \TAll\SVar\KSession{\CstrTrue}
      \\&
      \TArrBeginX{\StEmpty}{\Int}{\Empty}{\StEmpty}\\&
      \TArrBeginX
        {\StBind\TVarChan{\SSend\TVar{\KDomain\TShapeZero}{\StEmpty}{\Int}\SVar}}
        {\TChan\TVarChan}
        {\Empty}
        {\StBind\TVarChan\SVar}
        {\TUnit}
      \TArrEndX
      \TArrEndX
    \end{array}
  \end{gather}
\end{Newstuff}

\subsubsection{One channel}
\label{sec:one-channel}

We instantiate the general pattern
$$\SSend\TVar{\KDomain\POLYShape}\St\Typ\Ses$$
as follows to send a channel of type $\Ses'$.
\begin{itemize}
\item $\St$ is now a state with a single binding, so that
  $\TVar$ must range over $\KDomain\TShapeOne$;
\item consequently, $\St$ has the form $\StBind\TVar{\Ses'}$; and
\item $\Typ = \TChan\TVar$;
\end{itemize}
We omit the term, which is similar to the one in \eqref{eq:9}, and
just spell out the type. We quantify over the continuation session
type and the names of the two channels involved. There is one novelty:
we declare that the channels $\TVar$ and $\TVarChan$ are
different, so that 
they can be used as keys in the state. The \emph{disjointness
  constraint} $(\TVar \Disjoint \TVarChan)$ 
specifies that names in $\TVar$ are disjoint from names in $\TVarChan$.
\begin{Newstuff}{\NewstuffPadding}\small
  \begin{gather*}
    \begin{array}{r@{}l}
      \mathit{send1}\colon &
      \TAll\TVar{\KDomain\TShapeOne}{\CstrTrue}
      \TAll\TVarChan{\KDomain\TShapeOne}{(\TVar \Disjoint \TVarChan)}
      \TAll\SVar\KSession{\CstrTrue}
      \\&
      \TArrBeginX{\StEmpty}{\TChan\TVar}{\Empty}{\StEmpty}
      \\&
      \TArrBegin
        {\StBind\TVar{\Ses'},
         \StBind\TVarChan{\SSend\TVar{\KDomain\TShapeOne}{\StBind\TVar{\Ses'}}{\TChan\TVar}\SVar}}
        {\TChan\TVarChan}
      \\&\hspace{1.3mm}
      \TArrEnd
        {\Empty}
        {\StBind\TVarChan\SVar}
        {\TUnit}
      \TArrEndX
    \end{array}
  \end{gather*}
\end{Newstuff}

\subsubsection{Two channels}
\label{sec:two-channels}

Sending two channels of type $\Ses'$ and $\Ses''$ requires new ingredients and illustrates the
general case. The instantiation of the pattern   $\SSend\TVar{\KDomain\POLYShape}\St\Typ\Ses$ is as follows:
\begin{itemize}
\item the state $\St$ must have two bindings, one for each
  payload channel, so that $\TVar$ must range over a two element
  domain, e.g., $\KDomain{\TShapePair\TShapeOne\TShapeOne}$;
\item to write down $\St$, we need notation to address the
  $\TShapeOne$-shaped components of $\TVar$ as in $\TDomProj1\TVar$
  and $\TDomProj2\TVar$, so that we have $\St =
  \StBind{\TDomProj1\TVar}{\Ses'}, \StBind{\TDomProj2\TVar}{\Ses''}$;
\item to send a pair of channels: $\Typ
  = \TChan{(\TDomProj1\TVar)} \times \TChan{(\TDomProj2\TVar)}$.
\end{itemize}
\begin{Newstuff}{\NewstuffPadding}
  \begin{gather*}
    \begin{array}{l}
      \mathbf{let}\ \Ses (\SVar) =
      \SSendBegin
        \TVar
        {\KDomain{\TShapePair\TShapeOne\TShapeOne}}
      \SSendMid
        {\StBind{\TDomProj1\TVar}{\Ses'},
         \StBind{\TDomProj2\TVar}{\Ses''}}
      \\
      \hspace{39mm}\SSendEnd
        {\TChan{(\TDomProj1\TVar)} \times
         \TChan{(\TDomProj2\TVar)}}
        \SVar
      \ \mathbf{in}\\
      \TAll\TVara{\KDomain\TShapeOne}{\CstrTrue}
      \TAll\TVarb{\KDomain\TShapeOne}{(\TVara\Disjoint\TVarb)}
      \\
      \TAll\TVarChan{\KDomain\TShapeOne}{(\TVara\Disjoint\TVarChan, \TVarb\Disjoint\TVarChan)}
      \TAll\SVar\KSession{\CstrTrue}
      \\
      \TArrBeginX{\StEmpty}{\TChan\TVara \times \TChan\TVarb}{\Empty}{\StEmpty}
      \\
      \TArr{\StBind\TVara{\Ses'}, \StBind\TVarb{\Ses''},
      \StBind\TVarChan{\Ses (\SVar)}}{\TChan\TVarChan}{\Empty}{\StBind\TVarChan\SVar}{\TUnit}
      \TArrEndX
    \end{array}
  \end{gather*}
\end{Newstuff}
We use the let-notation informally to improve readability. It is not part of the type system. Close study
of the type reveals a discrepancy between the ``curried'' way to pass the arguments
$\TVara\colon\KDomain\TShapeOne$ and $\TVarb\colon\KDomain\TShapeOne$ and the ``uncurried'' kind
$\KDomain{\TShapePair\TShapeOne\TShapeOne}$
expected by the existential. To rectify this discrepancy, the term pairs the two domains to obtain
some $\TVarc = (\TVara, \TVarb)$ with
$\TVarc \colon \KDomain{\TShapePair\TShapeOne\TShapeOne}$ as needed for the existential. This
definition of $\TVarc$ implies that $\TVara = \TDomProj1\TVarc$ and $\TVarb =
\TDomProj2\TVarc$ which are needed to obtain the correct state and type
for  the body of the existential.

\subsubsection{The general case}
\label{sec:general-case}

In general a value can refer to an arbitrary number of
channels, which  should not be fixed  a priori. We exhibit and discuss the type of a general send
function \textit{gsend} and show how to obtain the previous examples by instantiation.
\begin{Newstuff}{\NewstuffPadding}
  \begin{gather}\label{eq:6}
    \begin{array}{@{}r@{~}l@{}}
      \mathit{gsend} \colon
      &
      \TAll\TVarShape\KShape{\CstrTrue}
      \TAll\TVar{\KDomain\TVarShape}{\CstrTrue}
      \\&
      \TAll\TVarStFun{\KStFun\TVarShape}{\CstrTrue}
      \TAll\TVarTFun{\KTypeFun\TVarShape}{\CstrTrue}
      \\&
      \TAll\TVarChan{\KDomain\TShapeOne}{(\TVar\Disjoint\TVarChan)}
      \TAll\SVar\KSession{\CstrTrue}
      \\&
      \TArrBeginX{\StEmpty}{\TVarTFun\ \TVar}{\Empty}{\StEmpty}
      \\&
      \TArrBeginX{\TVarStFun\ \TVar,
      \StBind\TVarChan{\SSend\TVar{\KDomain\TVarShape}{\TVarStFun\
      \TVar}{\TVarTFun\ \TVar}\SVar}}{\TChan\TVarChan}{\Empty}{\StBind\TVarChan\SVar}
      \\&
      {\TUnit}
      \TArrEndX
      \TArrEndX
    \end{array}
    \raisetag{10pt}
  \end{gather}
\end{Newstuff}
We abstract over the shape, $\TVarShape$, and the corresponding domain. As the state
depends on the domain $\TVar$, we supply it as a closed function $\TVarStFun$ from the domain so
that its components can
only be constructed from the domain elements. We supply the type in the same way as a closed
function $\TVarTFun$ from the domain. The remaining quantification over the channel name and the continuation
session is as usual. The disjointness constraint forces the channel name to be
different from any name in $\TVar$.
In the body of the type we have a function that takes an argument of type  $\TVarTFun\ \TVar$. It
returns a function that takes a channel $\TVarChan$ along with the resources provided by the state
$\TVarStFun\ \TVar$. It returns the updated channel type and removes the resources which are on the
way to the receiver.

The previous examples correspond to the following instantiations of \textit{gsend}:
\begin{itemize}
\item $\mathit{send0} = \mathit{gsend}\ \TShapeZero\ \TDomZero\
  (\lambda\_. \StEmpty)\ (\lambda\_. \Int)$ \\
  where $\TDomZero\colon\KDomain\TShapeZero$ is the
  unique value of this type;
\item $\mathit{send1} =
  \begin{array}[t]{@{}l}
    \ETAbs{\TVar}{\KDomain\TShapeOne}\CstrTrue
    \\
    \mathit{gsend}\ \TShapeOne\ \TVar\
    (\lambda\TVar. \StBind\TVar{\Ses'})\ (\lambda\TVar. \TChan\TVar)
    \text{;}
  \end{array}
$
\item $\mathit{send2} =
  \begin{array}[t]{@{}l}
  \ETAbs{\TVara}{\KDomain\TShapeOne}\CstrTrue
  \ETAbs{\TVarb}{\KDomain\TShapeOne}\CstrTrue
    \\
  \mathit{gsend}\
    \begin{array}[t]{@{}l}
      (\TShapePair\TShapeOne\TShapeOne)\ (\TVara, \TVarb)\
    (\lambda\TVarc. \StBind{\TDomProj1\TVarc}{\Ses'}, \StBind{\TDomProj2\TVarc}{\Ses''})\
      \\
    (\lambda\TVarc. \TChan{(\TDomProj1\TVarc)} \times \TChan{(\TDomProj2\TVarc)})
    \end{array}
  \end{array}
  $
\end{itemize}

\section{Formal Syntax and Semantics of \PolyVGR}
\label{sec:polym-imper-sess}
\begin{figure}[tp]
  \begin{align*}
    & \mathrm{Kinds} &
    \Kind ::=\ &
      \KType \mid
      \KSession \mid
      \KSt \mid
      \KShape \mid
      \\&&&
      \KDomain\POLYShape \mid
      \Kind \to \Kind
    \\
    & \mathrm{Labels} &
    \Label ::=\ &
      \LabelA \mid
      \LabelB
    \\
    & \mathrm{Types} &
    \Typ,\Ses,\POLYShape,\POLYDom,\St ::=\ &
      \TVar \mid
      \Typ~\Typ \mid
      \TLam \TVar \POLYShape \Typ \mid
      \\ & \hspace{5mm} \mathrm{Expression\ Types} \hspace{-40mm}&&
      \TAll\TVar\Kind\Cstr\Typ \mid
      \TArr\St\Typ\Ctx\St\Typ \mid
      \\ &&&
      \TChan\POLYDom \mid
      \TAccessPoint\Ses \mid
      \TUnit \mid
      \TPair\Typ\Typ \mid
      \\ & \hspace{5mm} \mathrm{Session\ Types} \hspace{-40mm} &&
      \SSend\TVar{\KDomain\POLYShape}\St\Typ\Ses \mid
      \\ &&&
      \SRecv\TVar{\KDomain\POLYShape}\St\Typ\Ses \mid
      \\ &&&
      \SChoice\Ses\Ses \mid
      \SBranch\Ses\Ses \mid
      \SEnd \mid
      \Dual\Ses \mid
      \\ & \hspace{5mm} \mathrm{Shapes} &&
      \TShapeZero \mid
      \TShapeOne \mid
      \TShapePair\POLYShape\POLYShape \mid
      \\ & \hspace{5mm} \mathrm{Domains} &&
      \TDomZero \mid \POLYDom,\POLYDom \mid \TDomProj \Label \POLYDom \mid
      \\ & \hspace{5mm} \mathrm{Session\ State} \hspace{-40mm} &&
      \Empty \mid
      \StBind\POLYDom\Ses \mid
      \St,\St
    \\
    & \mathrm{Type\ Environments} \hspace{-40mm} &
    \Ctx ::=\ &
      \Empty \mid
      \Ctx, \EVar:\Typ \mid
      \Ctx, \TVar:\Kind \mid
      \Ctx, \POLYDom \Disjoint \POLYDom
    \\
    & \mathrm{Constraints} \hspace{-40mm} &
    \Cstr ::=\ &
      \Empty \mid
      \Ctx, \POLYDom \Disjoint \POLYDom
    \\
    & \mathrm{Expressions} \hspace{-40mm} &
    \Exp ::=\ &
      \Val \mid
      \ELet\EVar\Exp\Exp \mid
      \EApp\Val\Val \mid
      \EProj \Label \Val \mid
      \ETApp\Val\Typ \mid
    \\ &&&
      \EFork\Val \mid
      \ENew\Ses
      \EAccept\Val \mid
      \ERequest\Val \mid
      \\ &&&
      \ESend\Val\Val \mid
      \ERecv\Val \mid
      \ESelect \Label \Val \mid
      \\ &&&
      \ECase\Val\Exp\Exp \mid
      \EClose\Val
    \\
    & \mathrm{Values} &
    \Val ::=\ &
      \EVar \mid
      \EChan\TVar \mid
      \EUnit \mid
      \EPair\Val\Val \mid
      \\ &&&
      \EAbs\St\EVar\Typ\Exp \mid
      \ETAbs\TVar\Kind\Cstr\Val 
    \\
    & \mathrm{Configurations} \hspace{-40mm} &
    \Cfg ::=\ &
      \CExp\Exp \mid
      (\CPar\Cfg\Cfg) \mid
      \CBindChan\TVar\TVar\Ses\Cfg \mid
      \CBindAP\EVar\Ses\Cfg
    \\
    & \mathrm{Expression\ Contexts} \hspace{-40mm} &
    \ECtxSym ::=\ &
      \Hole \mid
      \ELet\EVar\ECtxSym\Exp
    \\
    & \mathrm{Configuration\ Contexts} \hspace{-40mm} &
    \CCtxSym ::=\ &
      \Hole \mid
      \CBindChan\TVar\TVar\Ses\CCtxSym \mid
      \CBindAP\EVar\Ses\CCtxSym \mid
      (\CPar\CCtxSym\Cfg)
  \end{align*}
  \caption{Syntax of \PolyVGR}
  \label{fig:polyvgr-syntax}
\end{figure}


\subsection{Syntax}
\label{sec:syntax}

Figure~\ref{fig:polyvgr-syntax} defines the syntax of \PolyVGR starting with kinds
and types. Different metavariables for types indicate their
kinds with $\Typ$ as a fallback.
Kinds $\Kind$ distinguish between plain types ($\KType$), session
types ($\KSession$ ranged over by metavariable $\Ses$), states
($\KSt$ ranged over by $\St$), shapes
($\KShape$ ranged over by $\POLYShape$), domains ($\KDomain\POLYShape$
ranged over by $\POLYDom$), and arrow kinds.
The kind for domains depends on shapes. This
dependency as well as the introduction rules for arrow kinds are very
limited as they are tailored to express channel references as
discussed in Section~\ref{sec:data-transmission-vs}.

The type language comprises variables $\TVar$, application, and
abstraction over domains to support arrow kinds. Universal
quantification over types of any kind is augmented with constraints
$\Cstr$, function types contain pre- and post-states as well as existential quantification as
explained in Section~\ref{sec:channel-creation}. There are channel
references that refer to a domain, access points that refer to a
session type, the unit type (representing base types), and products to
characterize the values of expressions.
Session type comprise sending and receiving (cf.\
Section~\ref{sec:data-transmission-vs}), as well as choice and branch
types limited to two alternatives, $\SEnd$ to indicate the end of a protocol, and $\Dual\Ses$ to
indicate the dual of a session type (which flips sending and receiving
operations as well as choice and branch).
Shapes comprise the empty shape $\TShapeZero$, the single-channel shape
$\TShapeOne$, and the combination of two shapes
$\TShapePair\_\_$. The corresponding domains are the
empty domain $\TDomZero$, the combination of two domains $\_,\_$, and
the first/second projection of a domain. The latter selects a
component of a combined domain.
A session state can be empty, a binding of a
single-channel domain to a session type, or a combination of states. Most
of the time, the domain in the binding is a variable.

Type environments $\Ctx $ contain bindings for expression variables
and type variables, as well as disjointness constraints between
domains.
Constraints $\Cstr$ are type environments restricted to
bindings of disjointness constraints.

Following VGR \cite{DBLP:journals/tcs/VasconcelosGR06}, the expression
language is presented in A-normal form~\cite{DBLP:conf/pldi/FlanaganSDF93}, which means that
the subterms of each non-value expression are syntactic values $\Val$
and sequencing of execution is expressed using a single
$\Terminal{let}$ expression. This choice simplifies the dynamics as
there is only one kind of evaluation context $\ECtxSym$, which selects
the header expression of a  $\Terminal{let}$. The type system performs
best (i.e., it is most permissive) on expressions in strict A-normal
form, where the body of a $\Terminal{let}$ is either another
$\Terminal{let}$ or a syntactic value. Any expression can be
transformed into strict A-normal form with a simple variation of the
standard transformation from the literature. Strict A-normal form
is closed under reduction.

Besides values and the $\Terminal{let}$ expression, there is function
application, projecting a pair, type application,  fork to
start processes, accepting and requesting a channel, sending and receiving,
selection (i.e.\ sending) of a label and branching on a received
label, and closing a channel.

Values are variables, channel references, the unit value, 
pairs of values, lambda abstractions, and type
abstractions with constraints --- their body is restricted to a syntactic value
to avoid unsoundness in the presence of effects.

Configurations $\Cfg$ describe processes. They are either expression
processes, parallel processes, channel abstraction --- it abstracts
the two ends of a channel at once, and access point creation.

We already discussed expression contexts. Configuration contexts $\CCtxSym$ enable
reduction in any configuration context, also under channel
and access point abstractions.

\subsection{Statics for types}
\label{sec:statics-types}

Many of the judgments defining the type-level statics are mutually
recursive. We start with
\begin{itemize}
\item context formation $\PVGRIsCtx\Ctx$,
\item kind formation $\PVGRIsKind\Ctx\Kind$,
\item type formation $\PVGRHasKind\Ctx\Typ\Kind$.
\end{itemize}
All judgments depend on context formation, which depends on kind and
type formation. Based on these notions we define
\begin{itemize}
\item type conversion $\PVGRTypeConv\Typ\Typ$,
\item constraint entailment $\PVGRCstrEntail\Ctx\Cstr$,
\item context restriction operators $\CtxRestrictNonDom\Ctx$ and $\CtxRestrictOnlyDom\Ctx$,
\item disjoint context extension operator $\Ctx\DisjointAppend\Ctx$.
\end{itemize}

\begin{figure}[tp]
  \begin{mathpar}
    \rulePVGRIsCtxEmpty \and
    \rulePVGRIsCtxConsKind \and
    \rulePVGRIsCtxConsType \and
    \rulePVGRIsCtxConsCstr
  \end{mathpar}
  \caption{Context formation ($\PVGRIsCtx\Ctx$)}
  \label{fig:polyvgr-context-formation}
\end{figure}

Context formation (Figure~\ref{fig:polyvgr-context-formation}) is
standard up to the case for disjointness constraints. For those, we
have to show that each domain is wellformed with respect to the
current context $\Ctx$, which may be needed to construct the shape and
then the domain.

\begin{figure}[tp]
  \begin{mathpar}
    \rulePVGRIsKindType \and
    \rulePVGRIsKindSession \and
    \rulePVGRIsKindState \and
    \rulePVGRIsKindShape \and
    \rulePVGRIsKindDom \and
    \rulePVGRIsKindArr
  \end{mathpar}
  \caption{Kind formation ($\PVGRIsKind\Ctx\Kind$)}
  \label{fig:polyvgr-kind-formation}
\end{figure}

Kind formation is in Figure~\ref{fig:polyvgr-kind-formation}. Most
kinds are constants, domains must be indexed by 
shapes, arrow kinds are standard.

\begin{figure}[tp]
  \begin{mathpar}
    \rulePVGRKindingVar \and
    \rulePVGRKindingApp \and
    \rulePVGRKindingLam \and
    \rulePVGRKindingAll \and
    \rulePVGRKindingChan \and
    \rulePVGRKindingArr \and
    \rulePVGRKindingAccessPoint \and
    \rulePVGRKindingUnit \and
    \rulePVGRKindingPair \and
    \rulePVGRKindingSend \and
    \rulePVGRKindingRecv \and
    \rulePVGRKindingBranch \and
    \rulePVGRKindingDual \and
    \rulePVGRKindingChoice \and
    \rulePVGRKindingEnd
  \end{mathpar}
  \caption{Type formation, Part I ($\PVGRHasKind\Ctx\Typ\Kind$)}
  \label{fig:polyvgr-kinding-I}
\end{figure}
\begin{figure}[tp]
  \begin{mathpar}
    \rulePVGRKindingShapeZero \and
    \rulePVGRKindingShapeOne \and
    \rulePVGRKindingShapePair \and
    \rulePVGRKindingDomZero \and
    \rulePVGRKindingDomMerge \and
    \rulePVGRKindingDomProj \and
    \rulePVGRKindingStEmpty \and
    \rulePVGRKindingStChan \and
    \rulePVGRKindingStMerge
  \end{mathpar}
  \caption{Type formation, Part II ($\PVGRHasKind\Ctx\Typ\Kind$)}
  \label{fig:polyvgr-kinding-II}
\end{figure}
 Figures~\ref{fig:polyvgr-kinding-I} and~\ref{fig:polyvgr-kinding-II}
contain the rules for type formation and kinding.
The rules for variables and application are standard.
Abstractions (rule \TirName{K-Lam}) are severely restricted. Their
argument must be a domain and their result must be $\KType$ or
$\KShape$. Moreover, the body can only refer to the argument domain;
all other domains are removed from the assumptions.
Constrained universal quantification (rule \TirName{K-All}) is
standard.

To form a function type, rule \TirName{K-Arr} asks that the argument
state and type are wellformed with respect to the assumptions. The
return state and type must be wellformed with respect to the assumptions
extended with the state $\Ctx_2$ of channels created by the function. This
state must be disjoint from the assumptions as indicated by
$\Ctx_1\DisjointAppend\Ctx_2$ (see
Figure~\ref{fig:polyvgr-disjoint-ctx-extension}). We also make sure
that $\Ctx_2$ only contains domains.

A channel type can be formed from any single-channel domain of shape
$\TShapeOne$ (rule \TirName{K-Chan}). The rules for access points,
unit, and pairs are straightforward and standard.

The rule \TirName{K-Send} and \TirName{K-Recv} control wellformedness
of sending and receiving types. In both cases, we require that both
the state and the type describing the transmitted value can only
reference the domain abstracted in the existential. This restriction
is necessary to enforce proper transfer of channel ownership between
sender and receiver.

The remaining rules for session types are standard.

Figure~\ref{fig:polyvgr-kinding-II} contains the rules for shapes,
domains, and states. We discussed shapes with their syntax
already. The domain rules are similar to product rules with the
additional disjointness constraint on the components of the combined
domain.
Empty states are trivially wellformed. A single binding is wellformed
if it maps a single-channel domain to a session type.

\begin{figure}[tp]
  \begin{mathpar}
    \rulePVGRTypeConvTApp \and
    \rulePVGRTypeConvProj \and
    \rulePVGRTypeConvDualEnd \and
    \rulePVGRTypeConvDualVar \and
    \rulePVGRTypeConvDualSend \and
    \rulePVGRTypeConvDualRecv \and
    \rulePVGRTypeConvDualChoice \and
    \rulePVGRTypeConvDualBranch
  \end{mathpar}
  \caption{Type conversion ($\PVGRTypeConv\Typ\Typ$)}
  \label{fig:polyvgr-type-conv}
\end{figure}
Figure~\ref{fig:polyvgr-type-conv} defines type conversion, where we
omit the standard rules for reflexivity, transitivity, symmetry, and
congruence. Conversion comprises beta reduction for functions and
pairs, and simplification of the dual operator: $\SEnd$ is
self-dual, the dual operator is involutory, for sending/receiving as
well as for choice/branch the dual operator flips the direction of the
communication.

Conversion is needed in the context of the dual operator, because a
programmer may use the dual operator in a type. If this type is
polymorphic over a session-kinded type variable $\TVar$, then the operator
cannot be fully eliminated as in $\Dual \TVar$. Once a type
application instantiates $\TVar$, we invoke conversion to enable
pushing the dual operator further down into the session type.

The conversion judgment does not destroy the simple inversion
properties of the expression and value typing rules as it is
explicitly invoked in just two expression typing rules:
\TirName{T-Send} for the $\ESend\cdot\cdot$ operation and \TirName{T-TApp} for
type application (see Figure~\ref{fig:polyvgr-exp-typing}).

\begin{figure}[tp]
  \begin{mathpar}
    \rulePVGRCstrEntailAxiom \and
    \rulePVGRCstrEntailSym \and
    \rulePVGRCstrEntailZero \and
    \rulePVGRCstrEntailSplit \and
    \rulePVGRCstrEntailMerge\and
    \rulePVGRCstrEntailProjMerge \and
    \rulePVGRCstrEntailProjSplit \\
    \rulePVGRCstrEntailEmpty \and
    \rulePVGRCstrEntailCons
  \end{mathpar}
  \caption{Constraint entailment ($\PVGRCstrEntail\Ctx\Cstr$)}
  \label{fig:polyvgr-constraint-entailment}
\end{figure}

Constraint entailment is defined structurally in
Figure~\ref{fig:polyvgr-constraint-entailment}. Disjointness of
domains can hold by assumption. Disjointness is symmetric. The empty
domain is disjoint with any other domain. Disjointness distributes over
combination of domains and is compatible with projections. It extends
to conjunctions of constraints in the obvious way.

The context restriction operators, $\CtxRestrictNonDom\Ctx$ and
$\CtxRestrictOnlyDom\Ctx$,  are a technical device. Both 
operators keep only bindings of type variables. One removes all domain
bindings and the other removes all non-domain bindings.

Figure~\ref{fig:polyvgr-disjoint-ctx-extension} defines the operator
$\Ctx_1 \DisjointAppend \Ctx_2$. The assumption is that $\Ctx_1$ is
known to contain disjoint bindings. The generated constraints $\Cstr_2$
make sure that $\Ctx_2$'s bindings are also disjoint and $\Cstr_{12}$
ensures that they are also  disjoint from $\Ctx_1$'s bindings.
\begin{figure}[tp]
  \begin{align*}
    \Ctx_1 \DisjointAppend \Ctx_2 &=
      \Ctx_1, \Ctx_2, \Cstr_2, \Cstr_{12} \text{ where}
  \end{align*}
  \begin{align*}
    \Cstr_2 &= \{ \TVar_1 \Disjoint \TVar_2 \mid  \TVar_1, \TVar_2 \in \Dom{\CtxRestrictOnlyDom{\Ctx_2}}, \TVar_1 \neq \TVar_2 \} \\
    \Cstr_{12} &= \{ \TVar_1 \Disjoint \TVar_2 \mid  \TVar_1 \in \Dom{\CtxRestrictOnlyDom{\Ctx_1}}, \TVar_2 \in  \Dom{\CtxRestrictOnlyDom{\Ctx_2}} \}
  \end{align*}
  \caption{Disjoint context extension ($\Ctx\DisjointAppend\Ctx$)}
  \label{fig:polyvgr-disjoint-ctx-extension}
\end{figure}

\subsection{Statics for expressions and processes}
\label{sec:statics-expressions}

\begin{figure}[tp]
  \begin{mathpar}
    \rulePVGRTypingVar \and
    \rulePVGRTypingUnit \and
    \rulePVGRTypingPair \and
    \rulePVGRTypingTAbs \and
    \rulePVGRTypingChan \and
    \hspace{-5mm}
    \rulePVGRTypingAbs
  \end{mathpar}
  \caption{Value typing  ($\PVGRHasValType\Ctx\Val\Typ$)}
  \label{fig:polyvgr-val-typing}
\end{figure}

As the syntax of expressions obeys A-normal form, there are three main judgments
\begin{itemize}
\item value typing $\PVGRHasValType\Ctx\Val\Typ$,
\item expression typing $\PVGRHasExpType\Ctx\St\Exp\Ctx\St\Typ$, and
\item configuration typing $\PVGRHasConfType\Ctx\St\Cfg$.
\end{itemize}
The rules in Figure~\ref{fig:polyvgr-val-typing} define the value
typing judgment that applies to syntactic values.
The most notable issue with these rules is that they do not handle
states. As syntactic values have no effect, they cannot affect the
state and this restriction is already stated in the typing judgment.

The rules for variables, unit, pairs, and type abstraction are
standard. Channel values refer to single-channel domains. Rule
\TirName{T-Abs} for lambda abstraction checks wellformedness of the
function type and invokes
expression typing to obtain the return state and  type.

\begin{figure}[tp]
  \begin{mathpar}
    \rulePVGRTypingLet \and
    \rulePVGRTypingVal \and
    \rulePVGRTypingProj \and
    \rulePVGRTypingNew \and
    \rulePVGRTypingApp \and
    \rulePVGRTypingTApp \and
    \hspace{-2mm}
    \rulePVGRTypingRequest \and
    \hspace{-2mm}
    \rulePVGRTypingAccept \and
    \rulePVGRTypingSend \and
    \rulePVGRTypingRecv \and
    \rulePVGRTypingFork \and
    \rulePVGRTypingClose \and
    \rulePVGRTypingSelect \and
    \rulePVGRTypingCase
  \end{mathpar}
  \caption{Expression typing ($\PVGRHasExpType\Ctx\St\Exp\Ctx\St\Typ$)}
  \label{fig:polyvgr-exp-typing}
\end{figure}

Figure~\ref{fig:polyvgr-exp-typing} contains the rules for expression
typing. We concentrate on the state-handling aspect as the value level
is mostly standard. Recall that we assume expressions are in
strict A-normal form, which means that every expression consists of a cascade
of $\Terminal{let}$ expressions that ends in a syntactic value. Rule
\TirName{T-Val} embeds values in expression typing. It is
special as it threads the entire state $\St$ even though it makes no
use of it. This special treatment is needed at the end of a
$\Terminal{let}$ cascade because rule \TirName{T-Let} splits the
incoming state for ${\ELet\EVar{\Exp_1}{\Exp_2}}$ into the part
$\St_1$ required by the header expression $\Exp_1$ and $\St_2$ for the
continuation $\Exp_2$, but then it feeds the entire outgoing state of $\Exp_1$
combined with $\St_2$ into the continuation $\Exp_2$. All
remaining rules only take the portion of the incoming state that is processed
by the operation, so they are
designed to be applied in the header position $\Exp_1$ of a
$\Terminal{let}$. Thankfully, this use is guaranteed by strict A-normal form.

The remaining rules all assume the expression is used in header
position of a $\Terminal{let}$. Projection (rule \TirName{T-Proj})
requires no state.
Type application (rule \TirName{T-TApp}) checks the constraints after
instantiation and enables conversion of the instantiated
type. Conversion is needed (among others) to expose the session type
operators (see discussion for Figure~\ref{fig:polyvgr-type-conv}).

Function application (rule \TirName{T-App}) just
rewrites the function type to an expression judgment.
The existential part of this judgment is reintegrated into the
state in the \TirName{T-Let} rule, which inserts the necessary
disjointness constraints via the disjoint append-operator
$\_\DisjointAppend\_$. As the \TirName{T-Let} rule presents the
function application exactly with the state it can handle, we must delay
the creation of the constraints to the $\Terminal{let}$-expression
because it is here that the return state must be merged with the state
for the continuation, which may contain additional  domains.
Given that the existentially bound domains are subject to
$\alpha$-renaming, we can freely impose the corresponding
disjointness constraints to force local freshness of the
domains. Explicit disjointness is required because of the axiomatic
nature of our constraint system.

The $\ENew$ expression creates an access point which requires no state
(rule \TirName{T-New}).
The rules \TirName{T-Request} and \TirName{T-Accept} type the
establishment of a connection via an access point. They return one end
of the freshly created channel, so that the channel's domain is
existentially quantified. The kind of this domain is
$\KDomain\TShapeOne$ (omitted in the rules as it is implied by the
binding).

The rule \TirName{T-Send} for sending is particularly interesting.
It splits the incoming state into the channel $\POLYDom$ on which the
sending takes place and the state $\St$, which will be passed along
with the value. The rule guesses a domain $\POLYDom'$ such that the
state expected in the session type matches the state $\St$ and the
type expected by the session type matches the type of the provided
argument. This matching is achieved with a type conversion judgment
that implements reduction for functions and pairs at the type level (see
Figure~\ref{fig:polyvgr-type-conv}).  The outgoing state only retains
the channel $\POLYDom$ bound to the continuation session $\Ses$.

Receiving (rule \TirName{T-Recv}) is much simpler: we treat the
received channels like new created one in the existential component of
the typing judgment.

Forking (rule \TirName{T-Fork}) starts a new process from a $\TUnit
\to \TUnit$ function. The new process takes ownership of all incoming
state.
Closing a channel (rule \TirName{T-Close}) just requires a single channel with
type $\SEnd$ and returns an empty state.

Rule \TirName{T-Select} performs the standard rewrite of the session
type for selecting a branch in the protocol. The dual rule
\TirName{T-Case} is slightly more subtle. It requires that both
branches end in the same state, that is, they must create channels and
operate on open channels in the same way (or close them before
returning from the branch).

\begin{figure}[tp]
  \begin{mathpar}
    \rulePVGRTypingExp \and
    \rulePVGRTypingPar \and
    \rulePVGRTypingBindChan \and
    \rulePVGRTypingBindChanClosed \and
    \rulePVGRTypingBindAP
  \end{mathpar}
  \caption{Configuration typing ($\PVGRHasConfType\Ctx\St\Cfg$)}
  \label{fig:polyvgr-conf-typing}
\end{figure}
Figure~\ref{fig:polyvgr-conf-typing} contains the typing rules for
$\PolyVGR$ processes. They are straightforward with one exception. In
rule \TirName{T-NuChan}, we need to make sure that the newly
introduced channel ends are disjoint (i.e., different) from each other
and from previously defined domains. Rule \TirName{T-NuChanClosed}
replaces \TirName{T-NuChan} after the channel is closed. The
difference is that it no longer places the channels in the state
$\St$. This way, operations on the closed channel are disabled, but it is
still possible to have references to it in dead code.

\subsection{Dynamics}
\label{sec:dynamics}

\begin{figure}[tp]
  \begin{mathpar}
    \rulePVGRExprRedBetaFun \and
    \rulePVGRExprRedBetaPair \and
    \rulePVGRExprRedBetaAll \and
    \rulePVGRExprRedBetaLet \and
    \rulePVGRExprRedLift
  \end{mathpar}
  \caption{Expression reduction ($\Exp \ReducesToE \Exp$)}
  \label{fig:polyvgr-expr-reduction}
\end{figure}
Figure~\ref{fig:polyvgr-expr-reduction} defines expression reduction,
which is standard for a polymorphic call-by-value lambda
calculus. Recall that an evaluation context just selects the header of
a $\Terminal{let}$ expression.

\begin{figure}[tp]
  \begin{mathpar}
    \rulePVGRCongNull \and
    \rulePVGRCongComm  \and
    \rulePVGRCongLift \and
    \rulePVGRCongAssoc \and
    \rulePVGRCongSwap \and
    \rulePVGRCongScopeChan \and
    \rulePVGRCongScopeAP
  \end{mathpar}
  \caption{Configuration congruence ($\Cfg \Cong \Cfg$)}
  \label{fig:polyvgr-congruence}
\end{figure}
Figure~\ref{fig:polyvgr-congruence} defines a congruence relation on
processes. This standard relation (process composition is commutative,
associative with the unit process as a neutral element, and compatible with channel and scope abstractions)
enables us to reorganize processes such that process reductions are
simple to state. Channel abstraction may swap the channel names.

\begin{figure}[tp]
  \begin{mathpar}
    \rulePVGRCfgRedFork \and
    \rulePVGRCfgRedNew \and
    \rulePVGRCfgRedExpr \and
    \rulePVGRCfgRedRequestAccept \and
    \rulePVGRCfgRedSendRecv \and
    \rulePVGRCfgRedSelectCase \and
    \rulePVGRCfgRedClose
  \end{mathpar}
  \caption{Configuration reduction ($\Cfg \ReducesToC \Cfg$)}
  \label{fig:polyvgr-cfg-reduction}
\end{figure}
Figure~\ref{fig:polyvgr-cfg-reduction} defines reduction for
processes. Rules \TirName{CR-Fork} and \TirName{CR-New} apply to an expression
process. The $\EFork$ expression creates a new process that applies
the $\EFork$'s argument to $\EUnit$ while the old process continues
with $\EUnit$.  The $\ENew$ expression creates a new access point and
leaves its name in the evaluation context.

The remaining rules all concern communication between two processes. Our
rules have explicit assumptions that congruence rearranges processes as needed
for the reductions to apply. All these rules involve binders and
assume an additional process $C'$ running in parallel with the
processes participating in the redex, which keeps the processes with
references to the binder.

Rule \TirName{CR-RequestAccept} creates a channel when there is a
request and an accept on the same access point. The reduction creates
the two ends of the new channel and passes them to the processes.

Rules \TirName{CR-SendRecv} and \TirName{CR-SelectCase} are standard. They could be blocked without
the congruence rule \TirName{CC-Swap} in place.

Rule \TirName{CR-Close} is slightly unusual for readers familiar with
linear session type calculi. The rule does not remove the closed
channel from the configuration because the process under the binder
may still contain (dead) references to the channel. This design makes
reasoning about configurations in final state slightly more involved.

\section{Metatheory}
\label{sec:metatheory}

\newcommand*\SubjectReductionLemma[1]{
  \begin{lemma}[Subject Reduction]
    \label{#1}
    \
    \begin{enumerate}
    \item
      $
      \inferrule{
        \PVGRIsCtx{\Ctx_1} \\
        \PVGRHasKind{\Ctx_1}{\St_1}{\KSt} \\
        \PVGRHasExpType{\Ctx_1}{\St_1}{\Exp}{\Ctx_2}{\St_2}{\Typ} \\
        \Exp \ReducesToE \Exp'
      }{
        \exists \Typ'.\
          \PVGRHasExpType{\Ctx_1}{\St_1}{\Exp'}{\Ctx_2}{\St_2}{\Typ'} \land
          \PVGRTypeConv{\Typ'}{\Typ}
      }
      $
    \item
      $
      \inferrule{
        \PVGRIsCtx{\Ctx} \\
        \PVGRHasKind{\Ctx}{\St}{\KSt} \\
        \PVGRHasConfType{\Ctx}{\St}{\Cfg} \\
        \Cfg \ReducesToC \Cfg'
      }{
        \exists \St'.\ 
        \PVGRHasConfType{\Ctx}{\St'}{\Cfg'}
      }
      $
    \end{enumerate}
  \end{lemma}
}

\newcommand*\SubjectCongruenceLemma[1]{
  \begin{lemma}[Subject Congruence]
    \label{#1}
    \begin{align*}
      \inferrule{
        \PVGRIsCtx{\Ctx} \\
        \PVGRHasKind{\Ctx}{\St}{\KSt} \\
        \PVGRHasConfType{\Ctx}{\St}{\Cfg} \\
        \Cfg \Cong \Cfg'
      }{
        \PVGRHasConfType{\Ctx}{\St}{\Cfg'}
      }
    \end{align*}
  \end{lemma}
}

We establish session fidelity and type soundness by applying the usual
syntactic methods based on subject reduction and progress. Our subject
reduction result for expressions 
applies in any context. As the type system of \PolyVGR includes a
conversion judgment, we can only prove subject reduction up to
conversion.
Subject reduction also holds for configurations.

All proofs along with additional lemmas etc may be found in the
supplemental material.

\SubjectReductionLemma{lem:subject-reduction}

As configuration reduction is applied modulo the congruence relation,
we also need to show that congruence preserves typing.

\SubjectCongruenceLemma{lem:subject-congruence}

It is tricky to state a progress property in the context of
processes,  in particular when deadlocks may occur. Hence,
we define several predicates on expressions to state progress
concisely. The $\IsValue\Exp$ predicate should be self
explanatory. The $\IsComm\Exp$ predicate characterizes expressions
that cannot reduce at the expression level, but require reduction at
the level of configurations. Of those, the $\EFork\_$ case is
harmless, but the other cases require interaction with other processes
to reduce.
\begin{definition}
  The predicates $\IsValue\Exp$ and $\IsComm\Exp$ are defined inductively.
  \begin{itemize}
  \item $\IsValue\Exp$ if exists $\Val$ such that $\Exp=\Val$.
  \item $\IsComm\Exp$ if one of the following cases applies
    \begin{multicols}{2}
    \begin{itemize}
    \item $\Exp = \EFork{\EAbs{\St}\EVar{\Typ}\Exp_1}$,
    \item $\Exp = \ENew\Ses$,
    \item $\Exp = \EAccept\Val $,
    \item $\Exp = \ERequest\Val $,
    \item $\Exp = \ESend{\Val}{\EChan\POLYDom} $,
    \item $\Exp = \ERecv {\EChan\POLYDom}$,
    \item $\Exp = \ESelect\Label {\EChan\POLYDom}$,
    \item $\Exp = \ECase {\EChan\POLYDom}{\Exp_1}{\Exp_2}$,
    \item $\Exp = \EClose {\EChan\POLYDom} $, or
    \item $\Exp = \ELet\EVar{\Exp_1}{\Exp_2}$\\ where $\IsComm{\Exp_1}$.
    \end{itemize}
    \end{multicols}
  \end{itemize}
\end{definition}
We also need a predicate that characterizes contexts built in a
configuration. Besides type variables and constraints, they can only
bind access points.
\begin{definition}
  The predicate $\IsOuter\Ctx$ is defined by
  \begin{itemize}
  \item $\IsOuter\Empty$,
  \item $\IsOuter{(\Ctx, \TVar : \Kind)}$ if $\IsOuter\Ctx$,
  \item $\IsOuter{(\Ctx, \EVar : \Typ)}$ if $\IsOuter\Ctx$ and $\Typ =
    \Ap\Ses$, and
  \item $\IsOuter{(\Ctx, \POLYDom_1 \DisjointAppend \POLYDom_2)}$ if
    $\IsOuter\Ctx$. 
  \end{itemize}
\end{definition}

We are now ready to state progress for expressions. A typed expression
is either a value, stuck on a communication (or fork), or it reduces.
\begin{restatable}[Progress for expressions]{lemma}{ProgressLemma}\label{lem:progress-congruence}
  \begin{mathpar}
    \inferrule{
      \PVGRIsCtx{\Ctx} \\
      \IsOuter\Ctx \\
      \PVGRHasKind{\Ctx}{\St}{\KSt} \\
      \PVGRHasExpType{\Ctx}{\St}{\Exp}{\Ctx'}{\St'}{\Typ'}
    }{
      \IsValue\Exp \lor
      \IsComm\Exp \lor
      \exists \Exp'.\ \Exp \ReducesToE \Exp'
    }
  \end{mathpar}
\end{restatable}


We also need to characterize configurations. A final configuration
cannot reduce in a good way: All processes are reduced to values, all
protocols on channels have concluded as indicated by their session
type $\SEnd$, and there may be access points.
\begin{definition}\label{def:is-final}
  The predicate $\IsFinal\Cfg$ is defined inductively by the following cases:
  \begin{itemize}
  \item $\IsFinal \Val$ (an expression process reduced to a value),
  \item $\IsFinal{ (\CPar{\Cfg_1}{\Cfg_2})}$ if $\IsFinal{\Cfg_1}$ and
    $\IsFinal{\Cfg_2}$, 
  \item $\IsFinal{ (\CBindAP\EVar\Ses\Cfg_1)}$ if $\IsFinal{\Cfg_1}$,
    or
  \item $\IsFinal{ (\CBindChan\TVar{\TVar'}{\SEnd}\Cfg_1)}$ if $\IsFinal{\Cfg_1}$.
  \end{itemize}
\end{definition}
The other possibility is that a configuration is deadlocked. The
following definition lists all the ways in which reduction of a
configuration may be disabled.
\begin{definition}\label{def:is-deadlock}
  The predicate $\IsDeadlock\Cfg$ holds for a configuration $\Cfg$ iff:
  \begin{enumerate}
  \item\label{item:1} For all configuration contexts $\CCtxSym$, if $\Cfg =
    \CCtx\Exp$, then either $\IsValue\Exp$ or $\IsComm\Exp$ and $\Exp
    \ne \EFork\Val$ and $\Exp \ne \ENew\Ses$.
  \item\label{item:2} For all configuration contexts $\CCtxSym$, if $\Cfg =
    \CCtx{\CBindAP\EVar\Ses\Cfg'}$, then
    \begin{itemize}
    \item if $\Cfg' = \CCtxA{\ECtxA{\ERequest\EVar}}$, then there is
      no $\CCtxSym_2$, $\ECtxSym_2$ such that $\Cfg' = \CCtxB{\ECtxB{\EAccept\EVar}}$,
    \item if $\Cfg' = \CCtxA{\ECtxA{\EAccept\EVar}}$, then there is
      no $\CCtxSym_2$, $\ECtxSym_2$ such that $\Cfg' = \CCtxB{\ECtxB{\ERequest\EVar}}$.
    \end{itemize}
  \item\label{item:3} For all configuration contexts $\CCtxSym$, if $\Cfg =
    \CCtx{\CBindChan{\TVar_1}{\TVar_2}\Ses \Cfg'}$, then
\begin{itemize}
    \item if $\Cfg' = \CCtxA{\ECtxA{\ESend\Val{\EChan{\TVar_\Label}}}}$, then there is
      no $\CCtxSym_2$, $\ECtxSym_2$ such that $\Cfg' = \CCtxB{\ECtxB{\ERecv{\EChan{\TVar_{3-\Label}}}}}$,
    \item if $\Cfg' = \CCtxA{\ECtxA{\ERecv{\EChan{\TVar_{\Label}}}}}$, then there is
      no $\CCtxSym_2$, $\ECtxSym_2$ such that $\Cfg' = \CCtxB{\ECtxB{\ESend\Val{\EChan{\TVar_{3-\Label}}}}}$,
    \item if $\Cfg' = \CCtxA{\ECtxA{\ESelect{\Label'}{\EChan{\TVar_\Label}}}}$, then there is
      no $\CCtxSym_2$, $\ECtxSym_2$ such that $\Cfg' = \CCtxB{\ECtxB{\ECase{\EChan{\TVar_{3-\Label}}}{\Exp_1}{\Exp_2}}}$,
    \item if $\Cfg' = \CCtxA{\ECtxA{\ECase{\EChan{\TVar_{\Label}}}{\Exp_1}{\Exp_2}}}$, then there is
      no $\CCtxSym_2$, $\ECtxSym_2$ such that $\Cfg' = \CCtxB{\ECtxB{\ESelect{\Label'}{\EChan{\TVar_{3-\Label}}}}}$,
    \item if $\Cfg' = \CCtxA{\ECtxA{\EClose{\EChan{\TVar_\Label}}}}$, then there is
      no $\CCtxSym_2$, $\ECtxSym_2$ such that $\Cfg' = \CCtxB{\ECtxB{\EClose{\EChan{\TVar_{3-\Label}}}}}$.
    \end{itemize}
  \end{enumerate}
\end{definition}

\vspace{-3mm}
\begin{restatable}[Progress for configurations]{lemma}{ProgressConfigurations}\label{lem:progress-configurations}
  \begin{mathpar}
    \inferrule{
      \PVGRIsCtx{\Ctx} \\
      \IsOuter\Ctx \\
      \PVGRHasKind{\Ctx}{\St}{\KSt} \\
      \PVGRHasConfType{\Ctx}{\St}{\Cfg} \\
    }{
      \IsFinal\Cfg \lor
      \IsDeadlock\Cfg \lor
      \exists\Cfg'.\ \Cfg \ReducesToC \Cfg'
    }
  \end{mathpar}
\end{restatable}

\section{Implementation}
\label{sec:implementation}

We have implemented a type checker and an interpreter for \PolyVGR in Haskell.
The syntax accepted by the implementation is exactly as presented in
this paper, i.e., type annotations are required at lambda abstractions
for input type and input state. 

The implementation of the type checker requires an algorithmic
formulation of the typing.
We briefly sketch how to make the declarative typing presented in this
paper algorithmic.
\begin{itemize}
  \item The rules for type conversion give rise to a
    type normalization function. Type conversion can then be decided by checking
    alpha-equivalence of normalized types.
  \item Constraint solving $\PVGRCstrEntail{\Ctx}{\Cstr}$ is decidable
    by normalizing and decomposing $\Ctx$ and $\Cstr$ into closed sets of atomic
    constraints $A_\Ctx$ and $A_\Cstr$ and checking $A_\Ctx \supseteq
    A_\Cstr$. Decomposition is done according to \textsc{CE-Split} and
    \textsc{CE-Sym}, yielding constraints of form $d_1 \Disjoint d_2$ where 
    $d_i = \pi_{\Label_1} \ldots \pi_{\Label_{n_i}} \alpha$. Then the
    closure is taken with respect to \textsc{CE-ProjMerge},
    \textsc{CE-ProjSplit}, and \textsc{CE-Sym}.
  \item The \textsc{T-Let} and \textsc{T-Par} rules
    non-deterministically split the input state $\St_1,\St_2$ between the
    subterms.
    The implementation threads the entire input state
    through the first subterm and uses the resulting output state as
    the input for the second subterm.
  \item In \textsc{T-Case} we have to check if the existential parts
    $\exists \Gamma_i. \St_i; \Typ_i$ of both branch typings are
    equal. This equality should be up to alpha-renaming and reordering of
    the variables bound in $\Gamma_1$ and $\Gamma_2$ and up the reordering
    of the bindings in $\St_1$ and $\St_2$. This equality can be decided
    by computing a renaming  $\rho$ on type variables such that $ \rho\Typ_1 = \Typ_2$. If all
    variables in the domain of $\rho$ are bound in $\Gamma_1$ and
    $\rho\Gamma_1 = \Gamma_2$ and $\rho\St_1 = \St_2$ up to
    reordering, then both existential packages are equal.
    A similar approach is used in \textsc{T-Send}, where the
    existential package of the session type needs to be matched against
    the current context and state.
\end{itemize}

\section{Extensions}
\label{sec:extensions}

Typestate is notoriously difficult to scale up to sum types or, more
generally, to algebraic datatypes. In this section, we sketch our
approach to add sum types to \PolyVGR and offer some insights into the
additional problems involved in handling recursive datatypes like lists.


To understand the issues arising with sum types, consider the type
$\Chan\Nchan + \Chan\Nchanb$ in the context of state $\St$. The situation is clear at run time: we
either have a channel described by $\Nchan$ or one described by
$\Nchanb$. But which channel should be described in the state $\St$?
Clearly, $\St = \Nchan\mapsto S_1; \Nchanb\mapsto S_2$ does not work
because it does not express the mutual exclusiveness of the presence
of $\Nchan$ and $\Nchanb$. In fact, if we matched against a value of
type $\Chan\Nchan + \Chan\Nchanb$, we would only consume one of
$\Nchan$ or $\Nchanb$ and leave the other channel identity dangling in
the outgoing state.

Instead, we propose to add new shapes and domains to the type
system along with the sum type. As we will see, the remaining features
needed to deal with sum types are already provided for.
\begin{align*}
    & \text{Types} &
    \Typ,\POLYShape,\POLYDom ::=\ &
                                             \dots
                                             \mid \TSum\Typ\Typ
                                             \mid \TShapeSum\POLYShape\POLYShape
                                    \mid \TDomEither\POLYDom\POLYDom
                                             \mid \TDomFrom \Label \POLYDom
    \\
    & \mathrm{Expressions} &
    \Exp ::=\ &
      \dots \mid
                \EInj \Label \Val \mid
                \EMatch\Val{x:\Exp}{x:\Exp}
\end{align*}
Sum types come with the usual introduction and elimination
forms, a sum shape $\TSum\POLYShape\POLYShape$, the domain of a sum
shape, and two sum extractors $\TDomFrom \LabelA \POLYDom$ and $\TDomFrom
\LabelB \POLYDom$ pronounced ``from''. The domain of a sum shape is a
pair of the domains of the two alternatives of the sum. The extractors are only
applicable to domains of sum shape and behave like projections as
becomes clear from the 
formation rules (extending \Cref{fig:polyvgr-kinding-I}):
\begin{mathpar}
  \rulePVGRKindingSum

  \rulePVGRKindingDomFrom
  \\
  \rulePVGRKindingShapeSum
  \hspace{2mm}
  \rulePVGRKindingDomSum
\end{mathpar}
The typing of sum introduction and elimination needs to be adapted to
account for shapes.
\begin{mathpar}
  \rulePVGRTypingInj
\end{mathpar}
In rule \TirName{T-Inj1}, we are given a value $\Val:\Typ$ along with
some $\St$ that describes the channels contained in $\Val$.
We assume that the shape of $\St$ is described by $\POLYShape_\LabelA$ and
corresponding domain $\POLYDom$. We further assume that the
alternatives of the sum are described by type functions $\StFun_\LabelA,
\TFun_\LabelA$ and $\StFun_\LabelB,
\TFun_\LabelB$. The point is that the pair labeled $\LabelA$ describes the real
resources in $\Val$ represented by $\St$ and the pair labeled
$\LabelB$ describes virtual
resources that serve as placeholders to describe the (non-existent)
other alternative of the sum. The two conversions determine the
connection to the real resources. 

Injecting the value into the sum type creates a virtual resource for
the non-existing alternative, which is represented by domain
$\Nchanb$. The real part---labeled $\LabelA$---continues to refer to
the same resources $\POLYDom$, so that the sharing semantics of
further channel references for those resources is preserved. The
virtual part---labeled $\LabelB$---is never exercised because the
run-time value has the form $\EInj\LabelA\Val$.

The alert reader might wonder why we do not treat $\EInj\LabelA\Val$
as a value. Indeed, $\EInj\LabelA\Val$ comes with a reduction to
create the virtual resource $\Nchanb$, which returns a syntactic value
$\VInj\LabelA\Val$. We elide the corresponding value typing rule,
which is obtained from \TirName{T-Inj1} by stripping the $\St$
components and assuming the presence of both domains in $\Ctx$.
\begin{mathpar}
  \rulePVGRTypingMatch
\end{mathpar}
To match on a value $\Val$ of sum type the elimination rule \TirName{T-Match} requires a
corresponding domain $\POLYDom$ of sum shape and we must be able to
partition the incoming state according to its two alternatives. (If
one of the alternatives carries no channels, then its shape is
$\TShapeZero$ and the corresponding state is empty.) As in the
introduction rule, the type and state functions $\TFun_\Label$ and
$\StFun_\Label$ describe the partitioning.
The match keeps the selected part of the state, which corresponds to
the real resources, and drops the other part, which corresponds to the
virtual resources.

The same general approach would also work for lists. However, due to
the recursion in the list type, we cannot allow sharing
between values in the list and outside of it. Essentially, a channel
value that is incorporated in a list has to give up its identity, but
at the same time the identity has to be remembered so that the channel
can be reconnected when extracted from the list.

\section{Related Work}
\label{sec:related-work}

We do not attempt to survey the vast amount of work in the session
type community, but refer the reader to recent survey papers and books
\cite{DBLP:journals/csur/HuttelLVCCDMPRT16,DBLP:journals/jlp/BartolettiCDDGP15,DBLP:journals/ftpl/AnconaBB0CDGGGH16,gay17:_behav_types}. Instead
we comment on the use of polymorphism in session types,
the modeling of disjointness in the context of polymorphism, and
potential connections to other work.

\subsection{Polymorphism and Session Types}
\label{sec:polym-sess-types}

Polymorphism for session types was ignored for quite a while, although
there are low-hanging fruit like parameterizing over the continuation
session. The story starts with an investigation of  bounded
polymorphism over the type of transmitted values to avoid problems
with subtyping in a $\pi$-calculus setting  \cite{DBLP:journals/mscs/Gay08}.

\citet{DBLP:journals/jfp/Wadler14} includes polymorphism on session
types where the quantifiers $\forall$ and $\exists$ are interpreted as sending
and receiving types, similar to Turner's polymorphic $\pi$-calculus
\cite{DBLP:phd/ethos/Turner96}.
\citet{DBLP:conf/esop/CairesPPT13,DBLP:journals/iandc/PerezCPT14}
consider impredicative quantifiers with session types using the same interpretation.

\citet{DBLP:journals/iandc/DardhaGS17}
extend an encoding of session types into $\pi$-types with parametric
and bounded polymorphism.
\citet{lindley17:_light_funct_session_types} rely on row
polymorphism  to abstract over the irrelevant labels in a choice,
thereby eliding the need for supporting subtyping. Their
calculus FST (lightweight functional session types) supports
polymorphism over kinded type variables $\alpha :: K (Y,Z)$ where
$K=\mathit{Type}$, $Y=\circ$, and $Z=\pi$ indicates a variable ranging
over session types; choosing $K=\mathit{Row}$ yields a row variable.
%
\citet{ALMEIDA2022104948} consider impredicative polymorphism in the
context of context-free session types. Their main contribution is the
integration of algorithmic type checking for context-free sessions
with polymorphism.

All practically oriented works
\cite{lindley17:_light_funct_session_types,ALMEIDA2022104948}
rely on an elaborate kind system to distinguish linear from
non-linear values, session types from non-session types, and rows from
types (in the case of FST). \PolyVGR follows suit in that its kinds
distinguish session types and non-session types. Linearity is elided,
but kinds for states, shapes, and domains are needed to handle
channels. As a major novelty, \PolyVGR includes arrow kinds and type-level
lambda abstraction, but restricted such that abstraction ranges solely
over domains.


\subsection{Polymorphism and Disjointness}
\label{sec:disjointness}

Alias types
\cite{DBLP:conf/esop/SmithWM00} is
a type system for a low-level language where the type of a function expresses the shape of the store
on which the function operates. For generality, function types can abstract over store locations
$\alpha$ and the shape of the store is described by \emph{aliasing
  constraints} of the form $\{\alpha\mapsto T\}$. Constraint
composition resembles separating conjunction
\cite{DBLP:conf/lics/Reynolds02} and ensures that locations are
unique. Analogous to the channel types in our system, pointers in the
alias types system have a singleton type that indicates their
(abstract) store location and they can be duplicated. Alias types also
include non-linear constraints, which are not required in our
system. Alias types do not provide the means to abstract over groups
of store locations as is possible with our domain/shape approach. It
would be interesting to investigate such an extension to alias types.

Low-level liquid types \cite{DBLP:conf/popl/RondonKJ10} use a similar
setup as alias types with location-sensitive types to track pointers and pointer
arithmetic as well as to enable strong updates in a verification
setting for a C-like language. They also provide a mechanism of
unfolding and folding to temporarily strengthen pointers so that they
can be strongly updated. Such a mechanism is not needed for our
calculus as channel resources are never aliased.


Disjoint intersection types \cite{DBLP:conf/icfp/OliveiraSA16} have been
conceived to address the coherence problem of intersection type
systems with an explicit merge operator: if the two ``components'' of
the merge have the same type, then it is not clear which value should
be chosen by the semantics. They rule
out this scenario by requiring different types for all components of
an intersection. Disjoint polymorphism
\cite{DBLP:conf/esop/AlpuimOS17} lifts this idea to a
polymorphic calculus where type variables are introduced with
disjointness constraints that rule out overlapping instantiations.
\citet{DBLP:conf/ecoop/XieOBS20} show that calculi for disjoint
polymorphic intersection types are closely related to polymorphic
record calculi with symmetric concatenation
\cite{DBLP:conf/popl/HarperP91}.

Disjointness constraints for record types are related to
our setting, but the labels in the records types are fixed and two records are still deemed
disjoint if they share labels, as long as the corresponding field
types are disjoint. In contrast,  we have universal and existential
quantification over domains (generalizing channel names) and
single-channel domains disjoint by our axiomatic construction when
composing states.

\citet{DBLP:journals/pacmpl/MorrisM19} propose a generic system Rose for
typing features based on row types. Its basis is a partial monoid of
rows, which is chosen according to the application. Using rows for
record types, Rose can be instantiated to support symmetric
concatenation of records, shadowing concatenation, or even to allow
several occurrences of the same label.
While channel names are loosely related to record label and states
might be represented as records, our axiomatic approach to maintaining
disjointness is significantly different from their Rose system.

\subsection{Diverse Topics}
\label{sec:diverse-topics}

\citet{DBLP:conf/haskell/PucellaT08} give an embedding of a session
type calculus in Haskell. Their embedding is based on Atkey's parameterized
monads \cite{DBLP:journals/jfp/Atkey09}, layered on top of the IO monad
using phantom types. Their phantom type structure resembles our states
where de Bruijn indices serve as channels names. Linear handling of
the state is enforced by the monad abstraction, while channel
references can be handled freely. The paper
comes with a formalization and a soundness proof of the implementation.
\citet{SackmanE08} also encode session types for a single channel in
Haskell using an indexed (parameterized) monad.

A similar idea is the basis for work by
\citet{DBLP:conf/coordination/Saffrich021,DBLP:journals/lmcs/SaffrichT22},
which is closely related to our investigation. They also start from VGR,
point out some of its restrictions, but then continue to define a
translation into a linear parameterized monad, which can be
implemented in an existing monomorphic functional session type calculus
\cite{DBLP:journals/jfp/GayV10}, extended with some syntactic sugar in
the form of linear records. They prove that there are semantics- and
typing-preserving translations forth and back, provided the typing of
the functional calculus is severely restricted. Our work removes most
of the restrictions of VGR's type system by using higher-order
polymorphism. It remains to complete the diagram and identify a
polymorphic functional session type calculus (most likely FST) which
is suitable as a translation target.



\citet{DBLP:conf/cpp/HinrichsenLKB21} describe semantic session typing
as an alternative way to establish sound session type regimes. Instead
of delving into syntactic type soundness proofs, they suggest to
define a semantic notion of types and typing on top of an untyped
semantics. Their proposal is based on (step-indexed) logical relations
defined in terms of a suitable program logic
\cite{DBLP:journals/corr/abs-1103-0510} and it is fully mechanized in Coq.  Starting from a simple
session type system, they add polymorphism, subtyping, recursion, and
more. It seems plausible that their model would scale to provide
mechanized soundness proofs for \PolyVGR.

\citet{DBLP:journals/pacmpl/BalzerP17} considers a notion of manifest sharing
in session types. Their notion is substantially different from our
work. \PolyVGR facilitates (local) variables, not constrained by
linearity, bound to channel references. Thanks to typestate, the same
reference can refer to a channel in different states at different
points in a program. In manifest sharing, there are globally shared
channels which always offer the same state. Processes can pick up a
shared channel, run an unshared protocol on it, and return it in the
same shared state as before. 

\section{Conclusion}
\label{sec:conclusion}

We started this work on two premises:
\begin{itemize}
\item We believe it is important to map the unexplored part of
  the design space of session type systems based on typestate.
\item We believe that there are practical advantages in
  being able to write programs with session types in direct style as in
  Listing~\ref{lst:example-server}.
\end{itemize}
Looking back, we find that the direct style is scalable, it should be
on the map as it more easily integrates with imperative programming
styles and languages, and \PolyVGR explains in depth the type system
ingredients needed to decently program with session types in direct
style.
On the other hand, the amount of
parameterization required in \PolyVGR is significant and may be burdensome for
programmers. We are just starting to gather practical experience with
our implementation of \PolyVGR, so we cannot offer a final verdict at
this point.



\clearpage{}


\bibliography{biblio,Allais17,Atkey09,HuKPYH10,Honda93,GayV10,ScalasY16,StromY86,TakeuchiHK94,VasconcelosRG04,VasconcelosGR06,MorrisM19,SmithWM00,Reynolds02,AlpuimOS17,XieOBS20,OliveiraSA16,ChughHJ12,RondonKJ10,HarperP91,PucellaT08,SackmanE08,Padovani17,HuY16,FowlerLMD19,SabryF93,FlanaganSDF93}

\clearpage
\appendix
\onecolumn
\section{Appendix}


\subsection{Channel Aliasing}
\label{sec:channel-aliasing}

The VGR paper proposes the following function \lstinline/sendSend/.
\begin{lstlisting}
fun sendSend u v = send 1 on u; send 2 on v
\end{lstlisting}
It takes two channels and sends a
number on each. This use is reflected in the following typing.
\begin{gather}\label{eq:3}
  \mathtt{sendSend} :
 \Sigma_1; \Chan u \to (\Sigma_1; \Chan v \to \Unit; \Sigma_2); \Sigma_1
\end{gather}
with $\Sigma_1 = \{u : \Outp\Int{S_u}, v : \Outp\Int{S_v}\}$
and  $\Sigma_2 = \{u : {S_u}, v : {S_v}\}$.

Ignoring the types we observe that it would be semantically sound to pass a reference to
the same channel \lstinline/w/, say, of session type \lstinline/!Int.!Int.End/
for \lstinline/u/ and \lstinline/v/. However, \lstinline/sendSend w w/
does not type check with the type in~(\ref{eq:3}) because \lstinline/w/ would
have to have identity $u$ and $v$ at the same time, but  state
formation mandates they must be different.

In {\PolyVGR}, the type for \lstinline/sendSend/ would be universally quantified over the channel
names:
\begin{Newstuff}{\NewstuffPadding}
  \begin{gather*}
    \begin{array}{@{}l}
    \TAll\TVara{\KDomain\TShapeOne}{\CstrTrue}
    \TAll\TVarb{\KDomain\TShapeOne}{(\TVara\Disjoint\TVarb)}
    \TAll{\SVar_1}{\KSession}{\CstrTrue}
    \TAll{\SVar_2}{\KSession}{\CstrTrue}
      \\
    \TArr{\StEmpty}{\TChan\TVara}{\Empty}{\StEmpty}{
      \TArr{
      \left\{
      \begin{array}{l}
      \StBind\TVara{{\SSend\TVar{\KDomain\TShapeZero}{\StEmpty}{\Int}\SVar_1}}, \\
      \StBind\TVarb{{\SSend\TVar{\KDomain\TShapeZero}{\StEmpty}{\Int}\SVar_2}}
      \end{array}
      \right\}
      }{\TChan\TVarb}{\Empty}{
      \left\{
      \begin{array}{l}
        \StBind\TVara{\SVar_1}, \\
        \StBind\TVarb{\SVar_2}
      \end{array}
      \right\}
      }{
          \TUnit
      }
    }
    \end{array}
  \end{gather*}
\end{Newstuff}
Wellformedness of states requires that $\TVara$ and $\TVarb$ are different because
they index the same state. Hence, \lstinline/sendSend w w/ does not type check in \PolyVGR, either.

Another VGR typing of the same code would be \lstinline/sendSend/ :
$\Sigma_1;\Chan w \to (\Sigma_1; \Chan w \to \Unit; \Sigma_2);
\Sigma_1$ with $\Sigma_1 = \{ w : \Outp\Int{\Outp\Int{S_w}} \}$ and
$\Sigma_2 = \{ w : S_w \}$. With this typing, \lstinline/sendSend w w/
type checks. Indeed, the typing forces the two arguments to be aliases!

A similar type could be given in \PolyVGR:
\begin{Newstuff}{\NewstuffPadding}\small
  \begin{gather*}
    \begin{array}{@{}l}
    \TAll\TVara{\KDomain\TShapeOne}{\CstrTrue}
    \TAll{\SVar}{\KSession}{\CstrTrue}
      \\
    \TArr{\StEmpty}{\TChan\TVara}{\Empty}{\StEmpty}{
      \TArr{
      \StBind\TVara{{\SSend\TVar{\KDomain\TShapeZero}{\StEmpty}{\Int}{\SSend\TVar{\KDomain\TShapeZero}{\StEmpty}{\Int}\SVar}}}
      }{\TChan\TVara}{\Empty}{
        \StBind\TVara{\SVar}
      }{
          \TUnit
      }
    }
    \end{array}
  \end{gather*}
\end{Newstuff}

Presently it is not possible
to give a single type to both aliased and unaliased uses of the
function.

\subsection{Higher-Order Functions}
\label{sec:high-order-functions}

Section~\ref{sec:motivation} has shown that VGR lacks facilities for
abstraction. In this subsection, we give further indication of the
flexibility of our system by discussing different typings for a simple higher-order
function.

Consider a higher-order function that is a prototype for a protocol
adapter. Given an argument function that runs a protocol, the adapter
adds a prefix to the protocol, perhaps for authentication
or accounting. To keep our example simple, the prefix
consists of sending a single number, but more elaborate protocols are
possible. The implementation is straightforward:
\begin{lstlisting}
fun adapter f c =
  send 32168 on c;
  f c
\end{lstlisting}
The first type for \lstinline/f/ combines the channel creation pattern from
Section~\ref{sec:channel-creation} with the flexibility of supporting
arbitrarily many channels as in Section~\ref{sec:general-case}.
\begin{Newstuff}{\NewstuffPadding}
  \begin{gather}\label{eq:10}
    \begin{array}{@{}l}
      \TAll\TVarShape\KShape{\CstrTrue}
      \TAll\TVarStFun{\KStFun\TVarShape}{\CstrTrue}
      \TAll\TVarTFun{\KTypeFun\TVarShape}{\CstrTrue}
      \\
      \TAll\TVarc{\KDomain\TShapeOne}{\CstrTrue}
      \TAll{\SVar}{\KSession}{\CstrTrue}
      \TAll{\SVar'}{\KSession}{\CstrTrue}
      \\
      \TArrBegin{\StEmpty}{
      \TArr{\StBind\TVarc{\SVar}}{\TChan\TVarc}{\TVar\colon
      \KDomain\TVarShape}{\TVarStFun\ \TVar,
      \StBind\TVarc{\SVar'}}{\TVarTFun\ \TVar}
      }
      \\
      \qquad\TArrEnd {\Empty} {\StEmpty}{
      \TArr{
      \StBind\TVarc{\SSend\TVar{\KDomain\TShapeZero}{\StEmpty}{\Int}{\SVar}}
      }{\TChan\TVarc}{\TVar\colon \KDomain\TVarShape}{
      \TVarStFun\ \TVar,\StBind\TVarc{\SVar'}
      }{
          \TVarTFun\ \TVar
      }
    }
    \end{array}
  \end{gather}
\end{Newstuff}
With this  \PolyVGR{} type, the function parameter \lstinline/f/ can
create arbitrary many channels and it can return arbitrary values that
may or may not include channels. In the degenerate case where
$\TVarShape$ is $\TShapeZero$, the universal quantification
$\TAll\TVarTFun{\KTypeFun\TShapeZero}{\CstrTrue}\dots$ quantifies over
types that do not contain channel references. Existentially
quantified variables, like $\TVar$, carry an implicit disjointness
constraint with any other domain variable in scope. This constraint
ensures that $\TVarStFun\ \TVar,\StBind\TVarc{\SVar'}$ is wellformed.

However, there is an issue that the type in~(\ref{eq:10}) does
not address. The function parameter \lstinline/f/ cannot be closed over further
channels! To address that shortcoming requires another style of
parameterization over the incoming and the outgoing state of
\lstinline/f/. The shapes of these states are unknown and the may be
different, so we need two shape parameters. Moreover, these two states
never mix, so their domains $\TVar'$ and $\TVar''$ need \emph{not} be disjoint. For simplicity, we first
consider functions that do not create new channels, although both parameterizations can be
combined.

\begin{Newstuff}{\NewstuffPadding}
  \begin{gather}\label{eq:11}
    \begin{array}{@{}l}
      \TAll{\TVarShape'}\KShape{\CstrTrue}
      \TAll{\TVarStFun'}{\KStFun{\TVarShape'}}{\CstrTrue}
      \\
      \TAll{\TVarShape''}\KShape{\CstrTrue}
      \TAll{\TVarStFun''}{\KStFun{\TVarShape''}}{\CstrTrue}
      \TAll{\TVarTFun''}{\KTypeFun{\TVarShape''}}{\CstrTrue}
      \\
      \TAll{\TVar'}{\KDomain{\TVarShape'}}{\CstrTrue}
      \TAll{\TVar''}{\KDomain{\TVarShape''}}{\CstrTrue}
      \TAll\TVarc{\KDomain\TShapeOne}{(\TVar'\Disjoint\TVarc, \TVar''\Disjoint\TVarc)}
      \\
      \TAll{\SVar}{\KSession}{\CstrTrue}
      \TAll{\SVar'}{\KSession}{\CstrTrue}
      \\
      \TArrBegin{\StEmpty}{
      \TArr{\TVarStFun'\ \TVar', \StBind\TVarc{\SVar}}{\TChan\TVarc}{\Empty}{\TVarStFun''\ \TVar'',
      \StBind\TVarc{\SVar'}}{\TVarTFun''\ \TVar''}
      }
      \\
      \qquad\TArrEnd {\Empty} {\StEmpty}{
      \TArr{
      \TVarStFun'\ \TVar',
      \StBind\TVarc{\SSend\TVar{\KDomain\TShapeZero}{\StEmpty}{\Int}{\SVar}}
      }{\TChan\TVarc}{\Empty}{
      \TVarStFun''\ \TVar'',\StBind\TVarc{\SVar'}
      }{
          \TVarTFun''\ \TVar''
      }
    }
    \end{array}
  \end{gather}
\end{Newstuff}

To parameterize over functions with arbitrary free
channels and which may create channels, we need one more
ingredient to describe the shape of the new state that contains
the descriptions of the newly created channels. This extra shape and
its use is highlighted in the type.
\begin{Newstuff}{\NewstuffPadding}\small
  \begin{gather}\label{eq:12}
    \begin{array}{@{}l}
      \TAll{\TVarShape'}\KShape{\CstrTrue}
      \TAll{\TVarStFun'}{\KStFun{\TVarShape'}}{\CstrTrue}
      \\
      \TAll{\TVarShape''}\KShape{\CstrTrue}
      \\
      \textcolor{red}{\TAll{\TVarShape}\KShape{\CstrTrue}}
      {\TAll{\TVarStFun''}{\KStFun{\textcolor{red}{\TShapePair\TVarShape{\TVarShape''}}}}{\CstrTrue}}
      \TAll{\TVarTFun''}{\KTypeFun{\textcolor{red}{\TShapePair\TVarShape{\TVarShape''}}}}{\CstrTrue}
      \\
      \TAll{\TVar'}{\KDomain{\TVarShape'}}{\CstrTrue}
      \TAll{\TVar''}{\KDomain{\TVarShape''}}{\CstrTrue}
      \TAll\TVarc{\KDomain\TShapeOne}{(\TVar'\Disjoint\TVarc, \TVar''\Disjoint\TVarc)}
      \\
      \TAll{\SVar}{\KSession}{\CstrTrue}
      \TAll{\SVar'}{\KSession}{\CstrTrue}
      \\
      \TArrBegin{\StEmpty}{
      \TArr{\TVarStFun'\ \TVar',
      \StBind\TVarc{\SVar}}{\TChan\TVarc}{\TVar\colon
      \KDomain\TVarShape}{\TVarStFun''\ \textcolor{red}{(\TVar,\TVar'')},
      \StBind\TVarc{\SVar'}}{\TVarTFun''\ \textcolor{red}{(\TVar,\TVar'')}}
      }
      \\
      \qquad\TArrEnd {\Empty} {\StEmpty}{
      \TArr{
      \TVarStFun'\ \TVar',
      \StBind\TVarc{\SSend\TVar{\KDomain\TShapeZero}{\StEmpty}{\Int}{\SVar}}
      }{\TChan\TVarc}{\TVar\colon \KDomain\TVarShape}{
      \TVarStFun''\ \textcolor{red}{(\TVar,\TVar'')},\StBind\TVarc{\SVar'}
      }{
          \TVarTFun''\ \textcolor{red}{(\TVar,\TVar'')}
      }
    }
    \end{array}
  \end{gather}
\end{Newstuff}
This example relies on shape combination with the
$\TShapePair\_\_$ operator: in this
case, the shape of the state captured in the closure and the
shape of the state of newly created channels. Domain variables
are combined accordingly using the $\_,\_$ operator.

A remaining restriction is that the number of channels that are
handled is always fixed at compile time. Lifting this restriction
would go along with support for recursive datatypes, a topic of future
work. 

\subsection{Context Restriction Operators}
\label{sec:cont-restr-oper}

\begin{figure}[tp]
  Removing bindings, which might contain free domain variables
  \begin{align*}
    \CtxRestrictNonDom\Ctx &= \begin{cases}
      \Empty                                  & \text{if } \Ctx = \Empty \\
      \CtxRestrictNonDom{\Ctx'},\TVar : \Kind & \text{if } \Ctx = \Ctx',\TVar : \Kind \land \Kind \in \{ \KShape, \KSession, \KDomain\POLYShape \to \KType, \KDomain\POLYShape \to \KSt \} \\
      \CtxRestrictNonDom{\Ctx'}               & \text{if } \Ctx = \Ctx',\TVar : \Kind \land \Kind \in \{ \KDomain\POLYShape, \KSt, \KType \} \\
      \CtxRestrictNonDom{\Ctx'}               & \text{if } \Ctx = \Ctx',\EVar : \Typ \\
      \CtxRestrictNonDom{\Ctx'}               & \text{if } \Ctx = \Ctx',\POLYDom_1 \Disjoint \POLYDom_2 \\
    \end{cases}
  \end{align*}
  Removing non-domain bindings \hspace{45mm}
  \begin{align*}
    \CtxRestrictOnlyDom{\Ctx} &= \begin{cases}
      \Empty                                               & \text{if } \Ctx = \Empty \\
      \CtxRestrictOnlyDom{\Ctx'},\TVar : \KDomain\POLYShape & \text{if } \Ctx = \Ctx',\TVar : \KDomain\POLYShape \\
      \CtxRestrictOnlyDom{\Ctx'}                            & \text{if } \Ctx = \Ctx',\TVar : \Kind \land \Kind \neq \KDomain\POLYShape \\
      \CtxRestrictOnlyDom{\Ctx'}                            & \text{if } \Ctx = \Ctx',\EVar : \Typ \\
      \CtxRestrictOnlyDom{\Ctx'}                            & \text{if } \Ctx = \Ctx',\POLYDom_1 \Disjoint \POLYDom_2 \\
    \end{cases}
  \end{align*}
  \caption{Context restriction ($\CtxRestrictNonDom\Ctx$ and $\CtxRestrictOnlyDom\Ctx$)}
  \label{fig:context-restriction}
\end{figure}
Figure~\ref{fig:context-restriction} contains the definition of the
context restriction operators, which are mainly technical. Both
operators keep only bindings of type variables. One removes all domain
bindings and the other removes all non-domain bindings.

\subsection{Metatheory}

In this formalization we use inference rule notation to state
lemmas and use proved lemmas in proof trees. While this
notation is unconventional, we found that it significantly improves
readability.

\begin{lemma}[Context Restriction preserves Kind Formation]
  \label{lem:context-restriction-preserves-kind-formation}
  \begin{mathpar}
      \inferrule*[Left={\normalfont(1)}]{
        \PVGRHasKind{\Ctx}{\Typ}{\KShape}
      }{
        \PVGRHasKind{\CtxRestrictNonDom\Ctx}{\Typ}{\KShape}
      }
      \and
      \inferrule*[Left={\normalfont(2)}]{
        \PVGRIsKind{\Ctx}{\Kind}
      }{
        \PVGRIsKind{\CtxRestrictNonDom\Ctx}{\Kind}
      }
      \and
      \inferrule*[Left={\normalfont(3)}]{
        \PVGRIsCtx{\Ctx}
      }{
        \PVGRIsCtx{\CtxRestrictNonDom\Ctx}
      }
  \end{mathpar}
\end{lemma}
\begin{proof}
  \
  \begin{enumerate}
  \item
    Straightforward induction on the derivations with kind $\KShape$.
    The case of type application is not possible, since type lambdas cannot have codomain $\KShape$.
  \item
    Straightforward induction on the kind formation using (1) in the case $\textsc{KF-Dom}$.
  \item
    Straightforward induction on $\PVGRIsCtx{\Ctx}$.
    Since $\CtxRestrictNonDom\Ctx$ removes all value-level and
    constraint bindings, only the well-kindedness of bound type-variables
    needs to be preserved, which follows via (2).
    \qedhere
  \end{enumerate}
\end{proof}



\begin{definition}[Order Preserving Embedding (OPE)]
  \label{def:order-preserving-embedding}
  $\Ctx_2$ is an \emph{Order Preserving Embedding} of $\Ctx_1$, written as $\OPE{\Ctx_1}{\Ctx_2}$, iff
  \begin{enumerate}
  \item $\PVGRIsCtx{\Ctx_1}$
  \item $\PVGRIsCtx{\Ctx_2}$
  \item $\forall b \in \Ctx_1. b \in \Ctx_2$
  \end{enumerate}
\end{definition}

\begin{lemma}[Context Restriction Preserves OPE]
  \label{lem:context-restriction-preserves-ope}
  \begin{align*}
    \inferrule{
      \OPE{\Ctx_1}{\Ctx_2}
    } {
      \OPE{\CtxRestrictNonDom{\Ctx_1}}{\CtxRestrictNonDom{\Ctx_2}}
    }
  \end{align*}
\end{lemma}
\begin{proof}
  Axioms (1) and (2) follow via Lemma~\ref{lem:context-restriction-preserves-kind-formation}.3;
  axiom (3) holds, because if $\CtxRestrictNonDom{\cdot}$
  removes a binding from $\Ctx_2$, then that binding is either not
  present in $\Ctx_1$ or also removed in $\CtxRestrictNonDom{\Ctx_1}$.
\end{proof}

\begin{lemma}[Context Extension Preserves OPE]
  \label{lem:context-extension-preserves-ope}
  \begin{mathpar}
    \inferrule*[Left={\normalfont(1)}]{
      \OPE{\Ctx_1}{\Ctx_2} \\
      \PVGRIsCtx{\Ctx_1,\Ctx_3} \\
      \PVGRIsCtx{\Ctx_2,\Ctx_3}
    } {
      \OPE{(\Ctx_1,\Ctx_3)}{(\Ctx_2,\Ctx_3)}
    }
    \and
    \inferrule*[Left={\normalfont(2)}]{
      \OPE{\Ctx_1}{\Ctx_2} \\
      \PVGRIsCtx{\Ctx_1 \DisjointAppend \Ctx_3} \\
      \PVGRIsCtx{\Ctx_2 \DisjointAppend \Ctx_3}
    } {
      \OPE{(\Ctx_1 \DisjointAppend \Ctx_3)}{(\Ctx_2 \DisjointAppend \Ctx_3)}
    }
  \end{mathpar}
\end{lemma}
\begin{proof}
  \
  \begin{enumerate}
  \item Same as (2) but without the additional complication of constraints.
  \item
    OPE Axiom (1) and (2) follow by assumption.
    For Axiom (3) we need to show that
    \begin{align*}
    \forall b \in \Ctx_1,\Ctx_3,\Cstr_{13},\Cstr_{3}.\ b \in \Ctx_2,\Ctx_3,\Cstr_{23},\Cstr_{3}
    \end{align*}
    where the $\Cstr_i$ are defined as in Figure~\ref{fig:polyvgr-disjoint-ctx-extension}.
    Any binding from $\Ctx_1$ is also contained in $\Ctx_2$, due to Axiom (3) of $\OPE{\Ctx_1}{\Ctx_2}$.
    Hence, it also holds that $\Dom{\Ctx_1} \subseteq \Dom{\Ctx_2}$, which implies $\Cstr_{13} \subseteq \Cstr_{23}$.
    \qedhere
  \end{enumerate}
\end{proof}

\begin{lemma}[Context Formation and Concatenation]
\label{lem:ctx-form-and-cat}
  $\PVGRIsCtx{\Ctx_1,\Ctx_2} \iff \PVGRIsCtx{\Ctx_1\DisjointAppend\Ctx_2}$
\end{lemma}
\begin{proof}
  \
  \begin{itemize}
  \item[$\Leftarrow$:]
    Straightforward, because
    $\Ctx_1\DisjointAppend\Ctx_2 = \Ctx_1,\Ctx_2,\Cstr_{12},\Cstr_2$,
    so we can just split off the constraints from the context
    formation by repeated case analysis.
  \item[$\Rightarrow$:]
    To append the constraints $\Cstr_{12}$ and  $\Cstr_2$ to the context formation,
    we need to repeatedly apply \textsc{CF-ConsCstr}, which requires
    all constraint axioms to use well-kinded domains. By definition of
    $\Cstr_{12}$ and $\Cstr_{2}$, all domains are type variables from
    $\Dom{\Ctx_1} \cup \Dom{\Ctx_2}$, and are hence well-kinded.
    \qedhere
  \end{itemize}
\end{proof}

\begin{lemma}[Weakening]
  \label{lem:weakening}
  If $\OPE{\Ctx_1}{\Ctx_2}$ then
  \begin{mathpar}
    \inferrule*[Left={\normalfont(1)}]{
      \PVGRCstrEntail{\Ctx_1}{\Cstr}
    }{
      \PVGRCstrEntail{\Ctx_2}{\Cstr}
    }
    \and
    \inferrule*[Left={\normalfont(2)}]{
      \PVGRIsCtx{\Ctx_1,\Ctx_3}
    }{
      \PVGRIsCtx{\Ctx_2,\Ctx_3}
    }
    \and
    \inferrule*[Left={\normalfont(3)}]{
      \PVGRIsCtx{\Ctx_1\DisjointAppend\Ctx_3}
    }{
      \PVGRIsCtx{\Ctx_2\DisjointAppend\Ctx_3}
    }
    \and
    \inferrule*[Left={\normalfont(4)}]{
      \PVGRIsKind{\Ctx_1}{\Kind}
    }{
      \PVGRIsKind{\Ctx_2}{\Kind}
    }
    \and
    \inferrule*[Left={\normalfont(5)}]{
      \PVGRHasKind{\Ctx_1}{\Typ}{\Kind}
    }{
      \PVGRHasKind{\Ctx_2}{\Typ}{\Kind}
    }
    \\
    \inferrule*[Left={\normalfont(6)}]{
      \PVGRHasValType{\Ctx_1}\Val\Typ
    }{
      \PVGRHasValType{\Ctx_2}\Val\Typ
    }
    \and
    \inferrule*[Left={\normalfont(7)}]{
      \PVGRHasExpType{\Ctx_1}{\St_1}\Exp{\Ctx_3}{\St_2}{\Typ_2}
    }{
      \PVGRHasExpType{\Ctx_2}{\St_1}\Exp{\Ctx_3}{\St_2}{\Typ_2}
    }
    \and
    \inferrule*[Left={\normalfont(8)}]{
      \PVGRHasConfType{\Ctx_1}{\St}{\Cfg}
    }{
      \PVGRHasConfType{\Ctx_2}{\St}{\Cfg}
    }
  \end{mathpar}
\end{lemma}
\begin{proof}
  By mutual induction on the derivation to be weakened:
  \begin{enumerate}
  \item
    \begin{itemize}
    \item \IndCase{CE-Axiom} Direct consequence of Axiom~(3) of $\OPE{\Ctx_1}{\Ctx_2}$.
    \item All other cases are immediate by the induction hypotheses.
    \end{itemize}
  \item By induction on $\Ctx_3$:
    \begin{itemize}
    \item \IndCase{$\Ctx_3 = \Empty$}
      Immediate from Axiom (2) of $\OPE{\Ctx_1}{\Ctx_2}$.
    \item \IndCase{$\Ctx_3 = \Ctx_3',\TVar:\Kind$}
      In this case we have
      \begin{align*}
        \inferrule*[Right=CF-ConsKind]{
          \PVGRIsCtx{\Ctx_1,\Ctx_3'} \\
          \PVGRHasKind{\Ctx_1,\Ctx_3'}{\TVar}{\Kind}
        }{
          \PVGRIsCtx{\Ctx_1,\Ctx_3',\TVar:\Kind}
        }
      \end{align*}
      and need to show
      \begin{align*}
        \inferrule*[Right=CF-ConsKind]{
          \inferrule*[Left=(a)]{ }{
            \PVGRIsCtx{\Ctx_2,\Ctx_3'}
          } \\
          \inferrule*[Right=(b)]{ }{
            \PVGRHasKind{\Ctx_2,\Ctx_3'}{\TVar}{\Kind}
          }
        }{
          \PVGRIsCtx{\Ctx_2,\Ctx_3',\TVar:\Kind}
        }
      \end{align*}
      (\textsc{a}) follows from the inner induction hypothesis;
      (\textsc{b}) follows from the outer induction hypothesis for (5).
    \item The cases for value-level and constraint bindings are analogous to the previous case.
    \end{itemize}

  \item Follows by applying Lemma~\ref{lem:ctx-form-and-cat} to both
    the premise and conclusion of (2).

  \item All cases are immediate by the induction hypotheses.
  \item
    \begin{itemize}
    \item \IndCase{K-Var}
      Follows directly from Axiom (3) of $\OPE{\Ctx_1}{\Ctx_2}$.
    \item \IndCase{K-Lam, K-Send, K-Recv}
      Here we have assumptions using context restriction, like
      $\PVGRHasKind{\CtxRestrictNonDom{\Ctx_1},\TVar:\KDomain{\POLYShape}}{\St}{\KSt}$.
      In order to apply the induction hypothesis on those assumptions, we rely on
      Lemma~\ref{lem:context-restriction-preserves-ope}, which given
      $\OPE{\Ctx_1}{\Ctx_2}$ yields
      $\OPE{\CtxRestrictNonDom{\Ctx_1}}{\CtxRestrictNonDom{\Ctx_2}}$.
    \item \IndCase{K-Arr}
      Here we have assumptions using disjoint context concatenation, like
      $\PVGRHasKind{\Ctx_1 \DisjointAppend \Ctx_2'}{\St_2}{\KSt}$.
      In order to apply the induction hypothesis on those assumptions, we rely on
      Lemma~\ref{lem:context-extension-preserves-ope}.(2), which given
      $\OPE{\Ctx_1}{\Ctx_2}$ yields
      $\OPE{(\Ctx_1 \DisjointAppend \Ctx_2')}{(\Ctx_2 \DisjointAppend \Ctx_2')}$.
    \item All other cases are immediate by the induction hypotheses.
      Going under binders requires the Barendregt Convention and
      Lemma~\ref{lem:context-extension-preserves-ope}.(1) to extend the OPE.
    \end{itemize}
  \item
    \begin{itemize}
    \item \IndCase{T-Var}
      Same as \textsc{K-Var}.
    \item All other cases are immediate by the induction hypotheses.
      Going under binders requires the Barendregt Convention and
      Lemma~\ref{lem:context-extension-preserves-ope}.(1) to extend the OPE.
    \end{itemize}
  \item
    \begin{itemize}
    \item \IndCase{T-Let}
      Same as \textsc{K-Arr}.
    \item All other cases are immediate by the induction hypotheses.
    \end{itemize}
  \item
    \begin{itemize}
    \item \IndCase{T-NuChan}
      Same as \textsc{K-Arr}.
    \item All other cases are immediate by the induction hypotheses.
      Going under binders requires the Barendregt Convention and
      Lemma~\ref{lem:context-extension-preserves-ope}.(1) to extend the OPE.
      \qedhere
    \end{itemize}
  \end{enumerate}
\end{proof}


\newcommand*\Sub{\sigma}
\newcommand*\SubId{\textrm{id}}
\newcommand*\PVGRHasSubstType[3]{\vdash {#1} : {#2} \Rightarrow {#3}}

\begin{definition}[Substitution Typing]
  \label{def:sub-typing}
  We write $\PVGRHasSubstType{\Sub}{\Ctx_1}{\Ctx_2}$ iff
  \begin{enumerate}
  \item $\PVGRIsCtx{\Ctx_1}$
  \item $\PVGRIsCtx{\Ctx_2}$
  \item $\forall (\EVar:\Typ) \in \Ctx_1.\ \PVGRHasValType{\Ctx_2}{\Sub\EVar}{\Sub\Typ}$,
  \item $\forall (\TVar:\Kind) \in \Ctx_1.\ \PVGRHasKind{\Ctx_2}{\Sub\TVar}{\Sub\Kind}$, and
  \item $\forall (\POLYDom_1 \Disjoint \POLYDom_2) \in \Ctx_1.\ \PVGRCstrEntail{\Ctx_2}{\Sub\POLYDom_1 \Disjoint \Sub\POLYDom_2}$.
  \end{enumerate}
\end{definition}

\begin{lemma}[Typing of the Identity Substitution]
  \label{lem:sub-id-typing}
  If $\PVGRIsCtx\Ctx$, then $\PVGRHasSubstType{\SubId}{\Ctx}{\Ctx}$.
\end{lemma}
\begin{proof}
  (1) and (2) follow by assumption.

  (3), (4), and (5) follow via \textsc{T-Var}, \textsc{K-Var}, and \textsc{CE-Axiom},
  respectively.
  \qedhere

\end{proof}

\newpage

\begin{lemma}[Extending Substitution Typings]
  \label{lem:sub-ext-typing}
  \
  \begin{enumerate}
  \item
    $
    \inferrule{
      \PVGRHasSubstType
        {\Sub}
        {\Ctx_1}
        {\Ctx_2}
      \\
      \PVGRHasKind
        {\Ctx_1}
        {\Typ}
        {\KType}
      \\
      \PVGRHasValType
        {\Ctx_2}
        {\Val}
        {\Sub\Typ}
      \\
      \EVar \not\in \Dom{\Ctx_1}
    }{
      \PVGRHasSubstType
        {(\Sub,\EVar\mapsto\Val)}
        {(\Ctx_1,\EVar:\Typ)}
        {\Ctx_2}
    }
    $
  \item
    $
    \inferrule{
      \PVGRHasSubstType
        {\Sub}
        {\Ctx_1}
        {\Ctx_2}
      \\
      \PVGRIsKind
        {\Ctx_1}
        {\Kind}
      \\
      \PVGRHasKind
        {\Ctx_2}
        {\Typ}
        {\Sub\Kind}
      \\
      \TVar \not\in \Dom{\Ctx_1}
    }{
      \PVGRHasSubstType
        {(\Sub,\TVar\mapsto\Typ)}
        {(\Ctx_1,\TVar:\Kind)}
        {\Ctx_2}
    }
    $
  \item
    $
    \inferrule{
      \PVGRHasSubstType
        {\Sub}
        {\Ctx_1}
        {\Ctx_2}
      \\
      \PVGRIsCtx
        {\Ctx_1,\Cstr}
      \\
      \PVGRCstrEntail
        {\Ctx_2}
        {\Sub\Cstr}
    }{
      \PVGRHasSubstType
        {\Sub}
        {(\Ctx_1,\Cstr)}
        {\Ctx_2}
    }
    $
  \end{enumerate}
\end{lemma}
\begin{proof}
  Straightforward case analysis and rule applications.
  \qedhere
\end{proof}


\begin{lemma}[Weakening Substitution Typings]
  \label{lem:sub-weaken-typing}
    $$
    \inferrule{
      \PVGRHasSubstType
        {\Sub}
        {\Ctx_1}
        {\Ctx_2}
      \\
      \PVGRIsCtx{\Ctx_2,\Ctx_2'}
    }{
      \PVGRHasSubstType
        {\Sub}
        {\Ctx_1}
        {\Ctx_2,\Ctx_2'}
    }
    $$
\end{lemma}
\begin{proof}
  By Definition~\ref{def:sub-typing}, we have to prove:
  \begin{enumerate}
  \item $\PVGRIsCtx{\Ctx_1}$, which follows by
    $\PVGRHasSubstType{\Sub}{\Ctx_1}{\Ctx_2}$.
  \item $\PVGRIsCtx{\Ctx_2,\Ctx_2'}$, which follows by assumption.
  \item $\forall (\EVar:\Typ) \in \Ctx_1.\
    \PVGRHasValType{\Ctx_2,\Ctx_2'}{\Sub\EVar}{\Sub\Typ}$.
    Let $(\EVar:\Typ) \in \Ctx_1$, then by
    $\PVGRHasSubstType{\Sub}{\Ctx_1}{\Ctx_2}$ it follows that
    $$\PVGRHasValType{\Ctx_2}{\Sub\EVar}{\Sub\Typ}$$
    which by weakening via Lemma~\ref{lem:weakening} yields
    $$\PVGRHasValType{\Ctx_2,\Ctx_2'}{\Sub\EVar}{\Sub\Typ}$$
  \item $\forall (\TVar:\Kind) \in \Ctx_1.\
    \PVGRHasKind{\Ctx_2,\Ctx_2'}{\Sub\TVar}{\Sub\Kind}$,
    which follows similarly by weakening.
  \item $\forall (\POLYDom_1 \Disjoint \POLYDom_2) \in \Ctx_1.\
    \PVGRCstrEntail{\Ctx_2,\Ctx_2'}{\Sub\POLYDom_1 \Disjoint \Sub\POLYDom_2}$,
    which follows similarly by weakening.
    \qedhere
  \end{enumerate}
\end{proof}

\begin{lemma}[Lifting Substitution Typings]
  \label{lem:sub-lift-typing}
  \
  \begin{enumerate}
  \item
    $
    \inferrule{
      \PVGRHasSubstType
        {\Sub}
        {\Ctx_1}
        {\Ctx_2}
      \\
      \PVGRHasKind
        {\Ctx_1}
        {\Typ}
        {\KType}
      \\
      \PVGRHasKind
        {\Ctx_2}
        {\Sub\Typ}
        {\KType}
      \\
      \EVar \not\in \Dom{\Ctx_1}
      \\
      \EVar \not\in \Dom{\Ctx_2}
    }{
      \PVGRHasSubstType
        {(\Sub,\EVar\mapsto\EVar)}
        {(\Ctx_1,\EVar:\Typ)}
        {(\Ctx_2,\EVar:\Sub\Typ)}
    }
    $
  \item
    $
    \inferrule{
      \PVGRHasSubstType
        {\Sub}
        {\Ctx_1}
        {\Ctx_2}
      \\
      \PVGRIsKind
        {\Ctx_1}
        {\Kind}
      \\
      \PVGRIsKind
        {\Ctx_2}
        {\Sub\Kind}
      \\
      \TVar \not\in \Dom{\Ctx_1}
      \\
      \TVar \not\in \Dom{\Ctx_2}
    }{
      \PVGRHasSubstType
        {(\Sub,\TVar\mapsto\TVar)}
        {(\Ctx_1,\TVar:\Kind)}
        {(\Ctx_2,\TVar:\Sub\Kind)}
    }
    $
  \item
    $
    \inferrule{
      \PVGRHasSubstType
        {\Sub}
        {\Ctx_1}
        {\Ctx_2}
      \\
      \PVGRIsCtx
        {\Ctx_1,\Cstr}
      \\
      \PVGRIsCtx
        {\Ctx_2,\Sub\Cstr}
    }{
      \PVGRHasSubstType
        {\Sub}
        {(\Ctx_1,\Cstr)}
        {(\Ctx_2,\Cstr)}
    }
    $
  \item
    $
    \inferrule{
      \PVGRHasSubstType
        {\Sub}
        {\Ctx_1}
        {\Ctx_2}
      \\
      \PVGRIsCtx
        {\Ctx_1,\Ctx}
      \\
      \PVGRIsCtx
        {\Ctx_2,\Sub\Ctx}
    }{
      \PVGRHasSubstType
        {(\Sub,\SubId)}
        {(\Ctx_1,\Ctx)}
        {(\Ctx_2,\Sub\Ctx)}
    }
    $
  \item
    $
    \inferrule{
      \PVGRHasSubstType
        {\Sub}
        {\Ctx_1}
        {\Ctx_2}
      \\
      \PVGRIsCtx
        {\Ctx_1 \DisjointAppend \Ctx}
      \\
      \PVGRIsCtx
        {\Ctx_2 \DisjointAppend \Sub\Ctx}
    }{
      \PVGRHasSubstType
        {(\Sub,\SubId)}
        {(\Ctx_1 \DisjointAppend \Ctx)}
        {(\Ctx_2 \DisjointAppend \Sub\Ctx)}
    }
    $
  \end{enumerate}
\end{lemma}
\begin{proof}
  \
  \begin{enumerate}
  \item
    First, we weaken the substitution typing
    \begin{align*}
      \inferrule*[Right=Lemma~\ref{lem:sub-weaken-typing}]{
        \PVGRHasSubstType
          {\Sub}
          {\Ctx_1}
          {\Ctx_2}
        \\
        \inferrule*[Right=CF-ConsType]{
          \PVGRIsCtx{\Ctx_2} \\
          \PVGRHasKind{\Ctx_2}{\Sub\Typ}{\KType} \\
          \EVar \not\in \Dom{\Ctx_2}
        }{
          \PVGRIsCtx{\Ctx_2,\EVar:\Sub\Typ}
        }
      }{
        \PVGRHasSubstType
          {\Sub}
          {\Ctx_1}
          {\Ctx_2,\EVar:\Sub\Typ}
      }
    \end{align*}
    Then we extend the weakened substitution typing
    \begin{align*}
      \inferrule*[Right=Lemma~\ref{lem:sub-ext-typing}]{
        \PVGRHasSubstType
          {\Sub}
          {\Ctx_1}
          {\Ctx_2,\EVar:\Sub\Typ}
        \\
        \PVGRHasKind{\Ctx_1}{\Typ}{\KType}
        \\
        \inferrule*[Right=\textsc{T-Var}]{ }{
          \PVGRHasValType{\Ctx_2,\EVar:\Sub\Typ}{\EVar}{\Sub\Typ}
        }
        \\
        \EVar \not\in \Dom{\Ctx_1}
      }{
        \PVGRHasSubstType
          {\Sub,\EVar\mapsto\EVar}
          {\Ctx_1,\EVar:\Typ}
          {\Ctx_2,\EVar:\Sub\Typ}
      }
    \end{align*}
  \item Same as (1).
  \item
    First, we weaken the substitution typing
    \begin{align*}
      \inferrule*[Right=Lemma~\ref{lem:sub-weaken-typing}]{
        \PVGRHasSubstType
          {\Sub}
          {\Ctx_1}
          {\Ctx_2}
        \\
        \PVGRIsCtx{\Ctx_2,\Sub\Cstr}
      }{
        \PVGRHasSubstType
          {\Sub}
          {\Ctx_1}
          {\Ctx_2,\Sub\Cstr}
      }
    \end{align*}
    Then we extend the weakened substitution typing
    \begin{align*}
      \inferrule*[Right=Lemma~\ref{lem:sub-ext-typing}]{
        \PVGRHasSubstType
          {\Sub}
          {\Ctx_1}
          {\Ctx_2,\Sub\Cstr}
        \\
        \PVGRIsCtx{\Ctx_1,\Cstr}
        \\
        \PVGRCstrEntail{\Ctx_2,\Sub\Cstr}{\Sub\Cstr}
      }{
        \PVGRHasSubstType
          {\Sub,\EVar\mapsto\EVar}
          {\Ctx_1,\Cstr}
          {\Ctx_2,\Sub\Cstr}
      }
    \end{align*}
    where $\PVGRCstrEntail{\Ctx_2,\Sub\Cstr}{\Sub\Cstr}$ follows via repeated applications of
    \textsc{CE-Split}, \textsc{CE-Axiom}, and \textsc{CE-Merge}.
  \item Follows from (1) to (3) by induction on $\Ctx$.
  \item
    In this case we have
    \begin{align*}
      \Ctx_1\DisjointAppend\Ctx &= \Ctx_1,\Ctx,\Cstr_1,\Cstr_1' \\
      \Ctx_2\DisjointAppend\Sub\Ctx &= \Ctx_2,\Sub\Ctx,\Cstr_2,\Cstr_2'
    \end{align*}
    with
    \begin{align*}
      \Cstr_1 &=
        \{ \TVar_1 \Disjoint \TVar_2 \mid
           \TVar_1, \TVar_2 \in \Dom{\CtxRestrictOnlyDom{\Ctx}},
           \TVar_1 \neq \TVar_2 \}
      \\
      \Cstr_2 &=
        \{ \TVar_1 \Disjoint \TVar_2 \mid
           \TVar_1, \TVar_2 \in \Dom{\CtxRestrictOnlyDom{\Sub\Ctx}},
           \TVar_1 \neq \TVar_2 \}
      \\
      \Cstr_1' &=
        \{ \TVar_1 \Disjoint \TVar_2 \mid
           \TVar_1 \in \Dom{\CtxRestrictOnlyDom{\Ctx_1}},
           \TVar_2 \in  \Dom{\CtxRestrictOnlyDom{\Ctx}} \}
      \\
      \Cstr_2' &=
        \{ \TVar_1 \Disjoint \TVar_2 \mid
           \TVar_1 \in \Dom{\CtxRestrictOnlyDom{\Ctx_2}},
           \TVar_2 \in  \Dom{\CtxRestrictOnlyDom{\Sub\Ctx}} \}
    \end{align*}

    Via (4) follows
    \begin{align*}
      \PVGRHasSubstType
        {\Sub}
        {\Ctx_1,\Ctx}
        {\Ctx_2,\Sub\Ctx}
    \end{align*}

    We have $\Sub\Cstr_1 = \Cstr_1$, because $\Cstr_1$
    contains by definition only variables from $\Ctx$, so $\Sub$ behaves as the identity.
    Furthermore, we have $\Cstr_1 = \Cstr_2$, because $\Dom{\CtxRestrictOnlyDom{\Sub\Ctx}} = \Dom{\CtxRestrictOnlyDom{\Ctx}}$.
    Hence, we can apply (3) to the previous result and rewrite
    $\Sub\Cstr_1$ to $\Cstr_2$, which yields
    \begin{align*}
      \PVGRHasSubstType
        {\Sub}
        {\Ctx_1,\Ctx,\Cstr_1}
        {\Ctx_2,\Sub\Ctx,\Cstr_2}
    \end{align*}

    To conclude with
    \begin{align*}
      \PVGRHasSubstType
        {\Sub}
        {\Ctx_1,\Ctx,\Cstr_1,\Cstr_1'}
        {\Ctx_2,\Sub\Ctx,\Cstr_2,\Cstr_2'}
    \end{align*}
    we need to show that for any
    $(\POLYDom_1 \Disjoint \POLYDom_2) \in \Ctx_1,\Ctx,\Cstr_1,\Cstr_1'$
    it holds that
    \begin{align*}
        \PVGRCstrEntail
          {\Ctx_2,\Sub\Ctx,\Cstr_2,\Cstr_2'}
          {\Sub\POLYDom_1 \Disjoint \Sub\POLYDom_2}
    \end{align*}
    If the constraint axiom is in $\Ctx_1,\Ctx,\Cstr_1$, then the result
    follows by the previous substitution typing and weakening.
    If the constraint axiom is in $\Cstr_1'$, then
    $\Sub\POLYDom_1$ contains only variables from $\CtxRestrictOnlyDom{\Ctx_2}$
    and
    $\Sub\POLYDom_2$ contains only variables from $\CtxRestrictOnlyDom{\Sub\Ctx}$,
    so $(\Sub\POLYDom_1 \Disjoint \Sub\POLYDom_2)$ is part of
    $\Cstr_2'$ and can be proved via \textsc{CE-Axiom}.
    \qedhere
  \end{enumerate}
\end{proof}

\begin{lemma}[Context Restriction preserves Substitution Typing]
  \label{lem:context-restriction-preserves-sub-typing}
  \begin{align*}
    \inferrule{
      \PVGRHasSubstType{\Sub}{\Ctx_1}{\Ctx_2}
    }{
      \PVGRHasSubstType{\Sub}{\CtxRestrictNonDom{\Ctx_1}}{\CtxRestrictNonDom{\Ctx_2}}
    }
  \end{align*}
\end{lemma}
\begin{proof}
  \
  \begin{itemize}
  \item
  Axioms (1) and (2) follow via
  Lemma~\ref{lem:context-restriction-preserves-kind-formation}.3.
  \item
  Axiom (3) and (5) hold trivially, since $\CtxRestrictNonDom\cdot$
  removes all value-level and constraint bindings.
  \item
  For Axiom (4), let $(\TVar : \Kind) \in \CtxRestrictNonDom{\Ctx_1}$.
  From $\PVGRHasSubstType{\Sub}{\Ctx_1}{\Ctx_2}$ we know
  $\PVGRHasKind{\Ctx_2}{\Sub\TVar}{\Sub\Kind}$, but we need to prove
  $\PVGRHasKind{\CtxRestrictNonDom{\Ctx_2}}{\Sub\TVar}{\Sub\Kind}$.
  By definition of $\CtxRestrictNonDom{\cdot}$, we know that
  $$\Kind \in \{ \KShape, \KSession, \KDomain\POLYShape \to \KType, \KDomain\POLYShape \to \KSt \}.$$
  Types of those kinds, like $\Sub\TVar$, have free type variables only at positions, which
  are themselves restricted with $\CtxRestrictNonDom{\cdot}$, so
  $\PVGRHasKind{\CtxRestrictNonDom{\Ctx_2}}{\Sub\TVar}{\Sub\Kind}$ is a
  valid strenghtening of $\PVGRHasKind{\Ctx_2}{\Sub\TVar}{\Sub\Kind}$.
  \qedhere
  \end{itemize}
\end{proof}

\begin{lemma}[Substitution Preserves Derivations]
  \label{lem:sub-pres-typing}
  If $\PVGRHasSubstType{\Sub}{\Ctx_1}{\Ctx_2}$, then
  \begin{mathpar}
    \inferrule*[Left={\normalfont(1)}]{
      \PVGRCstrEntail{\Ctx_1}{\Cstr}
    }{
      \PVGRCstrEntail{\Ctx_2}{\Sub\Cstr}
    }
    \and
    \inferrule*[Left={\normalfont(2)}]{
      \PVGRIsCtx{\Ctx_1,\Ctx_3}
    }{
      \PVGRIsCtx{\Ctx_2,\Sub\Ctx_3}
    }
    \and
    \inferrule*[Left={\normalfont(3)}]{
      \PVGRIsCtx{\Ctx_1\DisjointAppend\Ctx_3}
    }{
      \PVGRIsCtx{\Ctx_2\DisjointAppend\Sub\Ctx_3}
    }
    \and
    \inferrule*[Left={\normalfont(4)}]{
      \PVGRIsKind{\Ctx_1}{\Kind}
    }{
      \PVGRIsKind{\Ctx_2}{\Sub\Kind}
    }
    \and
    \inferrule*[Left={\normalfont(5)}]{
      \PVGRHasKind{\Ctx_1}{\Typ}{\Kind} \and
    }{
      \PVGRHasKind{\Ctx_2}{\Sub\Typ}{\Sub\Kind}
    }
    \\
    \inferrule*[Left={\normalfont(6)}]{
      \PVGRHasValType{\Ctx_1}{\Val}{\Typ}
    }{
      \PVGRHasValType{\Ctx_2}{\Sub\Val}{\Sub\Typ}
    }
    \and
    \inferrule*[Left={\normalfont(7)}]{
      \PVGRHasKind{\Ctx_1}{\St_1}{\KSt} \and
      \PVGRHasExpType{\Ctx_1}{\St_1}{\Exp}{\Ctx_3}{\St_2}{\Typ}
    }{
      \PVGRHasExpType{\Ctx_2}{\Sub\St_1}{\Sub\Exp}{\Sub\Ctx_3}{\Sub\St_2}{\Sub\Typ}
    }
    \and
    \inferrule*[Left={\normalfont(8)}]{
      \PVGRTypeConv{\Typ_1}{\Typ_2}
    }{
      \PVGRTypeConv{\Sub\Typ_1}{\Sub\Typ_2}
    }
  \end{mathpar}
\end{lemma}
\begin{proof}
  By mutual induction on the derivations subject to substitution:
  \begin{enumerate}
  \item By induction on $\PVGRCstrEntail{\Ctx_1}{\Cstr}$:
    \begin{itemize}
    \item \IndCase{CE-Axiom}
      Here we have $\Ctx_1 = \Ctx_1',\POLYDom_1 \Disjoint \POLYDom_2$ and
      $$\PVGRCstrEntail{\Ctx_1',\POLYDom_1 \Disjoint \POLYDom_2}{\POLYDom_1 \Disjoint \POLYDom_2}$$
      We need to show
      $$\PVGRCstrEntail{\Ctx_2}{\Sub\POLYDom_1 \Disjoint \Sub\POLYDom_2}$$
      which follows immediately from Axiom (5) of the substitution typing:
      \begin{align*}
        \forall (\POLYDom_1 \Disjoint \POLYDom_2) \in \Ctx_1.\ \PVGRCstrEntail{\Ctx_2}{\Sub\POLYDom_1 \Disjoint \Sub\POLYDom_2}
      \end{align*}
    \item \IndCase{CE-Sym, CE-Empty, CE-Split, CE-Merge, CE-Empty, CE-Cons}
      Immediate from the induction hypotheses.
    \end{itemize}
  \item
    Analogous to the corresponding case in Lemma~\ref{lem:weakening} (Weakening).
  \item Follows by applying Lemma~\ref{lem:ctx-form-and-cat} to both
    the premise and conclusion of (2).
  \item
    \begin{itemize}
    \item \IndCase{KF-Type, KF-Session, KF-State, KF-Shape}
      Trivial.
    \item \IndCase{KF-Dom, KF-Arr}
      Immediate from the induction hypotheses.
    \end{itemize}
  \item
    \begin{itemize}
    \item \IndCase{K-Var}
      Immediate from Axiom (4) of the substitution typing.
    \item \IndCase{K-App}
      Immediate from the induction hypotheses.
    \item \IndCase{K-Lam}
      We first apply the induction hypothesis to the first subderivation
      $$\PVGRHasKind{\Ctx_1}{\POLYShape}{\KShape}$$
      which yields
      $$\PVGRHasKind{\Ctx_2}{\Sub\POLYShape}{\KShape}.$$

      To be able to apply the induction hypothesis to the second subderivation,
      we first lift the substitution typing and kindings over the context restriction
      \begin{mathpar}
        \inferrule*[right=Lemma~\ref{lem:context-restriction-preserves-sub-typing}]{
          \PVGRHasSubstType{\Sub}{\Ctx_1}{\Ctx_2}
        }{
          \PVGRHasSubstType{\Sub}{\CtxRestrictNonDom{\Ctx_1}}{\CtxRestrictNonDom{\Ctx_2}}
        }
        \and
        \inferrule*[right=Lemma~\ref{lem:context-restriction-preserves-kind-formation}]{
          \inferrule*[Right=KF-Dom]{
            \PVGRHasKind{\Ctx_1}{\POLYShape}{\KShape}
          }{
            \PVGRIsKind{\Ctx_1}{\KDomain\POLYShape}
          }
        }{
          \PVGRIsKind{\CtxRestrictNonDom{\Ctx_1}}{\KDomain\POLYShape}
        }
        \and
        \inferrule*[Right=Lemma~\ref{lem:context-restriction-preserves-kind-formation}]{
          \inferrule*[Right=KF-Dom]{
            \PVGRHasKind{\Ctx_2}{\Sub\POLYShape}{\KShape}
          }{
            \PVGRIsKind{\Ctx_2}{\KDomain{\Sub\POLYShape}}
          }
        }{
          \PVGRIsKind{\CtxRestrictNonDom{\Ctx_2}}{\KDomain{\Sub\POLYShape}}
        }
      \end{mathpar}
      and then lift the substitution typing over the new domain binding
      \begin{align*}
        \inferrule*[Right=Lemma~\ref{lem:sub-lift-typing}]{
          \PVGRHasSubstType{\Sub}{\CtxRestrictNonDom{\Ctx_1}}{\CtxRestrictNonDom{\Ctx_2}}
          \and
          \PVGRIsKind{\CtxRestrictNonDom{\Ctx_1}}{\KDomain\POLYShape}
          \and
          \PVGRIsKind{\CtxRestrictNonDom{\Ctx_2}}{\KDomain{\Sub\POLYShape}}
          \and
          \TVar \not\in \Dom{\Ctx_1},\Dom{\Ctx_2}
        }{
          \PVGRHasSubstType
            {\Sub}
            {\CtxRestrictNonDom{\Ctx_1},\TVar:\KDomain\POLYShape}
            {\CtxRestrictNonDom{\Ctx_2},\TVar:\KDomain{\Sub\POLYShape}}
        }
      \end{align*}
      where the fourth assumption follows via the Barendregt convention.

      The result then follows by reconstructing the \textsc{K-Lam} rule.
    \item \IndCase{K-All}
      For the first subderivation, we can directly apply the induction
      hypotheses and obtain
      $\PVGRIsCtx{\Ctx_2,\TVar:\Sub\Kind,\Sub\Cstr}$.
      For the second subderivation, we have to lift the substitution typing
      \begin{align*}
        \inferrule*[Right=Lemma~\ref{lem:sub-lift-typing}]{
          \PVGRHasSubstType{\Sub}{\Ctx_1}{\Ctx_2} \\
          \PVGRIsCtx{\Ctx_1,\TVar:\Kind,\Cstr} \\
          \PVGRIsCtx{\Ctx_2,\TVar:\Sub\Kind,\Sub\Cstr} \\
          \TVar \not\in \Dom{\Ctx_1},\Dom{\Ctx_2}
        }{
          \PVGRHasSubstType{\Sub}{\Ctx_1,\TVar:\Kind,\Cstr}{\Ctx_2,\TVar:\Sub\Kind,\Sub\Cstr}
        }
      \end{align*}
      where the fourth assumption follows from the Barendregt Convention.

      The result then follows by reconstructing the \textsc{K-All} rule.
    \item \IndCase{K-Arr}
      For the premises
      \begin{mathpar}
        \PVGRHasKind{\Ctx_1}{\St_1}{\KSt} \and
        \PVGRHasKind{\Ctx_1}{\Typ_1}{\KType} \and
        \PVGRIsCtx{\Ctx_1\DisjointAppend\Ctx}
      \end{mathpar}
      we directly apply the induction hypothesis and obtain
      \begin{mathpar}
        \PVGRHasKind{\Ctx_2}{\Sub\St_1}{\KSt} \and
        \PVGRHasKind{\Ctx_2}{\Sub\Typ_1}{\KType} \and
        \PVGRIsCtx{\Ctx_2\DisjointAppend\Sub\Ctx}
      \end{mathpar}

      For the premises
      \begin{mathpar}
        \PVGRHasKind{\Ctx_1\DisjointAppend\Ctx}{\St_2}{\KSt} \and
        \PVGRHasKind{\Ctx_1\DisjointAppend\Ctx}{\Typ_2}{\KType}
      \end{mathpar}
      we first lift the substitution typing
      \begin{align*}
        \inferrule*[Right=Lemma~\ref{lem:sub-lift-typing}.5]{
          \PVGRHasSubstType{\Sub}{\Ctx_1}{\Ctx_2} \\
          \PVGRIsCtx{\Ctx_1\DisjointAppend\Ctx} \\
          \PVGRIsCtx{\Ctx_2\DisjointAppend\Sub\Ctx}
        }{
          \PVGRHasSubstType{\Sub}{\Ctx_1\DisjointAppend\Ctx}{\Ctx_2\DisjointAppend\Sub\Ctx}
        }
      \end{align*}
      and then apply the induction hypothesis to obtain
      \begin{mathpar}
          \PVGRHasKind{\Ctx_2 \DisjointAppend \Sub\Ctx}{\Sub\St_2}{\KSt} \and
          \PVGRHasKind{\Ctx_2 \DisjointAppend \Sub\Ctx}{\Sub\Typ_2}{\KType}
      \end{mathpar}
      Finally, we reconstruct the \textsc{K-Arr} rule from the above results.
      \item \IndCase{K-Send, K-Recv}
        Same as \textsc{K-Lam}.
      \item The other cases follow immediately from the induction hypotheses.
    \end{itemize}
  \item
    \begin{itemize}
    \item \IndCase{T-Var}
      Immediate from Axiom (3) of the substitution typing.
    \item The other cases follow immediately from the induction hypotheses
      using Lemma~\ref{lem:sub-lift-typing} and the
      Barendregt convention to lift the substitution typing when going under
      binders.
    \end{itemize}
  \item
    \begin{itemize}
    \item \IndCase{T-TApp}
      In this case we have the premise
      $\PVGRTypeConv{\Subst{\Typ'}{\TVar}\Typ}{\Typ''}$
      for which the induction hypothesis yields
      $\PVGRTypeConv{\Sub(\Subst{\Typ'}{\TVar}\Typ)}{\Sub\Typ''}$
      which is equivalent to the required conclusion
      $\PVGRTypeConv{\Subst{\Sub\Typ'}{\TVar}(\Sub\Typ)}{\Sub\Typ''}$
      due to the Barendregt convention.
    \item \IndCase{T-Send}
      Same as \textsc{T-TApp}.
    \item \IndCase{T-Case}
      Here we have the assumption
      $\PVGRHasKind{\Ctx_1}{\St_1,\StBind\POLYDom{\SBranch{\Ses_1}{\Ses_2}}}{\KSt}$.
      In order to apply the induction hypothesis to the branch expressions, we need
      to prove
      \begin{mathpar}
        \PVGRHasKind{\Ctx_1}{\St_1,\StBind\POLYDom{\Ses_1}}{\KSt} \and
        \PVGRHasKind{\Ctx_1}{\St_1,\StBind\POLYDom{\Ses_2}}{\KSt}
      \end{mathpar}
      which follow by simple case analysis of the assumption and \textsc{K-StMerge}.
    \item The other cases follow immediately from the induction hypotheses
      using Lemma~\ref{lem:sub-lift-typing} and the
      Barendregt convention to lift the substitution typing when going under
      binders.
    \end{itemize}
  \item Straightforward induction due to the Barendregt convention.
    \qedhere
  \end{enumerate}
\end{proof}

\begin{lemma}[Removal of Implied Constraints]
\label{lem:cstr-removal-pres-typing}
  If $\PVGRCstrEntail{\Ctx}{\Cstr}$, then
  \begin{mathpar}
    \inferrule*[Left={\normalfont(1)}]{
      \PVGRCstrEntail{\Ctx,\Cstr}{\Cstr'}
    }{
      \PVGRCstrEntail{\Ctx}{\Cstr'}
    }
    \and
    \inferrule*[Left={\normalfont(2)}]{
      \PVGRIsKind{\Ctx,\Cstr}{\Kind}
    }{
      \PVGRIsKind{\Ctx}{\Kind}
    }
    \and
    \inferrule*[Left={\normalfont(3)}]{
      \PVGRHasKind{\Ctx,\Cstr}{\Typ}{\Kind}
    }{
      \PVGRHasKind{\Ctx}{\Typ}{\Kind}
    }
    \and
    \inferrule*[Left={\normalfont(4)}]{
      \PVGRHasValType{\Ctx,\Cstr}{\Val}{\Typ}
    }{
      \PVGRHasValType{\Ctx}{\Val}{\Typ}
    }
    \and
    \inferrule*[Left={\normalfont(5)}]{
      \PVGRHasExpType{\Ctx,\Cstr}{\St_1}{\Exp}{\Ctx_2}{\St_2}{\Typ}
    }{
      \PVGRHasExpType{\Ctx}{\St_1}{\Exp}{\Ctx_2}{\St_2}{\Typ}
    }
  \end{mathpar}
\end{lemma}
\begin{proof}
  This is a corollary of Lemma~\ref{lem:sub-pres-typing} due to the
  substitution typing
  $\PVGRHasSubstType{\SubId}{\Ctx,\Cstr}{\Ctx}$
  \qedhere
\end{proof}

\begin{lemma}[Evaluation Context Typings for Expressions]
  \label{lem:eval-ctx-typings}
  \
  \begin{enumerate}
  \item
    \label{lem:eval-ctx-typings:split}
    $
      \inferrule {
        \PVGRHasExpType{\Ctx_1}{\St_1}{\ECtx{\Exp}}{\Ctx_3}{\St_3}{\Typ}
      } {
        \exists \St_{11},\St_{12},\Ctx_{21},\Ctx_{22},\St_2,\Typ' \and
        \St_1 = \St_{11},\St_{12} \\
        \Ctx_3 = \Ctx_{21},\Ctx_{22} \\
        \PVGRHasExpType{\Ctx_1}{\St_{11}}{\Exp}{\Ctx_{21}}{\St_2}{\Typ'} \and
        \PVGRHasExpType{\Ctx_1 \DisjointAppend \Ctx_{21},\EVar:\Typ'}{\St_{12},\St_2}{\ECtx{\EVar}}{\Ctx_{22}}{\St_3}{\Typ}
      }
    $
    \vspace{5mm}
  \item
    \label{lem:eval-ctx-typings:merge}
    $
      \inferrule {
        \PVGRHasExpType{\Ctx_1}{\St_{11}}{\Exp}{\Ctx_{21}}{\St_2}{\Typ'} \\
        \PVGRHasExpType{\Ctx_1 \DisjointAppend \Ctx_{21},\EVar:\Typ'}{\St_{12},\St_2}{\ECtx{\EVar}}{\Ctx_{22}}{\St_3}{\Typ}
      } {
        \PVGRHasExpType{\Ctx_1}{\St_{11},\St_{12}}{\ECtx{\Exp}}{\Ctx_{21},\Ctx_{22}}{\St_3}{\Typ}
      }
    $
  \end{enumerate}
\end{lemma}
\begin{proof}
  \
  \begin{enumerate}

  \item By induction on the evaluation context $\ECtxSym$:
    \begin{itemize}
    \item \IndCase{$\Hole$}
      The assumption is
      \begin{align*}
        \PVGRHasExpType{\Ctx_1}{\St_1}{\Exp}{\Ctx_3}{\St_3}{\Typ}
      \end{align*}
      By choosing $\St_{11} = \St_1$, $\St_{12} = \Empty$, $\Ctx_{21} = \Ctx_3$, $\Ctx_{22} = \Empty$, $\St_2 = \St_3$, $\Typ' = \Typ$
      the goals become
      \begin{align*}
        \St_1 = \St_1 \tag{1}\\
        \Ctx_3 = \Ctx_3 \tag{2}\\
        \PVGRHasExpType{\Ctx_1}{\St_1}{\Exp}{\Ctx_3}{\St_3}{\Typ} \tag{3} \\
        \PVGRHasExpType{\Ctx_1 \DisjointAppend \Ctx_3,\EVar:\Typ}{\St_3}{\EVar}{\Empty}{\St_3}{\Typ} \tag{4}
      \end{align*}
      (1) and (2) are trivial, (3) follows by assumption, and (4) by \textsc{T-Var} and \textsc{T-Val}.
    \item \IndCase{$\ELet\EVar\ECtxSym\Exp$}
      The assumption is
      \begin{align*}
        \inferrule[T-Let]{
          \PVGRHasExpType
            {\Ctx_1}
            {\St_{11}}
            {\ECtx\Exp}
            {\Ctx_{21}}
            {\St_2}
            {\Typ'}
          \and
          \PVGRHasExpType
            {\Ctx_1 \DisjointAppend \Ctx_{21}, \EVar:\Typ'}
            {\St_{12},\St_2}
            {\Exp'}
            {\Ctx_{22}}
            {\St_3}
            {\Typ}
          \\
          \PVGRHasKind
            {\Ctx_1 \DisjointAppend \Ctx_{21}, \EVar:\Typ'}
            {\St_{12},\St_2}
            {\KSt}
        }{
          \PVGRHasExpType
            {\Ctx_1}
            {\St_{11},\St_{12}}
            {\ELet\EVar{\ECtx\Exp}{\Exp'}}
            {\Ctx_{21},\Ctx_{22}}
            {\St_3}
            {\Typ}
        }
      \end{align*}
      From the induction hypothesis follows
      \begin{align*}
        \inferrule*[Left=IH] {
          \PVGRHasExpType
            {\Ctx_1}
            {\St_{11}}
            {\ECtx\Exp}
            {\Ctx_{21}}
            {\St_2}
            {\Typ'}
        } {
          \exists \St_{111},\St_{112},\Ctx_{21},\Ctx_{22},\St_2',\Typ'' \and
          \St_{11} = \St_{111},\St_{112} \\
          \Ctx_{21} = \Ctx_{211},\Ctx_{212} \\
          \PVGRHasExpType{\Ctx_1}{\St_{111}}{\Exp}{\Ctx_{211}}{\St_2'}{\Typ''} \and
          \PVGRHasExpType{\Ctx_1 \DisjointAppend \Ctx_{211},\EVarY:\Typ''}{\St_{112},\St_2'}{\ECtx{\EVarY}}{\Ctx_{212}}{\St_2}{\Typ'}
        }
      \end{align*}
      By choosing $\St_{11} = \St_{111}$, $\St_{12} = \St_{112},\St_{12}$, $\Ctx_{21} = \Ctx_{211}$, $\Ctx_{22} = \Ctx_{212},\Ctx_{22}$, $\St_2 = \St_2'$, $\Typ' = \Typ''$
      the goals become
      \begin{align*}
        \St_{11},\St_{12} = \St_{111},\St_{112},\St_{12} \tag{1} \\
        \Ctx_{21},\Ctx_{22} = \Ctx_{211},\Ctx_{212},\Ctx_{22} \tag{2} \\
        \PVGRHasExpType{\Ctx_1}{\St_{111}}{\Exp}{\Ctx_{211}}{\St_2'}{\Typ''} \tag{3} \\
        \PVGRHasExpType{\Ctx_1 \DisjointAppend \Ctx_{211},\EVarY:\Typ'}{\St_{112},\St_{12},\St_2'}{\ELet\EVar{\ECtx\Exp}{\EVarY}}{\Ctx_{212},\Ctx_{22}}{\St_3}{\Typ} \tag{4}
      \end{align*}
      (1), (2) and (3) are direct consequences of the IH; (4) follows via
      \begin{align*}
        \inferrule*[Left=T-Let]{
          \PVGRHasExpType
            {\Ctx_1 \DisjointAppend \Ctx_{211},\EVarY:\Typ''}
            {\St_{112},\St_2'}
            {\ECtx{\EVarY}}
            {\Ctx_{212}}
            {\St_2}
            {\Typ'}
          \and
          \PVGRHasExpType
            {\Ctx_1 \DisjointAppend \Ctx_{21},\EVarY:\Typ'', \EVar:\Typ'}
            {\St_{12},\St_2}
            {\Exp'}
            {\Ctx_{22}}
            {\St_3}
            {\Typ}
          \\
          \PVGRHasKind
            {\Ctx_1 \DisjointAppend \Ctx_{21},\EVarY:\Typ'', \EVar:\Typ'}
            {\St_{12},\St_2}
            {\KSt}
        }{
          \PVGRHasExpType
            {\Ctx_1 \DisjointAppend \Ctx_{211},\EVarY:\Typ''}
            {\St_{112},\St_{12},\St_2'}
            {\ELet\EVar{\ECtx\Exp}{\EVarY}}
            {\Ctx_{212},\Ctx_{22}}
            {\St_3}
            {\Typ}
        }
      \end{align*}
      where the typing of $\Exp'$ and the kinding of $\St_{12},\St_2$ follow by weakening for $\EVarY:\Typ''$ via Lemma~\ref{lem:weakening}.
    \end{itemize}

  \item Analogous to (1).
    \qedhere
  \end{enumerate}
\end{proof}

\begin{lemma}[Evaluation Context Typings for Configurations]
  \label{lem:eval-ctx-typings-cfg}
  \begin{align*}
    \inferrule{
      \PVGRHasConfType{\Ctx}{\St}{\CCtx{\Cfg}}
    }{
      \exists\Ctx',\St' \and
      \PVGRHasConfType{\Ctx'}{\St'}{\Cfg} \and
      \forall \Cfg'.\
        (\PVGRHasConfType{\Ctx'}{\St'}{\Cfg'}) \Rightarrow
        (\PVGRHasConfType{\Ctx}{\St}{\CCtx{\Cfg'}})
    }
  \end{align*}
\end{lemma}
\begin{proof}
  By induction on the evaluation context:
  \begin{itemize}
  \item \IndCase{$\CCtxSym=\Hole$}
    We choose $\Ctx'=\Ctx$ and $\St'=\St$, which reduces our goals
    to assumptions and tautologies.
  \item \IndCase{$\CCtxSym=\CBindChan\TVar{\TVar'}\Ses\CCtxSym$}
    Here the assumption has the form
    \begin{align*}
      \inferrule*[Left=T-NuChan]{
        \TVar,\TVar' \text{ not free in } \Ctx \\
        \PVGRHasKind
          {\Ctx}
          {\Ses}
          {\KSession} \\
        \PVGRHasConfType
          {\Ctx_\TVar}
          {\St_\TVar}
          {\CCtx{\Cfg}}
      }{
        \CBindChan\TVar{\TVar'}\Ses{\CCtx{\Cfg}}
      }
    \end{align*}
    where $\Ctx_\TVar = \Ctx \DisjointAppend \TVar:\KDomain\TShapeOne \DisjointAppend \TVar':\KDomain\TShapeOne$
    and $\St_\TVar = \St,\StBind\TVar\Ses,\StBind{\TVar'}{\Dual\Ses}$.

    From the induction hypothesis follows
    \begin{mathpar}
      \exists\Ctx',\St'\and
      \PVGRHasConfType{\Ctx'}{\St'}{\Cfg} \and
      \forall \Cfg'.\
        (\PVGRHasConfType{\Ctx'}{\St'}{\Cfg'}) \Rightarrow
        (\PVGRHasConfType{\Ctx_\TVar}{\St_\TVar}{\CCtx{\Cfg'}})
    \end{mathpar}
    For the goal we choose the same $\Ctx'$ and $\St'$.
    Let $\Cfg'$ be some configuration such that $\PVGRHasConfType{\Ctx'}{\St'}{\Cfg'}$.
    From the result of the induction hypothesis follows
    \begin{mathpar}
      \PVGRHasConfType{\Ctx_\TVar}{\St_\TVar}{\CCtx{\Cfg'}}
    \end{mathpar}
    which allows us to reconstruct the \textsc{T-NuChan} rule.
  \item \IndCase{$\CCtxSym=\CBindAP\EVar\Ses\CCtxSym$}
    Similar as the previous case.
  \item \IndCase{$\CCtxSym=\CPar\CCtxSym\Cfg$}
    Similar as the previous case.
    \qedhere
  \end{itemize}
\end{proof}

\begin{lemma}[Wellformed Inputs imply Wellformed Outputs]
  \label{lem:sanity}
  \
  \begin{enumerate}
  \item
    $
    \inferrule{
      \PVGRIsCtx\Ctx \\
      \PVGRHasKind\Ctx\Typ\Kind
    }{
      \PVGRIsKind\Ctx\Kind
    }
    $
  \item
    $
    \inferrule{
    \PVGRIsCtx\Ctx \\
    \PVGRHasValType\Ctx\Val\Typ
    }{
    \PVGRHasKind\Ctx\Typ\KType
    }
    $
  \item
    $
    \inferrule{
      \PVGRIsCtx\Ctx \\
      \PVGRHasKind\Ctx{\St_1}\KSt \\
      \PVGRHasExpType\Ctx{\St_1}\Exp{\Ctx'}{\St_2}{\Typ_2}
    }{
      \PVGRIsCtx{\Ctx \DisjointAppend \Ctx'} \\
      \PVGRHasKind{\Ctx \DisjointAppend \Ctx'}{\St_2}\KSt \\
      \PVGRHasKind{\Ctx \DisjointAppend \Ctx'}{\Typ_2}\KType
    }
    $
  \end{enumerate}
\end{lemma}
\begin{proof}
  \
  \begin{enumerate}
  \item By induction on the kinding derivation:
    \begin{itemize}
    \item \IndCase{K-Var} Follows from the context formation due to \textsc{CF-ConsKind}.
    \item \IndCase{K-App} The induction hypothesis for
      $\PVGRHasKind{\Ctx}{\Typ_1}{\Kind_1 \to \Kind_2}$
      yields
      $\PVGRIsKind{\Ctx}{\Kind_1 \to \Kind_2}$
      which by case-analysis yields
      $\PVGRIsKind{\Ctx}{\Kind_2}$.
    \item \IndCase{K-Lam}
      From $\PVGRHasKind{\Ctx}{\POLYShape}{\KShape}$ follows $\PVGRIsKind{\Ctx}{\KDomain\POLYShape}$.
      From $\Kind \in \{\KType, \KSt\}$ follows $\PVGRIsKind{\Ctx}{\Kind}$.
      Via $\textsc{KF-Arr}$ follows $\PVGRIsKind{\Ctx}{\KDomain\POLYShape \to \Kind}$.
    \item \IndCase{K-DomMerge}
      Applying the induction hypothesis to $\PVGRHasKind{\Ctx}{\POLYDom_i}{\KDomain{\POLYShape_i}}$ yields
      $\PVGRIsKind{\Ctx}{\KDomain{\POLYShape_i}}$, which by case analysis on
      \textsc{KF-Dom} yields $\PVGRHasKind{\Ctx}{\POLYShape_i}{\KShape}$.
      The result then follows from the following proof tree:
      \begin{align*}
        \inferrule*[Right=KF-Dom]{
          \inferrule*[Right=K-ShapePair]{
            \PVGRHasKind{\Ctx}{\POLYShape_1}{\KShape} \\
            \PVGRHasKind{\Ctx}{\POLYShape_2}{\KShape}
          }{
            \PVGRHasKind{\Ctx}{\TShapePair{\POLYShape_1}{\POLYShape_2}}{\KShape}
          }
        }{
          \PVGRIsKind{\Ctx}{\KDomain{\TShapePair{\POLYShape_1}{\POLYShape_2}}}
        }
      \end{align*}
    \item \IndCase{K-DomProj}
      By induction hypothesis we know $\PVGRIsKind{\Ctx}{\KDomain{\TShapePair{\POLYShape_1}{\POLYShape_2}}}$,
      which by case analysis gives us
      $\PVGRHasKind{\Ctx}{\TShapePair{\POLYShape_1}{\POLYShape_2}}{\KShape}$,
      which by further case analysis gives us
      $\PVGRHasKind{\Ctx}{\POLYShape_1}{\KShape}$ and $\PVGRHasKind{\Ctx}{\POLYShape_2}{\KShape}$,
      which via $\textsc{KF-Dom}$ gives us $\PVGRIsKind{\Ctx}{\KDomain{\POLYShape_\Label}}$.
    \item All other cases follow immediately via \textsc{KF-Type},
      \textsc{KF-Session}, \textsc{KF-State}, or \textsc{KF-Shape}.
    \end{itemize}

  \item
    By induction on the value typing derivation:
    \begin{itemize}
    \item \IndCase{T-Var} Follows from the context formation due to \textsc{CF-ConsType}.
    \item \IndCase{T-Unit} Follows directly from \textsc{K-Unit}.
    \item \IndCase{T-Pair} The induction hypothesis yields
      $\PVGRHasKind{\Ctx}{\Typ_1}{\KType}$ and
      $\PVGRHasKind{\Ctx}{\Typ_2}{\KType}$, which via \textsc{K-Pair} yields
      $\PVGRHasKind{\Ctx}{\TPair{\Typ_1}{\Typ_2}}{\KType}$.
    \item \IndCase{T-Abs, T-TAbs} Follows directly from the first assumption of their case's rule.
    \item \IndCase{T-Chan} Follows directly from the assumption via \textsc{K-Chan}.
    \end{itemize}

  \item
    By induction on the expression typing derivation:
    \begin{itemize}
    \item \IndCase{T-Val} Follows from (2) and the assumption $\PVGRHasKind{\Ctx_1}{\St_1}{\KSt}$.
    \item \IndCase{T-Let}
      From the assumption $\PVGRHasKind{\Ctx_1}{\St_1,\St_2}{\KSt}$ follows by
      case analysis $\PVGRHasKind{\Ctx_1}{\St_1}{\KSt}$ and $\PVGRHasKind{\Ctx_1}{\St_2}{\KSt}$.

      We then apply the induction hypothesis on the typing of $\Exp_1$:
      \begin{align*}
        \inferrule*[Right=IH]{
          \PVGRIsCtx{\Ctx_1} \\
          \PVGRHasKind{\Ctx_1}{\St_1}{\KSt} \\
          \PVGRHasExpType{\Ctx_1}{\St_1}{\Exp_1}{\Ctx_2}{\St_2'}{\Typ_1}
        }{
          \PVGRIsCtx{\Ctx_1 \DisjointAppend \Ctx_2} \\
          \PVGRHasKind{\Ctx_1 \DisjointAppend \Ctx_2}{\St_2'}{\KSt} \\
          \PVGRHasKind{\Ctx_1 \DisjointAppend \Ctx_2}{\Typ_1}{\KType}
        }
      \end{align*}
      and extend the context formation as follows:
      \begin{align*}
        \inferrule*[Right=CF-ConsType]{
          \PVGRIsCtx{\Ctx_1 \DisjointAppend \Ctx_2} \\
          \PVGRHasKind{\Ctx_1 \DisjointAppend \Ctx_2}{\Typ_1}{\KType}
        } {
          \PVGRIsCtx{\Ctx_1 \DisjointAppend \Ctx_2,\EVar:\Typ_1}
        }
      \end{align*}
      We then apply the induction hypothesis on the typing of $\Exp_2$:
      \begin{align*}
        \inferrule*[Right=IH]{
          \PVGRIsCtx{\Ctx_1 \DisjointAppend \Ctx_2,\EVar:\Typ_1} \and
          \PVGRHasKind{\Ctx_1 \DisjointAppend \Ctx_2,\EVar:\Typ_1}{\St_2,\St_2'}{\KSt} \and
          \PVGRHasExpType{\Ctx_1 \DisjointAppend \Ctx_2,\EVar:\Typ_1}{\St_2,\St_2'}{\Exp_2}{\Ctx_3}{\St_3}{\Typ_2}
        }{
          \PVGRIsCtx{\Ctx_1 \DisjointAppend \Ctx_2,\EVar:\Typ_1 \DisjointAppend \Ctx_3} \\
          \PVGRHasKind{\Ctx_1 \DisjointAppend \Ctx_2,\EVar:\Typ_1 \DisjointAppend \Ctx_3}{\St_3}{\KSt} \\
          \PVGRHasKind{\Ctx_1 \DisjointAppend \Ctx_2,\EVar:\Typ_1 \DisjointAppend \Ctx_3}{\Typ_2}{\KType}
        }
      \end{align*}
      As value-level bindings do neither affect context formation nor
      kinding relations, we can safely remove the $\EVar : \Typ_1$ binding
      from the conclusion and obtain
      \begin{mathpar}
          \PVGRIsCtx{\Ctx_1 \DisjointAppend \Ctx_2 \DisjointAppend \Ctx_3} \and
          \PVGRHasKind{\Ctx_1 \DisjointAppend \Ctx_2 \DisjointAppend \Ctx_3}{\St_3}{\KSt} \and
          \PVGRHasKind{\Ctx_1 \DisjointAppend \Ctx_2 \DisjointAppend \Ctx_3}{\Typ_2}{\KType}
      \end{mathpar}

      Since $(\Ctx_1 \DisjointAppend \Ctx_2) \DisjointAppend \Ctx_3$
      is equivalent to $\Ctx_1 \DisjointAppend (\Ctx_2, \Ctx_3)$ up
      to the order of the constraints (which is irrelevant) the result follows.
    \item \IndCase{T-Proj}
      The first two results follow trivially.
      The third result follows from the induction hypothesis and subsequent case analysis.
    \item \IndCase{T-App}
      Applying the induction hypothesis on $\Val_1$ yields
      $\PVGRHasKind{\Ctx_1}{\TArr{\St_1}{\Typ_1}{\Ctx_2}{\St_2}{\Typ_2}}{\KType}$,
      which by inversion yields the results.
    \item \IndCase{T-TApp}
      From the induction hypothesis on the typing of $\Val$ and subsequent case analysis follows
      $\PVGRHasKind{\Ctx,\TVar:\Kind,\Cstr}{\Typ}{\KType}$.
      By Lemma~\ref{lem:sub-pres-typing}.5 and Lemma~\ref{lem:cstr-removal-pres-typing}.3 then follows
      $\PVGRHasKind{\Ctx}{\Subst{\Typ'}{\TVar}{\Typ}}{\KType}$.
      By Lemma~\ref{lem:sub-pres-typing}.8 for
      $\PVGRTypeConv{\Subst{\Typ'}{\TVar}{\Typ}}{\Typ''}$ then follows our
      result $\PVGRHasKind{\Ctx}{\Typ''}{\KType}$.
    \item \IndCase{T-Send, T-Select}
      The first and third results follow trivially.
      Repeated case analysis on the kinding of the input state yields
      $\PVGRHasKind\Ctx\POLYDom{\KDomain\TShapeOne}$ and
      $\PVGRHasKind\Ctx\Ses\KSession$, which allows to construct the second result
      $\PVGRHasKind{\Ctx}{\StBind\POLYDom\Ses}{\KSt}$ via \textsc{K-StChan}.
    \item \IndCase{T-Recv}
      Repeated case analysis on the kinding of the input state yields
      \begin{mathpar}
        \PVGRHasKind\Ctx\POLYDom{\KDomain\TShapeOne} \and
        \PVGRHasKind\Ctx\Ses\KSession \and
        \PVGRHasKind\Ctx\POLYShape\KShape \and
        \PVGRHasKind{\CtxRestrictNonDom\Ctx,\TVar' : \KDomain\POLYShape}{\St'}{\KSt} \and
        \PVGRHasKind{\CtxRestrictNonDom\Ctx,\TVar' : \KDomain\POLYShape}{\Typ'}{\KType}
      \end{mathpar}

      The first result
      $\PVGRIsCtx{\Ctx\DisjointAppend\TVar' : \KDomain\POLYShape}$
      follows via \textsc{CF-ConsKind} and \textsc{CF-ConsCstr}.

      Applying Lemma~\ref{lem:weakening} (Weakening) on the kindings of $\St'$ and $\Typ'$
      yields
      \begin{mathpar}
        \PVGRHasKind{\Ctx,\TVar' : \KDomain\POLYShape}{\St'}{\KSt} \and
        \PVGRHasKind{\Ctx,\TVar' : \KDomain\POLYShape}{\Typ'}{\KType}
      \end{mathpar}
      Further weakening yields
      \begin{mathpar}
        \PVGRHasKind{\Ctx\DisjointAppend\TVar' : \KDomain\POLYShape}{\St'}{\KSt} \and
        \PVGRHasKind{\Ctx\DisjointAppend\TVar' : \KDomain\POLYShape}{\Typ'}{\KType}
      \end{mathpar}
      from which the second result
      $\PVGRHasKind{\Ctx\DisjointAppend \TVar':\KDomain\POLYShape}{\St',\StBind\POLYDom\Ses}{\KSt}$
      and third result
      $\PVGRHasKind{\Ctx\DisjointAppend \TVar':\KDomain\POLYShape}{\Typ'}{\KType}$
      can be constructed.
    \item \IndCase{T-Case}
      From $\PVGRHasKind\Ctx{\St,\StBind\TVar{\SBranch{\Ses_1}{\Ses_2}}}\KSt$
      follows $\PVGRHasKind\Ctx{\St,\StBind\TVar{\Ses_1}}\KSt$ via
      case analysis and kinding rules.
      The results then follow by the induction hypothesis on $\Exp_1$.
    \item \IndCase{T-Fork, T-Close}
      Follows immediately from \textsc{K-StEmpty} and \textsc{K-Unit}.
      \qedhere
    \end{itemize}
  \end{enumerate}
\end{proof}

\newpage

\SubjectCongruenceLemma{lem:subject-congruence-appendix}
\begin{proof}
  By induction on the congruence derivation:
  \begin{itemize}
  \item \IndCase{CC-Lift}
    Follows from the induction hypothesis in combination with
    Lemma~\ref{lem:eval-ctx-typings-cfg} to reach into the evaluation
    context.
  \item All other cases are straightforward by reordering the
    derivation trees of the configuration typings.
  \end{itemize}
\end{proof}

\SubjectReductionLemma{lem:subject-reduction-appendix}
\begin{proof}
  \
  \begin{enumerate}
  \item
    By induction on the $\Exp \ReducesToE \Exp'$ derivation:
    \begin{itemize}

    \item \IndCase{ER-BetaFun}
      The assumptions have the following structure:
      \begin{mathpar}
        \scriptsize
        \inferrule*[left=ER-BetaFun]{ }{
          (\EAbs{\St_1}{\EVar}{\Typ}{\Exp_1})~\Val_2 \ReducesToE \Subst{\Val_2}{\EVar} \Exp_1
        }
        \and
        \inferrule*[Left=T-App]{
          \inferrule*[Left=T-Abs]{
            \PVGRHasKind
              {\Ctx_1}
              {\TArr {\St_1} {\Typ_1} {\Ctx_2} {\St_2} {\Typ_2}}
              {\KType}
            \\
            \PVGRHasExpType
              {\Ctx_1, \EVar : \Typ_1}
              {\St_1}
              {\Exp_1}
              {\Ctx_2}
              {\St_2}
              {\Typ_2}
          }{
            \PVGRHasValType
              {\Ctx_1}
              {\EAbs{\St_1}\EVar{\Typ_1}{\Exp_1}}
              {\TArr {\St_1} {\Typ_1} {\Ctx_2} {\St_2} {\Typ_2}}
          }
          \\
          \PVGRHasValType
            {\Ctx_1}
            {\Val_2}
            {\Typ_1}
        }{
          \PVGRHasExpType
            {\Ctx_1}
            {\St_1}
            {(\EAbs{\St_1}\EVar{\Typ_1}{\Exp_1})~\Val_2}
            {\Ctx_2}
            {\St_2}
            {\Typ_2}
        }
      \end{mathpar}
      The result follows via Lemma~\ref{lem:sub-pres-typing}:
      \begin{align*}
        \scriptsize
        \inferrule*[Left=Lemma~\ref{lem:sub-pres-typing}]{
          \inferrule*[Left=Lemma~\ref{lem:weakening}]{
            \PVGRHasKind{\Ctx_1}{\St_1}{\KSt}
          } {
            \PVGRHasKind{\Ctx_1,\EVar:\Typ_1}{\St_1}{\KSt}
          } \and
          \PVGRHasExpType{\Ctx_1,\EVar:\Typ_1}{\St_1}{\Exp_1}{\Ctx_2}{\St_2}{\Typ_2} \and
          \PVGRHasSubstType{\Subst{\Val_2}{\EVar}}{\Ctx_1,\EVar:\Typ_1}{\Ctx_1}
        }{
          \PVGRHasExpType{\Ctx_1}{\St_1}{\Subst{\Val_2}{\EVar}{\Exp_1}}{\Ctx_2}{\St_2}{\Typ_2}
        }
      \end{align*}

    \item \IndCase{ER-BetaAll}
      The assumptions have the following structure:
      \begin{mathpar}
        \scriptsize
        \inferrule*[left=ER-BetaAll]{ }{
          \ETApp{(\ETAbs{\TVar}{\Kind}{\Cstr}{\Val})}{\Typ'} \ReducesToE \Subst{\Typ'}{\TVar} \Val
        }
        \\
        \inferrule*[left=T-TApp]{
          \inferrule*[Left=T-TAbs]{
            \inferrule*[Left=T-KAll]{
              \PVGRIsCtx{\Ctx_1, \TVar : \Kind, \Cstr} \\
              \PVGRHasKind{\Ctx_1, \TVar : \Kind, \Cstr}{\Typ}{\KType}
            } {
              \PVGRHasKind
                {\Ctx_1}
                {\TAll {\TVar} {\Kind} {\Cstr} {\Typ}}
                {\KType}
            }
            \\
            \PVGRHasValType
              {\Ctx_1, \TVar : \Kind, \Cstr}
              {\Val}
              {\Typ}
          }{
            \PVGRHasValType
              {\Ctx_1}
              {\ETAbs {\TVar} {\Kind} {\Cstr} {\Val}}
              {\TAll {\TVar} {\Kind} {\Cstr} {\Typ}}
          }
          \and
          \PVGRHasKind
            {\Ctx_1}
            {\Typ'}
            {\Kind}
          \and
          \PVGRCstrEntail
            {\Ctx_1}
            {\Subst{\Typ'}\TVar{\Cstr}}
          \and
          \PVGREvalType
            {\Subst{\Typ'}\TVar\Typ}
            {\Typ''}
        }{
          \PVGRHasExpType
            {\Ctx_1}
            {\Empty}
            {\ETApp{(\ETAbs {\TVar} {\Kind} {\Cstr} {\Val})}{\Typ'}}
            {\Empty}
            {\Empty}
            {\Typ''}
        }
      \end{mathpar}
      The result follows via Lemma~\ref{lem:sub-pres-typing}~and~\ref{lem:cstr-removal-pres-typing}:
      \begin{align*}
        \scriptsize
        \inferrule*[Left=T-Val]{
          \inferrule*[Left=Lemma~\ref{lem:cstr-removal-pres-typing}]{
            \PVGRCstrEntail{\Ctx_1}{\Subst{\Typ'}{\TVar}\Cstr} \and
            \inferrule*[left=Lemma~\ref{lem:sub-pres-typing}]{
              \PVGRHasValType{\Ctx_1,\TVar:\Kind,\Cstr}{\Val}{\Typ} \and
              \PVGRHasSubstType
                {\Subst{\Typ'}{\TVar}}
                {(\Ctx_1,\TVar:\Kind,\Cstr)}
                {(\Ctx_2,\Subst{\Typ'}{\TVar}\Cstr)}
            }{
              \PVGRHasValType{\Ctx_1,\Subst{\Typ'}{\TVar}\Cstr}{\Subst{\Typ'}{\TVar}\Val}{\Subst{\Typ'}{\TVar}\Typ}
            }
          }{
            \PVGRHasValType{\Ctx_1}{\Subst{\Typ'}{\TVar}\Val}{\Subst{\Typ'}{\TVar}\Typ}
          }
        }{
          \PVGRHasExpType{\Ctx_1}{\Empty}{\Subst{\Typ'}{\TVar}\Val}{\Empty}{\Empty}{\Subst{\Typ'}{\TVar}\Typ}
        }
      \end{align*}

    \item \IndCase{ER-BetaLet}
      The assumptions have the following structure:
      \begin{align*}
        \scriptsize
        \rulePVGRExprRedBetaLet
        \and
        \inferrule*[left=T-Let]{
          \inferrule*[Left=T-Val]{
            \PVGRHasValType
              {\Ctx_1}
              {\Val_1}
              {\Typ_1}
          } {
            \PVGRHasExpType
              {\Ctx_1}
              {\Empty}
              {\Val_1}
              {\Empty}
              {\Empty}
              {\Typ_1}
          }
          \and
          \PVGRHasExpType
            {\Ctx_1, \EVar:\Typ_1}
            {\St_2}
            {\Exp_2}
            {\Ctx_3}
            {\St_3}
            {\Typ_2}
          \and
          \PVGRHasKind
            {\Ctx_1, \EVar:\Typ_1}
            {\St_2}
            {\KSt}
        }{
          \PVGRHasExpType
            {\Ctx_1}
            {\St_2}
            {\ELet\EVar{\Val_1}{\Exp_2}}
            {\Ctx_3}
            {\St_3}
            {\Typ_2}
        }
      \end{align*}
      The result follows via:
      \begin{align*}
        \scriptsize
        \inferrule*[Left=Lemma~\ref{lem:sub-pres-typing}]{
          \inferrule*[Left=Lemma~\ref{lem:weakening}]{
            \PVGRHasKind{\Ctx_1}{\St_2}{\KSt}
          } {
            \PVGRHasKind{\Ctx_1,\EVar:\Typ_1}{\St_2}{\KSt}
          } \and
          \PVGRHasExpType{\Ctx_1,\EVar:\Typ_1}{\St_2}{\Exp_2}{\Ctx_3}{\St_3}{\Typ_2} \and
          \PVGRHasSubstType{\Subst{\Val_1}{\EVar}}{\Ctx_1,\EVar:\Typ_1}{\Ctx_1}
        }{
          \PVGRHasExpType{\Ctx_1}{\St_2}{\Subst{\Val_1}{\EVar}\Exp_2}{\Ctx_3}{\St_3}{\Typ_2}
        }
      \end{align*}

    \item \IndCase{ER-BetaPair}
      The assumptions have the following structure:
      \begin{align*}
        \scriptsize
        \rulePVGRExprRedBetaPair
        \and
        \inferrule*[left=T-Proj]{
          \inferrule*[Left=T-Pair]{
            \PVGRHasValType
              {\Ctx}
              {\Val_1}
              {\Typ_1}
            \\
            \PVGRHasValType
              {\Ctx}
              {\Val_2}
              {\Typ_2}
          }{
            \PVGRHasValType
              {\Ctx}
              {\EPair {\Val_1} {\Val_2}}
              {\TPair {\Typ_1} {\Typ_2}}
          }
        }{
          \PVGRHasExpType
            {\Ctx}
            {\Empty}
            {\EProj \Label {\EPair {\Val_1} {\Val_2}}}
            {\Empty}
            {\Empty}
            {\Typ_\Label}
        }
      \end{align*}
      The result follows via
      \begin{align*}
        \scriptsize
        \inferrule*[Left=T-Val]{
          \PVGRHasValType
            {\Ctx}
            {\Val_\Label}
            {\Typ_\Label}
        } {
          \PVGRHasExpType
            {\Ctx}
            {\Empty}
            {\Val_\Label}
            {\Empty}
            {\Empty}
            {\Typ_\Label}
        }
      \end{align*}

    \item \IndCase{ER-Lift}
      The assumptions have the following structure:
      \begin{align*}
        \scriptsize
        \rulePVGRExprRedLift
        \and
        \PVGRHasExpType{\Ctx_1}{\St_1}{\ECtx{\Exp_1}}{\Ctx_3}{\St_3}{\Typ}
      \end{align*}
      We first extract the typing of $\Exp_1$ from the evaluation context:
      \begin{align*}
        \scriptsize
        \inferrule*[left=Lemma~\ref{lem:eval-ctx-typings}.\ref{lem:eval-ctx-typings:split}]{
          \PVGRHasExpType{\Ctx_1}{\St_1}{\ECtx{\Exp_1}}{\Ctx_3}{\St_3}{\Typ}
        } {
          \exists \St_{11},\St_{12},\Ctx_{21},\Ctx_{22},\St_2,\Typ' \and
          \St_1 = \St_{11},\St_{12} \\
          \Ctx_3 = \Ctx_{21},\Ctx_{22} \\
          \PVGRHasExpType{\Ctx_1}{\St_{11}}{\Exp_1}{\Ctx_{21}}{\St_2}{\Typ'} \and
          \PVGRHasExpType{\Ctx_1\DisjointAppend \Ctx_{21},\EVar:\Typ'}{\St_{12},\St_2}{\ECtx{\EVar}}{\Ctx_{22}}{\St_3}{\Typ}
        }
      \end{align*}
      From $\St_1 = \St_{11},\St_{12}$ and $\PVGRHasKind{\Ctx_1}{\St_1}{\KSt}$ follows via inversion
      \begin{mathpar}
        \PVGRHasKind{\Ctx_1}{\St_{11}}{\KSt} \and
        \PVGRHasKind{\Ctx_1}{\St_{12}}{\KSt}
      \end{mathpar}
      Then we apply the induction hypothesis:
      \begin{align*}
        \scriptsize
        \inferrule*[Left=IH]{
          \PVGRIsCtx{\Ctx_1} \and
          \PVGRHasKind{\Ctx_1}{\St_{11}}{\KSt} \and
          \PVGRHasExpType{\Ctx_1}{\St_{11}}{\Exp_1}{\Ctx_{21}}{\St_2}{\Typ'} \and
          \Exp_1 \ReducesToE \Exp_2
        } {
          \PVGRHasExpType{\Ctx_1}{\St_{11}}{\Exp_2}{\Ctx_{21}}{\St_2}{\Typ'}
        }
      \end{align*}
      Then we plug the typing of $\Exp_2$ back into the evaluation context:
      \begin{align*}
        \scriptsize
        \inferrule*[Left=Lemma~\ref{lem:eval-ctx-typings}.\ref{lem:eval-ctx-typings:merge}]{
          \PVGRHasExpType{\Ctx_1}{\St_{11}}{\Exp_2}{\Ctx_{21}}{\St_2}{\Typ'} \and
          \PVGRHasExpType{\Ctx_1\DisjointAppend \Ctx_{21},\EVar:\Typ'}{\St_{12},\St_2}{\ECtx{\EVar}}{\Ctx_{22}}{\St_3}{\Typ}
        }{
          \PVGRHasExpType{\Ctx_1}{\St_1}{\ECtx{\Exp_2}}{\Ctx_3}{\St_3}{\Typ_2}
        }
      \end{align*}
    \end{itemize}

  \item
    By induction on the $\Cfg \ReducesToC \Cfg'$ derivation.
    For the sake of readability, we apply
    Lemma~\ref{lem:eval-ctx-typings} and
    \ref{lem:eval-ctx-typings-cfg} informally to talk about the typings
    inside evaluation contexts.
    \begin{itemize}

    \item \IndCase{CR-Expr}
      Immediate from (1).

    \item \IndCase{CR-Fork}
      The assumptions have the following structure:
      \begin{mathpar}
        \scriptsize
        \rulePVGRCfgRedFork
        \and
        \inferrule*[Right=Lemma~\ref{lem:eval-ctx-typings-cfg}]{
          \inferrule*[Right=T-Exp]{
            \inferrule*[Right=Lemma~\ref{lem:eval-ctx-typings}]{
              \inferrule*[Right=T-Fork]{
                \PVGRHasValType
                  {\Ctx}
                  {\Val}
                  {\TArr{\St_1}{\TUnit}{\cdot}{\cdot}{\TUnit}}
              } {
                \PVGRHasExpType
                  {\Ctx}
                  {\St_1}
                  {\EFork\Val}
                  {\Empty}
                  {\Empty}
                  {\Typ}
              }
            } {
              \PVGRHasExpType
                {\Ctx}
                {\St_1,\St_2}
                {\ECtx{\EFork\Val}}
                {\Ctx'}
                {\Empty}
                {\Typ}
            }
          }{
            \PVGRHasConfType
              {\Ctx}
              {\St_1,\St_2}
              {\ECtx{\EFork\Val}}
          }
        }{
          \PVGRHasConfType
            {\Ctx_0}
            {\St_0}
            {\CCtx{\ECtx{\EFork\Val}}}
        }
      \end{mathpar}

      The result follows via
      \begin{align*}
        \scriptsize
        \inferrule*[Left=Lemma~\ref{lem:eval-ctx-typings-cfg}]{
          \inferrule*[Left=T-Par]{
            \inferrule*[Left=T-Exp]{
              \inferrule*[Left=T-App]{
                \PVGRHasValType
                  {\Ctx}
                  {\Val}
                  {\TArr{\St_1}{\TUnit}{\cdot}{\cdot}{\TUnit}}
                \\
                \inferrule*[Right=T-Unit]{ }{
                  \PVGRHasValType
                    {\Ctx}
                    {\EUnit}
                    {\TUnit}
                }
              }{
                \PVGRHasExpType
                  {\Ctx}
                  {\St_1}
                  {\EApp{\Val}{\EUnit}}
                  {\Empty}
                  {\Empty}
                  {\TUnit}
              }
            }{
              \PVGRHasConfType
                {\Ctx}
                {\St_1}
                {\EApp{\Val}{\EUnit}}
            }
            \and
            \inferrule*[Right=Lemma~\ref{lem:eval-ctx-typings}]{ }{
              \PVGRHasConfType
                {\Ctx}
                {\St_2}
                {\ECtx{\EVar}}
            }
          }{
            \PVGRHasConfType
              {\Ctx}
              {\St_1,\St_2}
              {\CPar
                {\EApp{\Val}{\EUnit}}
                {\ECtx{\EVar}}
              }
          }
        }{
          \PVGRHasConfType
            {\Ctx_0}
            {\St_0}
            {\CCtx{\CPar
              {\EApp{\Val}{\EUnit}}
              {\ECtx{\EVar}}
            }}
        }
      \end{align*}

    \item \IndCase{CR-New}
      The assumptions have the following structure:
      \begin{mathpar}
        \scriptsize
        \rulePVGRCfgRedNew
        \and
        \inferrule*[Left=Lemma~\ref{lem:eval-ctx-typings-cfg}]{
          \inferrule*[Left=T-Exp]{
            \inferrule*[Left=Lemma~\ref{lem:eval-ctx-typings}]{
              \inferrule*[Left=T-New]{
                \PVGRHasKind
                  {\Ctx}
                  {\Ses}
                  {\KSession}
              } {
                \PVGRHasExpType
                  {\Ctx}
                  {\St_1}
                  {\New\Ses}
                  {\Empty}
                  {\Empty}
                  {\TAccessPoint\Ses}
              }
            } {
              \PVGRHasExpType
                {\Ctx}
                {\St}
                {\ECtx{\ENew\Ses}}
                {\Ctx'}
                {\Empty}
                {\Typ}
            }
          }{
            \PVGRHasConfType
              {\Ctx}
              {\St}
              {\ECtx{\ENew\Ses}}
          }
        }{
          \PVGRHasConfType
            {\Ctx_0}
            {\St_0}
            {\CCtx{\ECtx{\ENew\Ses}}}
        }
      \end{mathpar}
      The result follows via
      \begin{align*}
        \scriptsize
        \inferrule*[Left=Lemma~\ref{lem:eval-ctx-typings-cfg}]{
          \inferrule*[Left=T-NuAccess]{
            \EVar \text{fresh} \\
            \PVGRHasKind
              {\Ctx}
              {\Ses}
              {\KSession} \\
            \inferrule*[Right=T-Exp]{
              \inferrule*[Right=Lemma~\ref{lem:eval-ctx-typings}]{
                \inferrule*[Right=T-Val]{
                  \inferrule*[Right=T-Var]{ } {
                    \PVGRHasValType
                      {\Ctx, \EVar:\TAccessPoint\Ses}
                      {\EVar}
                      {\TAccessPoint\Ses}
                  }
                }{
                  \PVGRHasExpType
                    {\Ctx, \EVar:\TAccessPoint\Ses}
                    {\Empty}
                    {\EVar}
                    {\Empty}
                    {\Empty}
                    {\TAccessPoint\Ses}
                }
              }{
                \PVGRHasExpType
                  {\Ctx, \EVar:\TAccessPoint\Ses}
                  {\St}
                  {\ECtx{\EVar}}
                  {\Ctx'}
                  {\Empty}
                  {\Typ}
              }
            }{
              \PVGRHasConfType
                {\Ctx, \EVar:\TAccessPoint\Ses}
                {\St}
                {\ECtx{\EVar}}
            }
          }{
            \PVGRHasConfType
              {\Ctx}
              {\St}
              {\CBindAP\EVar\Ses{\ECtx{\EVar}}}
          }
        }{
          \PVGRHasConfType
            {\Ctx_0}
            {\St_0}
            {\CCtx{\CBindAP\EVar\Ses{\ECtx{\EVar}}}}
        }
      \end{align*}

    \item \IndCase{CR-RequestAccept}
      The assumptions have the following structure:
      \begin{mathpar}
        \scriptsize
        \rulePVGRCfgRedRequestAccept
        \and
        \inferrule*[Right=Lemma~\ref{lem:eval-ctx-typings-cfg}]{
          \PVGRHasConfType{\Ctx}{\St}{\CBindAP\EVar\Ses\Cfg}
        }{
          \PVGRHasConfType{\Ctx_0}{\St_0}{\CCtx{\CBindAP\EVar\Ses\Cfg}}
        }
      \end{mathpar}
      Applying Lemma~\ref{lem:subject-congruence-appendix} to the configuration
      typing and the congruency yields a configuration typing, which by inversion has the the following structure, where $\St = \St_1,\St_2,\St_3$ are the channels used by $\ECtxSym_1$, $\ECtxSym_2$ and $\Cfg'$, respectively:
      \begin{mathpar}
        \scriptsize
        \inferrule*[Right=T-NuAccess]{
          \EVar \text{ not free in } \Ctx \\
          \PVGRHasKind
            {\Ctx}
            {\Ses}
            {\KSession} \\
          \inferrule*[Right=T-Par]{
            \inferrule*[Left=T-Par]{
              (1)
              \\
              (2)
            }{
              \PVGRHasConfType
                {\Ctx, \EVar:\TAccessPoint\Ses}
                {\St_1,\St_2}
                {(
                  \CPar{
                    \ECtxA{\ERequest\EVar}
                  }{
                    \ECtxB{\EAccept\EVar}
                  }
                )}
            }
            \\
            \PVGRHasConfType
              {\Ctx, \EVar:\TAccessPoint\Ses}
              {\St_3}
              {\Cfg'}
          }{
            \PVGRHasConfType
              {\Ctx, \EVar:\TAccessPoint\Ses}
              {\St_1,\St_2,\St_3}
              {(
                \CPar{
                  \CPar{
                    \ECtxA{\ERequest\EVar}
                  }{
                    \ECtxB{\EAccept\EVar}
                  }
                }{
                  \Cfg'
                }
              )}
          }
        }{
          \PVGRHasConfType
            {\Ctx}
            {\St_1,\St_2,\St_3}
            {
              \CBindAP\EVar\Ses{(
                \CPar{
                  \CPar{
                    \ECtxA{\ERequest\EVar}
                  }{
                    \ECtxB{\EAccept\EVar}
                  }
                }{
                  \Cfg'
                }
              )}
            }
        }
      \end{mathpar}
      where
      \begin{enumerate}
      \item[(1)] 
        \adjustbox{valign=t}{
        $
          \scriptsize
          \inferrule*[Right=T-Exp]{
            \inferrule*[Right=Lemma~\ref{lem:eval-ctx-typings}]{
              \inferrule*[Right=T-Request]{
                \PVGRHasValType{\Ctx, \EVar:\TAccessPoint\Ses}{\EVar}{\TAccessPoint\Ses}
              }{
                \PVGRHasExpType
                  {\Ctx, \EVar:\TAccessPoint\Ses}
                  {\Empty}
                  {\ERequest\EVar}
                  {\TVar:\KDomain\TShapeOne}
                  {\StBind\TVar\Ses}
                  {\TChan\TVar}
              }
            }{
              \PVGRHasExpType
                {\Ctx, \EVar:\TAccessPoint\Ses}
                {\St_1}
                {\ECtxA{\ERequest\EVar}}
                {\Ctx_1',\TVar:\KDomain\TShapeOne}
                {\Empty}
                {\Typ}
            }
          }{
            \PVGRHasConfType
              {\Ctx, \EVar:\TAccessPoint\Ses}
              {\St_1}
              {\ECtxA{\ERequest\EVar}}
          }
        $
        }
        \vspace{3mm}
      \item[(2)]
        \adjustbox{valign=t}{
        $
          \scriptsize
          \inferrule*[Right=T-Exp]{
            \inferrule*[Right=Lemma~\ref{lem:eval-ctx-typings}]{
              \inferrule*[Right=T-Accept]{
                \PVGRHasValType{\Ctx, \EVar:\TAccessPoint\Ses}{\EVar}{\TAccessPoint\Ses}
              }{
                \PVGRHasExpType
                  {\Ctx, \EVar:\TAccessPoint\Ses}
                  {\Empty}
                  {\EAccept\EVar}
                  {\TVar:\KDomain\TShapeOne}
                  {\StBind\TVar{\Dual\Ses}}
                  {\TChan\TVar}
              }
            }{
              \PVGRHasExpType
                {\Ctx, \EVar:\TAccessPoint\Ses}
                {\St_1}
                {\ECtxB{\EAccept\EVar}}
                {\Ctx_2',\TVar:\KDomain\TShapeOne}
                {\Empty}
                {\Typ}
            }
          }{
            \PVGRHasConfType
              {\Ctx, \EVar:\TAccessPoint\Ses}
              {\St_1}
              {\ECtxB{\EAccept\EVar}}
          }
        $
        }
      \end{enumerate}
      The result follows via
      \begin{align*}
        \scriptsize
        \inferrule*[Right=Lemma~\ref{lem:eval-ctx-typings-cfg}]{
          \inferrule*[Right=T-NuAccess]{
            \EVar \text{ not free in } \Ctx \and
            \PVGRHasKind
              {\Ctx}
              {\Ses}
              {\KSession} \and
            \inferrule*[Right=T-NuChan]{
              \TVar,\TVar' \text{ fresh} \\
              \PVGRHasKind
                {\Ctx}
                {\Ses}
                {\KSession} \\
              (3)
            }{
              \PVGRHasConfType
                {\Ctx,\EVar:\TAccessPoint\Ses}
                {\St_1,\St_2,\St_3}
                {
                  \CBindChan\TVar{\TVar'}\Ses{(
                    \CPar{
                      \CPar{
                        \ECtxA{\EChan\TVar}
                      }{
                        \ECtxB{\EChan{\TVar'}}
                      }
                    }{
                      \Cfg'
                    }
                  )}
                }
            }
          }{
            \PVGRHasConfType
              {\Ctx}
              {\St_1,\St_2,\St_3}
              {
                \CBindAP\EVar\Ses{
                  \CBindChan\TVar{\TVar'}\Ses{(
                    \CPar{
                      \CPar{
                        \ECtxA{\EChan\TVar}
                      }{
                        \ECtxB{\EChan{\TVar'}}
                      }
                    }{
                      \Cfg'
                    }
                  )}
                }
              }
          }
        }{
            \PVGRHasConfType
              {\Ctx_0}
              {\St_0}
              {\CCtx{
                \CBindAP\EVar\Ses{
                  \CBindChan\TVar{\TVar'}\Ses{(
                    \CPar{
                      \CPar{
                        \ECtxA{\EChan\TVar}
                      }{
                        \ECtxB{\EChan{\TVar'}}
                      }
                    }{
                      \Cfg'
                    }
                  )}
                }
              }}
        }
      \end{align*}
      where
        \vspace{3mm}
      \begin{enumerate}
      \item[(3)]
        \adjustbox{valign=t}{
        $
        \scriptsize
            \inferrule*[left=T-Par]{
              \inferrule*[Left=T-Par]{
                (4)
                \\
                (5)
              }{
                \PVGRHasConfType
                  {\Ctx'}
                  {\St_1,\St_2,\StBind\TVar\Ses,\StBind{\TVar'}{\Dual\Ses}}
                  {(
                    \CPar{
                      \ECtxA{\EChan\TVar}
                    }{
                      \ECtxB{\EChan{\TVar'}}
                    }
                  )}
              }
              \\
              \PVGRHasConfType
                {\Ctx'}
                {\St_3}
                {\Cfg'}
            }{
              \PVGRHasConfType
                {\Ctx'}
                {\St_1,\St_2,\St_3,\StBind\TVar\Ses,\StBind{\TVar'}{\Dual\Ses}}
                {(
                  \CPar{
                    \CPar{
                      \ECtxA{\EChan\TVar}
                    }{
                      \ECtxB{\EChan{\TVar'}}
                    }
                  }{
                    \Cfg'
                  }
                )}
            }
        $
        }
        \vspace{3mm}
        \\
        for $\Ctx' = \Ctx, \EVar:\TAccessPoint\Ses \DisjointAppend \TVar:\KDomain\TShapeOne \DisjointAppend \TVar':\KDomain\TShapeOne$
        \vspace{5mm}
      \item[(4)]
        \adjustbox{valign=t}{
        $
        \scriptsize
        \inferrule*[Right=T-Exp]{
          \inferrule*[Right=Lemma~\ref{lem:eval-ctx-typings}]{
            \inferrule*[Right=T-Val]{
              \inferrule*[Right=T-Chan]{
                \inferrule*[Right=T-TVar]{ }{
                  \PVGRHasKind{\Ctx'}\TVar{\KDomain\TShapeOne}
                }
              }{
                \PVGRHasValType
                  {\Ctx'}
                  {\EChan\TVar}
                  {\TChan\TVar}
              }
            }{
              \PVGRHasExpType
                {\Ctx'}
                {\Empty}
                {\EChan\TVar}
                {\Empty}
                {\Empty}
                {\TChan\TVar}
            }
          }{
            \PVGRHasExpType
              {\Ctx'}
              {\St_1,\StBind\TVar\Ses}
              {\ECtxA{\EChan\TVar}}
              {\Ctx_1'}
              {\Empty}
              {\Typ}
          }
        }{
          \PVGRHasConfType
            {\Ctx'}
            {\St_1,\StBind\TVar\Ses}
            {\ECtxA{\EChan\TVar}}
        }
        $
        }
        \vspace{3mm}
      \item[(5)]
        Similar to (4).
      \end{enumerate}
      Note that the channels, which in the pre-reduction tree are
      introduced existentially by the $\ERequest\cdot$ and $\EAccept\cdot$ operations,
      are in the post-reduction tree provided from the outside via the $\nu$-Binder.
      Lemma~\ref{lem:eval-ctx-typings} is strong enough to support this.

    \item \IndCase{CR-SendRecv}
      The assumptions have the following structure:
      \begin{mathpar}
        \scriptsize
        \inferrule*[left=CR-SendRecv]{
          \Cfg \Cong (
            \CPar{
            \CPar
              {\ECtxA{\ESend\Val{\EChan\TVar}}}
              {\ECtxB{\ERecv{\EChan\TVar'}}}}{\Cfg'}
          )
        }{
          \CBindChan {\TVar} {\TVar'}
          {\Ses'} {\Cfg}
          \ReducesToC
          \CBindChan {\TVar} {\TVar'} {\Ses} {(
            \CPar{
            \CPar
              {\ECtxA{\EUnit}}
              {\ECtxB{\Val}}}{\Cfg'}
          )}
        }
        \\
        \inferrule*[left=Lemma~\ref{lem:eval-ctx-typings-cfg}]{
          \inferrule*[Left=T-NuChan]{
            \TVar,\TVar' \text{ not free in } \Ctx \\
            \PVGRHasKind
              {\Ctx}
              {\Ses'}
              {\KSession} \\
            \PVGRHasConfType
              {\Ctx'}
              {\St,\StBind\TVar{\Ses'},\StBind{\TVar'}{\Dual{\Ses'}}}
              {\Cfg}
          }{
            \PVGRHasConfType
              {\Ctx}
              {\St}
              {\CBindChan\TVar{\TVar'}{\Ses'}{\Cfg}}
          }
        }{
          \PVGRHasConfType
            {\Ctx_0}
            {\St_0}
            {\CCtx{\CBindChan\TVar{\TVar'}{\Ses'}{\Cfg}}}
        }
      \end{mathpar}
      where $\Ctx' = \Ctx \DisjointAppend \TVar:\KDomain\TShapeOne \DisjointAppend \TVar':\KDomain\TShapeOne$ and $\Ses' = \SSend{\TVar''}{\KDomain\POLYShape}{\St'}{\Typ'}\Ses$.

      Applying Lemma~\ref{lem:subject-congruence-appendix} to the configuration
      typing of $\Cfg$ and the congruency yields a configuration
      typing, which by inversion reveals the following structure with
      $\St = \St_!,\St_1,\St_2,\St_3$, where $\St_!$ are the channels
      that are sent and received, and $\St_1$, $\St_2$, and $\St_3$ are the
      channels used in $\ECtxSym_1$, $\ECtxSym_2$, and $\Cfg'$, respectively:
      \begin{mathpar}
        \scriptsize
        \inferrule*[Right=T-Par]{
          \inferrule*[Right=T-Par]{
            (1)
            \\
            (2)
          }{
            \PVGRHasConfType
              {\Ctx'}
              {\St_!,\St_1,\St_2,\StBind\TVar{\Ses'},\StBind{\TVar'}{\Dual{\Ses'}}}
              {(
              \CPar
                {\ECtxA{\ESend{\Val}{\EChan\TVar}}}
                {\ECtxB{\ERecv{\EChan\TVar'}}}
              )}
          }
          \\
          \PVGRHasConfType
            {\Ctx'}
            {\St_3}
            {\Cfg'}
        }{
          \PVGRHasConfType
            {\Ctx'}
            {\St_!,\St_1,\St_2,\St_3,\StBind\TVar{\Ses'},\StBind{\TVar'}{\Dual{\Ses'}}}
            {(
              \CPar{
                \CPar
                  {\ECtxA{\ESend{\Val}{\EChan\TVar}}}
                  {\ECtxB{\ERecv{\EChan\TVar'}}}
              }{
                \Cfg'
              }
            )}
        }
      \end{mathpar}
      where
      \begin{enumerate}
      \item[(1)] 
        \adjustbox{valign=t}{
        $
          \scriptsize
          \inferrule*[Right=T-Exp]{
            \inferrule*[Right=Lemma~\ref{lem:eval-ctx-typings}]{
              \inferrule*[Right=T-Send]{
                \PVGRHasKind \Ctx {\POLYDom} {\KDomain\POLYShape} \and
                \PVGRTypeConv {\Subst{\POLYDom}{\TVar''}{\St'}} {\St_!} \and
                \PVGRTypeConv {\Subst{\POLYDom}{\TVar''}{\Typ'}} {\Typ''} \\
                \PVGRHasValType{\Ctx}{\Val}{\Typ''} \and
                \PVGRHasValType{\Ctx}{\EChan\TVar}{\TChan \TVar}
              }{
                \PVGRHasExpType
                  {\Ctx'}
                  {\St_!,\StBind\TVar{\Ses'}}
                  {\ESend{\Val}{\EChan\TVar}}
                  {\Empty}
                  {\StBind\TVar\Ses}
                  {\TUnit}
              }
            }{
              \PVGRHasExpType
                {\Ctx'}
                {\St_!,\St_1,\StBind\TVar{\Ses'}}
                {\ECtxA{\ESend{\Val}{\EChan\TVar}}}
                {\Ctx_1'}
                {\Empty}
                {\Typ_1}
            }
          }{
            \PVGRHasConfType
              {\Ctx'}
              {\St_!,\St_1,\StBind\TVar{\Ses'}}
              {\ECtxA{\ESend{\Val}{\EChan\TVar}}}
          }
        $
        }
        \vspace{3mm}
      \item[(2)]
        \adjustbox{valign=t}{
        $
          \scriptsize
          \inferrule*[Right=T-Exp]{
            \inferrule*[Right=Lemma~\ref{lem:eval-ctx-typings}]{
              \inferrule*[Right=T-Recv]{
                \PVGRHasKind \Ctx {\TVar'} {\KDomain\TShapeOne} \and
                \PVGRHasValType{\Ctx}{\EChan\TVar'}{\TChan{\TVar'}}
              }{
                \PVGRHasExpType
                  {\Ctx'}
                  {\StBind{\TVar'}{\Dual{\Ses'}}}
                  {\ERecv{\EChan\TVar'}}
                  {\TVar'':\KDomain\POLYShape}
                  {\St',\StBind{\TVar'}{\Dual\Ses}}
                  {\Typ'}
              }
            }{
              \PVGRHasExpType
                {\Ctx'}
                {\St_2,\StBind{\TVar'}{\Dual{\Ses'}}}
                {\ECtxB{\ERecv{\EChan\TVar'}}}
                {\Ctx_2',\TVar'':\KDomain\POLYShape}
                {\Empty}
                {\Typ_2}
            }
          }{
            \PVGRHasConfType
              {\Ctx'}
              {\St_2,\StBind{\TVar'}{\Dual{\Ses'}}}
              {\ECtxB{\ERecv{\EChan\TVar'}}}
          }
        $
        }
      \end{enumerate}

      The result follows via
      \begin{mathpar}
        \scriptsize
        \inferrule*[Right=Lemma~\ref{lem:eval-ctx-typings-cfg}]{
          \inferrule*[Right=T-NuChan]{
            \TVar,\TVar' \text{ not free in } \Ctx \\
            \PVGRHasKind
              {\Ctx}
              {\Ses}
              {\KSession} \\
            \inferrule*[Right=T-Par]{
              \inferrule*[Right=T-Par]{
                (3) \\ (4)
              }{
                \PVGRHasConfType
                  {\Ctx'}
                  {\St_!,\St_1,\St_2,\StBind\TVar{\Ses},\StBind{\TVar'}{\Dual{\Ses}}}
                  {
                    \CPar
                      {\ECtxA{\EUnit}}
                      {\ECtxB{\Val}}
                  }
              }
              \\
              \PVGRHasConfType
                {\Ctx'}
                {\St_3}
                {\Cfg'}
            }{
              \PVGRHasConfType
                {\Ctx'}
                {\St_!,\St_1,\St_2,\St_3,\StBind\TVar{\Ses},\StBind{\TVar'}{\Dual{\Ses}}}
                {
                  \CPar{
                    \CPar
                      {\ECtxA{\EUnit}}
                      {\ECtxB{\Val}}
                  }{
                    \Cfg'
                  }
                }
            }
          }{
            \PVGRHasConfType
              {\Ctx}
              {\St_!,\St_1,\St_2,\St_3}
              {\CBindChan {\TVar} {\TVar'} {\Ses} {(
                \CPar{
                  \CPar
                    {\ECtxA{\EUnit}}
                    {\ECtxB{\Val}}
                }{
                  \Cfg'
                }
              )}}
          }
        }{
          \PVGRHasConfType
            {\Ctx_0}
            {\St_0}
            {\CCtx{\CBindChan {\TVar} {\TVar'} {\Ses} {(
              \CPar{
                \CPar
                  {\ECtxA{\EUnit}}
                  {\ECtxB{\Val}}
              }{
                \Cfg'
              }
            )}}}
        }
      \end{mathpar}
      where
      \begin{enumerate}
      \item[(3)] 
        \adjustbox{valign=t}{
        $
          \scriptsize
          \inferrule*[Right=T-Exp]{
            \inferrule*[Right=Lemma~\ref{lem:eval-ctx-typings}]{
              \inferrule*[Right=T-Val]{
                \inferrule*[Right=T-Unit]{ }{
                  \PVGRHasValType
                    {\Ctx'}
                    {\EUnit}
                    {\TUnit}
                }
              }{
                \PVGRHasExpType
                  {\Ctx'}
                  {\Empty}
                  {\EUnit}
                  {\Empty}
                  {\Empty}
                  {\TUnit}
              }
            }{
              \PVGRHasExpType
                {\Ctx'}
                {\St_1,\StBind\TVar{\Ses}}
                {\ECtxA{\EUnit}}
                {\Ctx_1'}
                {\Empty}
                {\Typ_1}
            }
          }{
            \PVGRHasConfType
              {\Ctx'}
              {\St_1,\StBind\TVar{\Ses}}
              {\ECtxA{\EUnit}}
          }
        $
        }
        \vspace{3mm}
      \item[(4)] 
        \adjustbox{valign=t}{
        $
          \scriptsize
          \inferrule*[Right=T-Exp]{
            \inferrule*[Right=Lemma~\ref{lem:eval-ctx-typings}]{
              \inferrule*[Right=T-Val]{
                \PVGRHasValType
                  {\Ctx'}
                  {\Val}
                  {\Typ'}
              }{
                \PVGRHasExpType
                  {\Ctx'}
                  {\Empty}
                  {\Val}
                  {\Empty}
                  {\Empty}
                  {\Typ'}
              }
            }{
              \PVGRHasExpType
                {\Ctx'}
                {\St_!,\St_2,\StBind{\TVar'}{\Dual\Ses}}
                {\ECtxB{\Val}}
                {\Ctx_2'}
                {\Empty}
                {\Typ_2}
            }
          }{
            \PVGRHasConfType
              {\Ctx'}
              {\St_!,\St_2,\StBind{\TVar'}{\Dual\Ses}}
              {\ECtxB{\Val}}
          }
        $
        }
      \end{enumerate}

    \item \IndCase{CR-SelectCase, CR-Close}
      Similar to the previous cases.

    \end{itemize}
  \end{enumerate}
\end{proof}

\begin{lemma}[Context inversion]\label{lem:context-inversion}
  \begin{mathpar}
    \inferrule{
      \PVGRIsCtx{\Ctx} \\
      \IsOuter\Ctx \\
      \PVGRHasKind{\Ctx}{\St}{\KSt} \\
      \PVGRHasConfType{\Ctx}{\St}{\CCtx\Exp} \\
    }{
      \exists \Ctx',\St', T.\
      \PVGRIsCtx{\Ctx'} \and
      \IsOuter{\Ctx'} \and
      \PVGRHasKind{\Ctx'}{\St'}{\KSt} \and
      \PVGRHasExpType{\Ctx'}{\St'}{\Exp}{\Empty}{\Empty}{T}
    }
  \end{mathpar}
\end{lemma}
\begin{proof}
  By induction on $\PVGRHasConfType{\Ctx}{\St}{\CCtx\Exp}$.
\end{proof}

\begin{restatable}[Canonical forms]{lemma}{CanonicalForms}\label{lem:canonical-forms}
  Suppose that $\PVGRHasValType{\Ctx}\Val\Typ$ and $\IsOuter\Ctx$.
  \begin{itemize}
  \item If $\Typ$ is ${\TArr{\St}{\Typ}{\Ctx'}{\St'}{\Typ'}}$, then $\Val$ is $\EAbs{\St}\EVar{\Typ}\Exp$, for some $\Exp$.
  \item If $\Typ$ is $\TPair{\Typ_1}{\Typ_2}$, then $\Val$ is $\EPair {\Val_1} {\Val_2}$, for some $\Val_1$ and $\Val_2$.
  \item If $\Typ$ is ${\TAll\TVar\Kind{\Cstr}\Typ}$, then $\Val$ is $\ETAbs {\TVar} {\Kind} {\Cstr} {\Val_1}$, for some $\Val_1$.
  \item If $\Typ$ is $\TUnit$, then $\Val$ is $\EUnit$.
  \item If $\Typ$ is $\TChan\POLYDom$, then $\Val$ is
    $\EChan\POLYDom$, for some $\POLYDom$.
  \item If $\Typ$ is $\TAccessPoint\Ses$, then $\Val$ is $\EVar$, for
    some $\EVar \in \Dom{\Ctx}$.
  \end{itemize}
\end{restatable}
\begin{proof}
  By inversion of the value typing judgment $\PVGRHasValType{\Ctx}\Val\Typ$.
\end{proof}
\ProgressLemma*
\begin{proof}
  The proof is by induction on the expression $\Exp$.

  \textbf{Case }$\Val$: $\IsValue\Val$ holds.

  \textbf{Case }$\ELet\EVar{\Exp_1}{\Exp_2}$: By the IH for $\Exp_1$ we have three cases:
  \begin{itemize}
  \item if $\IsValue{\Exp_1}$, then $\ELet\EVar{\Exp_1}{\Exp_2} \ReducesToE \Subst{\Exp_1}{\EVar} \Exp_2$ by \TirName{ER-BetaLet};
  \item if $\IsComm{\Exp_1}$, then $\IsComm{(\ELet\EVar{\Exp_1}{\Exp_2})}$;
  \item if $\Exp_1 \ReducesToE \Exp_1'$, then the $\Terminal{let}$ reduces, too.
  \end{itemize}

  \textbf{Case }$\EApp{\Val_1}{\Val_2}$:
  Inversion of $\PVGRHasExpType{\Ctx}{\St}{\EApp{\Val_1}{\Val_2}}{\Ctx'}{\St'}{\Typ'}$ yields
  \begin{gather}
    \PVGRHasValType
    {\Ctx}
    {\Val_1}
    {\TArr{\St}{\Typ}{\Ctx'}{\St'}{\Typ'}}
    \\
    \PVGRHasValType
    {\Ctx}
    {\Val_2}
    {\Typ}
  \end{gather}
  By Lemma~\ref{lem:canonical-forms}, $\Val_1 = \EAbs{\St}\EVar{\Typ}{\Exp_1}$, for some $\Exp_1$. Hence, $\EApp{\Val_1}{\Val_2} \ReducesToE$  by \TirName{ER-BetaFun}.

  \textbf{Case }$\EProj \Label \Val$:
  Inversion of $\PVGRHasExpType{\Ctx}{\St}{\EProj \Label \Val}{\Ctx'}{\St'}{\Typ'}$ yields
  \begin{gather}
    \PVGRHasValType{\Ctx}{\Val}{\TPair{\Typ_1}{\Typ_2}}
  \end{gather}
  By Lemma~\ref{lem:canonical-forms}, $\Val = \EPair {\Val_1} {\Val_2}$, for some $\Val_1$ and $\Val_2$. Hence, $\EProj \Label \Val \ReducesToE $ by \TirName{ER-BetaPair}.

  \textbf{Case }$\ETApp\Val{\Typ''} $:
  Inversion of $\PVGRHasExpType{\Ctx}{\St}{\ETApp\Val{\Typ''}}{\Ctx'}{\St'}{\Typ'}$ yields
  \begin{gather}
    \PVGRHasValType
      {\Ctx}
      {\Val}
      {\TAll\TVar\Kind{\Cstr}\Typ}
    \\
    \PVGRHasKind
      {\Ctx}
      {\Typ''}
      {\Kind}
    \\
    \PVGRCstrEntail
      {\Ctx}
      {\Subst{\Typ''}\TVar{\Cstr}}
    \\
    \PVGRTypeConv
      {\Subst{\Typ''}\TVar\Typ}
      {\Typ'}
  \end{gather}
  By Lemma~\ref{lem:canonical-forms}, $\Val$ is $\ETAbs {\TVar} {\Kind} {\Cstr} {\Val_1}$, for some $\Val_1$. Hence, $\ETApp\Val{\Typ''} \ReducesToE $ by \TirName{ER-BetaAll}.

  \textbf{Case }$\EFork\Val $: Inversion of
  $\PVGRHasExpType{\Ctx}{\St}{\EFork\Val}{\Ctx'}{\St'}{\Typ'}$ yields
  $\Ctx' = \Empty$, $\St' = \Empty$, $\Typ' = \TUnit$, and
  \begin{gather}
    \PVGRHasValType
    {\Ctx}
    {\Val}
    {\TArr{\St}{\TUnit}{\cdot}{\cdot}{\TUnit}}
  \end{gather}
  By Lemma~\ref{lem:canonical-forms}, $\Val$ is
  $\EAbs{\St}\EVar{\TUnit}{\Exp_1}$, hence
  $\IsComm{(\EFork\Val)}$.

  \textbf{Case }$\ENew\Ses $: Inversion of
  $\PVGRHasExpType{\Ctx}{\St}{\ENew\Ses}{\Ctx'}{\St'}{\Typ'}$ yields
  $\St = \Empty$,  $\Ctx' = \Empty$, $\St' = \Empty$, and $\Typ' =
  \Ap\Ses$.
  Hence $\IsComm{(\ENew\Ses)}$.

  \textbf{Case }$\EAccept\Val $: Inversion of
  $\PVGRHasExpType{\Ctx}{\St}{\EAccept\Val}{\Ctx'}{\St'}{\Typ'}$
  yields
  \begin{gather}
    \PVGRHasValType{\Ctx}{\Val}{\TAccessPoint\Ses}
  \end{gather}
  By Lemma~\ref{lem:canonical-forms}, $\Val$ is
  $\EVar$, hence
  $\IsComm{(\EAccept\Val)}$.

  \textbf{Case }$\ERequest\Val $: by similar reasoning, $\Val = \EVar$
  and $\IsComm{(\ERequest\Val)}$.

  \textbf{Case }$\ESend{\Val_1}{\Val_2} $: Inversion of
  $\PVGRHasExpType{\Ctx}{\St}{\ESend{\Val_1}{\Val_2}}{\Ctx'}{\St'}{\Typ'}$
  yields
  \begin{itemize}
  \item $\St = {\St_1,
      \StBind\POLYDom{\SSend{\TVar'}{\KDomain\POLYShape}{\St'}{\Typ''}\Ses}}$,
  \item $\Ctx' = \Empty$,
  \item  $\St' = {\StBind\POLYDom\Ses}$,
  \item $\Typ' = {\TUnit}$, and
  \end{itemize}
  \begin{gather}
    \PVGRHasKind \Ctx {\POLYDom'} {\KDomain\POLYShape} \\
    \PVGRTypeConv {\Subst{\POLYDom'}{\TVar'}{\St'}} {\St_1} \\
    \PVGRTypeConv {\Subst{\POLYDom'}{\TVar'}{\Typ''}} \Typ \\
    \PVGRHasKind \Ctx {\POLYDom} {\KDomain\TShapeOne} \\
    \PVGRHasValType{\Ctx}{\Val_1}{\Typ} \\
    \PVGRHasValType{\Ctx}{\Val_2}{\TChan \POLYDom}
  \end{gather}
  By Lemma~\ref{lem:canonical-forms}, $\Val_2$ is
  $\EChan\POLYDom$, hence
  $\IsComm{(\ESend{\Val_1}{\Val_2})}$.

  \textbf{Case }$\ERecv\Val $: by similar reasoning as in the previous
  case, $\Val = \EChan\POLYDom$ and $\IsComm{(\ERecv\Val)}$.

  \textbf{Case }$\ESelect\Label\Val $: by similar reasoning, $\Val =
  \EChan\POLYDom$ and $\IsComm{(\ESelect\Label\Val)}$.

  \textbf{Case }$\ECase\Val{\Exp_1}{\Exp_2} $:  by similar reasoning, $\Val =
  \EChan\POLYDom$ and $\IsComm{(\ECase\Val{\Exp_1}{\Exp_2})}$.

  \textbf{Case }$\EClose\Val $:  by similar reasoning, $\Val =
  \EChan\POLYDom$ and $\IsComm{(\EClose\Val)}$.
\end{proof}

\ProgressConfigurations*
\begin{proof}
  Suppose that $\neg\IsFinal\Cfg$ and $\neg\IsDeadlock\Cfg$. Hence,
  one of the three items in Definition~\ref{def:is-deadlock} must be
  violated and we show that $\Cfg$ reduces in each case.

  \textbf{Suppose item~\ref{item:1} is violated.}
  Hence, there is some $\CCtxSym$ such that $\Cfg = \CCtx\Exp$ and
  $\neg\IsValue\Exp$ and $\neg\IsComm\Exp$ or $\Exp =
  \ECtx{\EFork\Val}$ or $\Exp = \ECtx{\ENew\Ses}$, for some
  $\ECtxSym$, $\Val$, $\Ses$. 

  If $\Exp = \ECtx{\EFork\Val}$, then $\Val =
  \EAbs{\St}\EVar{\TUnit}{\Exp_1}$ and $\Cfg$ reduces as follows
  \begin{gather*}
    \CCtx{\ECtx{\EFork{\EAbs{\St}\EVar{\TUnit}{\Exp_1}}}}
    \ReducesToC
    \CCtx{\CPar{\ECtx{\EUnit}}{\EApp{(\EAbs{\St}\EVar{\TUnit}{\Exp_1})}\EUnit}}
  \end{gather*}

  If $\Exp = \ECtx{\ENew\Ses}$, then $\Cfg$ reduces by \TirName{CR-New}.

  Otherwise, by Lemma~\ref{lem:progress-congruence} (which is
  applicable because of context inversion, Lemma~\ref{lem:context-inversion}), there exists some
  $\Exp'$ such that  $\Exp\ReducesToE\Exp'$. Hence, $\CCtx\Exp
  \ReducesToC \CCtx{\Exp'}$.

  \textbf{Suppose item~\ref{item:2} is violated.}
  That is, there are configuration and evaluation contexts $\CCtxSym$, $\CCtxSym_1$,
  $\CCtxSym_2$, $\ECtxSym_1$, and $\ECtxSym_2$ such that $\Cfg =
  \CCtx{\CBindAP\EVar\Ses\Cfg'}$ and
  $\Cfg' = \CCtxA{\ECtxA{\ERequest\EVar}}$ and
  $\Cfg' = \CCtxB{\ECtxB{\EAccept\EVar}}$.
  Exploiting congruence we can find a configuration context
  $\CCtxSym'$ and process $\Cfg''$ such that
  $\Cfg \Cong
  \CCtxSym'[\CBindAP\EVar\Ses{\CPar{\CPar{\ECtxA{\ERequest\EVar}}{\ECtxB{\EAccept\EVar}}}{\Cfg''}}]$,
  which reduces by \TirName{CR-RequestAccept}.

  \textbf{Suppose item~\ref{item:3} is violated.}
  Consider the case for $\ESend\_\_$ and $\ERecv$.
  That is, there are configuration and evaluation contexts $\CCtxSym$, $\CCtxSym_1$,
  $\CCtxSym_2$, $\ECtxSym_1$, and $\ECtxSym_2$ such that $\Cfg =
  \CCtx{\CBindChan{\TVar_1}{\TVar_2}\Ses \Cfg'}$ and
  $\Cfg' = \CCtxA{\ECtxA{\ESend\Val{\EChan{\TVar_\Label}}}}$ and
  $\Cfg' = \CCtxB{\ECtxB{\ERecv{\EChan{\TVar_{3-\Label}}}}}$.
  Exploiting congruence we can find a configuration context
  $\CCtxSym'$ and process $\Cfg''$ such that
  \begin{gather*}
    \Cfg \Cong
    \CCtxSym'[\CBindChan{\TVar_1}{\TVar_2}\Ses{\CPar{\CPar{\ECtxA{\ESend\Val{\EChan{\TVar_\Label}}}}{\ECtxB{\ERecv{\EChan{\TVar_{3-\Label}}}}}}{\Cfg''}}]   
  \end{gather*}
  which reduces by \TirName{CR-SendRecv}.

  The remaining cases are similar.

\end{proof}
\end{document}